\documentclass[11pt]{article}
\usepackage[a4paper,
            left=1.25in,
            right=1.25in,
            top=1.25in,
            bottom=1.25in,
            footskip=.5in]{geometry}

\usepackage{algorithm}
\usepackage{algpseudocode}
\usepackage{authblk}
\usepackage{graphicx}
\usepackage{dsfont}
\usepackage{enumerate}
\usepackage{mathtools}
\usepackage[colorlinks=true,linkcolor=blue,citecolor=blue,urlcolor=blue,plainpages=false,pdfpagelabels]{hyperref}
\usepackage{amsmath}
\usepackage{amssymb}
\usepackage{physics}
\usepackage{bbm}
\usepackage{nicefrac}
\usepackage{url}
\usepackage{complexity}
\usepackage{xcolor}
\usepackage{tikz}
\definecolor{figaccent}{RGB}{31,90,142}
\usepackage[normalem]{ulem}
\usepackage{amsthm}
\usepackage{thmtools}
\usepackage{appendix}
\usepackage[capitalize,noabbrev]{cleveref}
\usepackage{pdfpages}

\usepackage[backend=biber, style=alphabetic, backref=true, hyperref=true, maxbibnames=99]{biblatex}

\usepackage[babel,english=british]{csquotes}
\DefineBibliographyStrings{english}{%
    backrefpage  = {cited on p.}, 
    backrefpages = {cited on pp.} 
}
\newcommand{\stkout}[1]{\ifmmode\text{\sout{\ensuremath{#1}}}\else\sout{#1}\fi}
\newif\ifverbose{}
\verbosefalse{}

\newif\ifcomments{}
\commentstrue{}

\newtheorem{theorem}{Theorem}[section]

\newtheorem{question}[theorem]{Question}
\newtheorem{proposition}[theorem]{Proposition}

\newtheorem{corollary}[theorem]{Corollary}

\newtheorem{lemma}[theorem]{Lemma}
\newtheorem{definition}[theorem]{Definition}
\newtheorem{claim}[theorem]{Claim}

\theoremstyle{definition}
\newtheorem{example}[theorem]{Example}
\newtheorem{remark}[theorem]{Remark}

\numberwithin{equation}{section}

\makeatletter

\def\theHALG@line{\thealgorithm.\arabic{ALG@line}}
\makeatother

\DeclareMathOperator{\cut}{cut}

\DeclareMathOperator{\lc}{lc}
\DeclareMathOperator{\tw}{tw}
\DeclareMathOperator{\CC}{cc}

\DeclareMathOperator{\cw}{cw}
\DeclareMathOperator{\tcw}{tcw}
\DeclareMathOperator{\adh}{adh}
\DeclareMathOperator{\tor}{tor}
\DeclareMathOperator{\depth}{depth}

\begin{document}

\title{
\makebox[\textwidth][c]{
\begin{minipage}{\dimexpr\textwidth+0.5in\relax}
    \centering
    Parameterised graph theory for tensor networks: entanglement rerouting, structural simplification, and agnostic tomography
\end{minipage}
}
}

\author[1]{Matthias C. Caro\thanks{\href{mailto:matthias.caro@warwick.ac.uk}{matthias.caro@warwick.ac.uk}}}
\author[1]{Natalie McHugh\thanks{\href{mailto:natalie.mchugh@warwick.ac.uk}{natalie.mchugh@warwick.ac.uk}}}
\author[2]{Sergii Strelchuk\thanks{\href{mailto:sergii.strelchuk@cs.ox.ac.uk}{sergii.strelchuk@cs.ox.ac.uk}}}
\affil[1]{Department of Computer Science, University of Warwick, Coventry, UK}
\affil[2]{Department of Computer Science, University of Oxford, Oxford, UK}

\date{}
\setcounter{Maxaffil}{0}
\renewcommand\Affilfont{\itshape\small}

\maketitle

\begin{abstract}

\noindent Parameterised graph theory studies how the complexity of graph-theoretic problems depends on structural parameters of the input graph. This perspective has proved useful in analysing tensor-network simulation~\cite{Markov_2008}. Its implications for tensor-network representations and tomography are less well understood. In particular, which graph parameters determine whether a tensor-network state (TNS) admits a tractable matrix product state (MPS) or tree tensor network (TTN) representation, and which control the complexity of learning the state?

We address these questions using parameterised graph theory. First, we show that cutwidth and tree-cutwidth bound the bond dimension overhead required to represent a TNS as an MPS or TTN. In the TTN case, tree-cutwidth also bounds the local dimension of the grouped subsystems. The proofs are based on entanglement rerouting, a tensor-network analogue of rerouting information in a classical network. Second, we derive graph-dependent upper bounds on the sample and computational complexity of realisable TNS tomography, with exponents that depend on cutwidth, tree-cutwidth, and a new graph parameter, learning complexity, which we bound in terms of degree and treewidth. We obtain these results by extending the disentangling MPS learner of~\cite{Cramer_2010}, as analysed further in~\cite{bakshi2025learning,lin2025efficientclosestmatrixproduct}, to TTNs and to tensor networks on arbitrary known graphs. Finally, we extend the framework beyond the realisable setting. For an arbitrary input state, our agnostic learner outputs a pure state whose fidelity is within additive error $\epsilon$ of the optimum over tensor-network states on the given graph with a given bond dimension, with explicit graph-dependent bounds on sample and computational complexity.
\end{abstract}

\newpage
\tableofcontents
\newpage

\section{Introduction}

Tensor networks (TNs) provide a unifying language for quantum computation and quantum information by expressing large linear maps and quantum states as compositions of small tensors wired according to an underlying graph. In this framework, vertices represent individual tensors containing local physical degrees of freedom, while edges correspond to shared indices that are contracted to form the composite object. The topology of this graph encodes the entanglement structure of the physical state, effectively mapping quantum correlations to geometric connectivity.
This representation is useful for two closely related reasons. First, a wide range of objects arising in quantum computation, such as quantum circuits, measurement patterns, or decoders for quantum error correcting codes, can be written as TNs whose contraction yields an output amplitude or probability~\cite{bravyi2014efficient,vidal2003efficient,shi2006classical,gross2007novel,piveteau2024tensor}. 
Second, a TN representation often exposes which structural aspects of a quantum computation are responsible for classical hardness (e.g., entanglement growth), thereby suggesting tractable subclasses and principled approximation schemes. One of the first results in this direction was obtained by Markov and Shi~\cite{Markov_2008}, showing that quantum circuits with small underlying graph treewidth can be simulated efficiently by contracting the associated TN. 
Currently, TN contraction methods are among the most competitive approaches to the classical simulation and benchmarking of quantum computations, and they underpin large simulation efforts for random circuit sampling for quantum advantage~\cite{zhou2020limits, gray2021hyper, ayral2023density}.
Conversely, tensor-network contraction can encode general quantum computation: additive approximation of a tensor-network contraction is complete for $\mathsf{BQP}$~\cite{arad2010quantum}.

Whereas TN representations with bounded bond dimension and controlled graph parameters admit efficient algorithms for contraction~\cite{o2019parameterization}, optimization~\cite{schollwock2011density}, and verification~\cite{harrow2013testing, soleimanifar2022testing}, TN representations that require large bond dimension or induce large intermediate tensors lead to exponential computational overhead~\cite{shi2006classical, orus2014practical}. Therefore, characterizing states that admit tractable TN representations as well as transforming between different TN representations become essential for determining the algorithmic feasibility of problems in quantum computation and quantum information. 

In addition to efficient contractions and simulability, understanding the learnability of tensor networks is equally important, since it determines when quantum states or processes can be efficiently reconstructed or verified from limited data. 
The learnability of TN states (TNSs) has been studied primarily under strong structural assumptions. For matrix product states (MPSs)~\cite{rommer2007thermodynamic, weichselbaum2009variational}, efficient learning and tomography are possible due to the existence of small separators: cutting a single internal edge separates the system, allowing the global state to be reconstructed from local reduced density matrices with polynomial sample and computational complexity~\cite{landoncardinal2010, Cramer_2010, bakshi2025learning, lin2025efficientclosestmatrixproduct}. 
As we show, these ideas extend to tree tensor networks (TTNs)~\cite{shi2006classical} when the tree structure is known, although additional challenges arise in learning the network topology itself \cite{hashemizadeh2020adaptive}. In contrast, for higher-dimensional TNs and for TNs with loops, such as PEPS~\cite{verstraete2004renormalization}, separators grow with system size, leading to both information-theoretic and computational obstacles to efficient learning, and no general efficient learning algorithms are known in this setting \cite{schuch2007computational,wahl2023simulating}.

\subsection{Overview of results}

In this work, we use parameterised graph theory to study properties of TN representations beyond simulability. We show how to transform a TN representation of a general graph to an MPS or TTN representation with bond dimension bounded in terms of suitable graph parameters. We also give upper bounds on the sample complexity of TNS learning for general graphs in terms of graph parameters. 

To formulate our results, let $G=([n],E)$ denote the unweighted tensor-network topology. For a weight function $w:E\to\mathbb N$, we write $\mathcal S_d(G,w)$ for the pure $n$-qudit states admitting a TN representation with virtual dimension $w(e)$ on each edge $e$. We write $\mathcal S_d(G,\chi)$ for the class in which only a common upper bound $\chi$ on the virtual dimensions is specified. Equivalently,
\begin{equation}
    \mathcal S_d(G,\chi)
    =
    \bigcup_{w:E\to[\chi]}\mathcal S_d(G,w).
\end{equation}
See \Cref{subsec:TN-notation} for the formal definition. 
The convention $w(f)=1$ for $f\notin E$ does not add $f$ to the edge set.
For a fixed weight function $w$, removing an edge $e\in E$ with $w(e)=1$
does not change the represented state class.  Graph parameters for $\mathcal S_d(G,\chi)$ are evaluated on the given graph $G$. We use three standard graph parameters: cutwidth ($\cw$), tree-cutwidth ($\tcw$), and treewidth ($\tw$). This overview gives only their intuitive meaning. Precise definitions and the properties used below appear in \Cref{subsec:graph-theory}.

\paragraph{Entanglement rerouting.}

Our first result is a local rerouting rule suggested by classical communication. A direct link from Alice to Bob can be replaced by a route through Charlie when both Alice and Bob are linked to Charlie, without changing the transmitted message. See \Cref{fig:rerouting-information} for an illustration. We prove an analogous operation for virtual indices, which we call \emph{entanglement rerouting}:

\begin{figure}[t]
\centering
\begin{tikzpicture}[
  v/.style={circle,draw=black!80,fill=white,line width=0.7pt,minimum size=17pt,inner sep=1pt,font=\small},
  gedge/.style={line width=0.7pt,black!60},
  redge/.style={line width=1.2pt,figaccent},
  ghost/.style={line width=0.7pt,black!45,dashed},
  lab/.style={font=\small,figaccent},
  note/.style={font=\small\itshape,black!70}
]
\begin{scope}
  \node[v] (A) at (0,0) {$A$};
  \node[v] (B) at (3.2,0) {$B$};
  \node[v] (C) at (1.6,1.45) {$C$};
  \draw[gedge] (A)--(C);
  \draw[gedge] (C)--(B);
  \draw[redge,-latex] (A)--node[below,lab]{message}(B);
  \node[font=\small] at (-0.8,1.5) {(a)};
\end{scope}
\draw[-latex,line width=0.9pt,black!70] (4.3,0.6) -- (5.5,0.6)
  node[midway,above,note]{reroute via $C$};
\begin{scope}[xshift=7.2cm]
  \node[v] (A) at (0,0) {$A$};
  \node[v] (B) at (3.2,0) {$B$};
  \node[v] (C) at (1.6,1.45) {$C$};
  \draw[ghost] (A)--(B);
  \draw[redge,-latex] (A)--node[above left=-2pt,lab]{message}(C);
  \draw[redge,-latex] (C)--(B);
  \node[font=\small] at (-0.8,1.5) {(b)};
\end{scope}
\end{tikzpicture}
\caption{\textbf{Rerouting classical information.} A message from $A$ to $B$ can be
sent over a direct link, panel (a), or mediated by $C$, panel (b). Entanglement
rerouting, \Cref{inf-thm:entanglement-rerouting}, is the tensor network analogue of
this operation.}
    \label{fig:rerouting-information}
\end{figure}
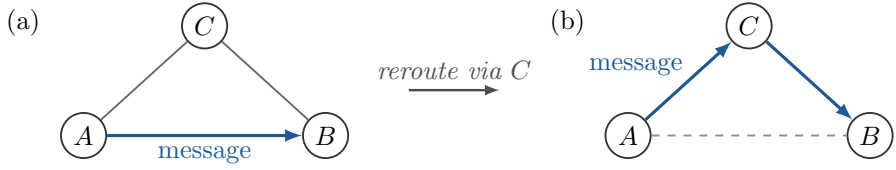

\begin{theorem}[Entanglement rerouting]\label{inf-thm:entanglement-rerouting}
Let $G=([n],E)$ be a graph, let $w:E\to\mathbb N$, and let
$|\psi\rangle\in\mathcal S_d(G,w)$. Choose an edge $e=\{x,y\}\in E$ and a
vertex $z\in[n]\setminus\{x,y\}$. Let
\begin{equation}
    E'=(E\setminus\{e\})\cup\{\{x,z\},\{y,z\}\},
\end{equation}
and extend $w$ by the convention $w(f)=1$ for $f\notin E$. Define
$w':E'\to\mathbb N$ by
\begin{equation}
    w'(f)
    =
    \begin{cases}
        w(f)w(e), & f\in\{\{x,z\},\{y,z\}\},\\
        w(f), & \text{otherwise}.
    \end{cases}
\end{equation}
Then, for $G'=([n],E')$,
\begin{equation}
    |\psi\rangle\in\mathcal S_d(G',w').
\end{equation}
Since $G'$ and $w'$ depend only on $G$, $w$, $e$, and $z$, the same
construction applies to every state in $\mathcal S_d(G,w)$; thus, $\mathcal S_d(G,w)
    \subseteq \mathcal S_d(G',w')$.
In particular, if every original edge dimension is at most $\chi$, then the
rerouted representation has maximum bond dimension at most $\chi^2$.
\end{theorem}

\Cref{inf-thm:entanglement-rerouting} (proved formally as \Cref{thm:entanglement-rerouting} below) shows that a virtual index can be routed through an intermediate vertex by enlarging the two edges on the new route.

\paragraph{Tensor network representations from graph parameters.}

Repeated applications of entanglement rerouting transform a TN on a general graph into one on a chosen target graph. This may increase the bond dimension. Our next two results bound this increase for paths and trees using cutwidth and tree-cutwidth.

\begin{theorem}[Transforming to MPS (informal)]\label{inf-thm:trafo-to-mps}
    A tensor network state on $n$ qudits with graph $G$, physical dimension $d$, and maximum bond dimension $\chi$ can be represented as an MPS on $n$ qudits with physical dimension $d$ and maximum bond dimension $\chi^{\cw(G)}$. Here, the cutwidth $\cw(G)$ is the smallest number $c$ such that the vertices of $G$ can be arranged on a line with at most $c$ edges crossing between any prefix of the arrangement and the rest. An ordering achieving the cutwidth can be computed in time $2^{O(\cw(G)^2)}n$.
\end{theorem}

\begin{theorem}[Transforming to TTN (informal)]\label{inf-thm:trafo-to-ttn}
    A tensor network state on $n$ qudits with graph $G$, physical dimension $d$, and maximum bond dimension $\chi$ can be represented as a TTN on at most $n$ grouped subsystems with physical dimension at most $d^{2\tcw(G)}$ and maximum bond dimension at most $\chi^{2\tcw(G)}$. Here, the tree-cutwidth $\tcw(G)$ is the analogue of cutwidth in which the vertices are organised into a tree of small groups rather than along a line. A tree-cut decomposition certifying these bounds can be computed in time $2^{O(\tcw(G)^2\log\tcw(G))}n^2$.
\end{theorem}

Every quantum state admits an MPS or TTN representation at sufficiently large bond dimension. \Cref{inf-thm:trafo-to-mps,inf-thm:trafo-to-ttn}, proved as \Cref{thm:reroute-to-mps,thm:reroute-to-ttn}, give explicit bounds on the bond dimension required when the starting point is a TN on a general weighted graph.

\paragraph{Learning tensor network states.}

We next use these representation results to study TNS tomography. Throughout, ``learning a state'' means performing tomography of that state, and we use the two terms interchangeably. Given a TNS on a complicated graph, we may first obtain a representation on a simpler graph and then apply a learner tailored to that representation.

We begin by proving efficient tomography of TTNs with known topology.

\begin{theorem}[TTN state tomography (informal)]\label{inf-thm:ttn-tomography}
    We can perform tomography of an unknown $n$-qudit TTN state with bond dimension $\chi$ (and with a known tree structure) from
    \begin{equation}
        O\left( \frac{n^{3}}{\varepsilon^{4}}\left( (d\chi)^{\max\{2,\Delta\}} + \log(n/\delta) \right)\right)
    \end{equation}
    many copies, where $\Delta=\Delta(T)$ is the maximum degree of the underlying tree $T$, $\varepsilon$ is the desired accuracy in trace norm, and $1-\delta$ is the desired success probability.
    The runtime of the tomography procedure is polynomial in $n$, $(d\chi)^{\max\{2,\Delta\}}$, $1/\varepsilon$, and $\log(1/\delta)$.
\end{theorem}

\Cref{inf-thm:ttn-tomography}, proved as \Cref{thm:ttn-tomography}, extends efficient MPS tomography~\cite{Cramer_2010,landoncardinal2010} to TTNs. For constant $d$ and $\chi$, the algorithm is both sample efficient and computationally efficient for TTNs of maximum degree $O(\log n)$. For paths, it recovers the known MPS bounds.

This gives two tractable target representations for learning, namely MPSs and TTNs. Combining the representation transformations with the corresponding tomography algorithms yields the following result for general TNSs:
\begin{theorem}[Parameterised complexity of black-box TNS tomography (informal)]\label{inf-thm:general-tns-learning-blackbox}
    We can perform tomography of an unknown TN state on $n$ qudits with (known) graph $G$, physical dimension $d$, and maximum bond dimension $\chi$ from
    \begin{equation}
        O\left( \frac{n^{3}}{\varepsilon^{4}}\left( \min\left\{ d^{2}\chi^{2\cw(G)},\ (d\chi)^{2\tcw(G)\max\{2,\Delta(\hat{\mathcal{T}})\}} \right\} + \log(n/\delta) \right)\right)
    \end{equation}
    many copies, where $\Delta(\hat{\mathcal{T}})$ is the maximum degree of the decomposition tree obtained from the computed tree-cut decomposition, $\varepsilon$ is the desired accuracy in trace norm, and $1-\delta$ is the desired success probability. The learner runs in time polynomial in $n$, the quantity inside the minimum, $1/\varepsilon$, and $\log(1/\delta)$, after computing an optimal cutwidth ordering in time $2^{O(\cw(G)^2)}n$ or an approximate tree-cut decomposition in time $2^{O(\tcw(G)^2\log\tcw(G))}n^2$.
\end{theorem}

The proof of \Cref{inf-thm:general-tns-learning-blackbox}, given in \Cref{thm:blackbox-mps-learning,thm:blackbox-tcw-learning}, first uses cutwidth or tree-cutwidth to obtain a structured representation and then applies the corresponding tomography algorithm. The graph parameters bound the increase in bond dimension and thereby the overall learning complexity.

The approach to general-graph TNS tomography in \Cref{inf-thm:general-tns-learning-blackbox} uses the transformations of \Cref{inf-thm:trafo-to-mps,inf-thm:trafo-to-ttn} as black boxes (see also the discussion in \Cref{subsec:technical-overview} below).
The preceding construction is modular, but treating the representation transformations as black boxes can yield weaker bounds. Our next result works directly with cuts of the original graph along a chosen sequence of subsets, and its cost is described by a finer graph parameter.

\begin{theorem}[Parameterised complexity of direct TNS tomography (informal)]\label{inf-thm:general-tns-learning-white-box}
    We can perform tomography of an unknown TNS on $n$ qudits with (known) graph $G$, physical dimension $d$, and maximum bond dimension $\chi$ from
    \begin{equation}
        O\left( \frac{n^{3}}{\varepsilon^{4}}\left( d^{\,\lc_{d,\chi}(G)} + \log(n/\delta) \right)\right)
    \end{equation}
    many copies, where $\varepsilon$ is the desired accuracy in trace norm, and $1-\delta$ is the desired success probability.
Here, $\lc_{d,\chi}(G)$ is a graph parameter that we call \emph{learning complexity}. It is the minimum over learning sequences of the largest combined size of an active subsystem and its retained residual register. We show that
    \begin{equation}
        \lc_{d,\chi}(G)
        \leq
        \min\left\{
            n,
            3\max\{1,\lceil \CC(G)\log_d\chi\rceil\}
        \right\}.
    \end{equation}
    In particular,
    $\lc_{d,\chi}(G)=O\left(\Delta\,\tw(G)\max\{1,\lceil\log_d\chi\rceil\}\right)$
    when $G$ is connected and has at least two vertices. Here, $\CC(G)$ is the contraction complexity of Markov and Shi \cite{Markov_2008}, $\Delta=\Delta(G)$ is the maximum degree of $G$, and the treewidth $\tw(G)$ measures how well the vertices of $G$ can be organised into a tree of small overlapping bags.
    The learner runs in time polynomial in $n$, $d^{\,\lc_{d,\chi}(G)}$, $1/\varepsilon$, and $\log(1/\delta)$ once a learning sequence achieving the bound is supplied.  A learning sequence attaining the upper bound can be constructed from an optimal contraction sequence in time $2^{O(\CC(G)^2)}\poly(n)$.
\end{theorem}

The formal version of \Cref{inf-thm:general-tns-learning-white-box} is
\Cref{thm:learning-sequence-tomography}, while \Cref{cor:lc-upper-bound} gives
the displayed estimate in terms of contraction complexity. The direct approach
does not choose an MPS or TTN target in advance. Its exponent depends on the cuts
encountered along the chosen learning sequence.

The precise learning-sequence bound contains the path and tree schedules used by
the black-box algorithms as special cases. Optimising over learning sequences can
only improve on these choices. The contraction-complexity estimate is a more general
bound obtained from a particular learning sequence and can be less sharp on some
graph families.

Finally, we extend the direct approach to agnostic tomography. The input may
lie outside $\mathcal S_d(G,\chi)$, and the algorithm outputs a pure state whose
fidelity is within additive error $\epsilon$ of the optimum over that class.
The output is represented by the circuit produced by the algorithm and need not lie in $\mathcal S_d(G,\chi)$. That is, it need not admit a
representation on $G$ with bond dimension at most $\chi$.

\begin{theorem}[Agnostic TNS tomography (informal)]
\label{inf-thm:agnostic-tns-tomography}
Let $\rho$ be an arbitrary $n$-qudit state, and let a graph $G$, a bond dimension bound $\chi$, and a learning sequence for $G$ be known. Suppose further that
\begin{equation}
    \operatorname{OPT}_{G,\chi}(\rho)
    :=
    \sup_{|\phi\rangle\in\mathcal S_d(G,\chi)}
    \langle\phi|\rho|\phi\rangle
    \geq
    \vartheta
\end{equation}
for a known $\vartheta>0$.
There is an algorithm that outputs a pure state $|\hat\psi\rangle$ such that
\begin{equation}
    \langle\hat\psi|\rho|\hat\psi\rangle
    \geq
    \operatorname{OPT}_{G,\chi}(\rho)-\epsilon
\end{equation}
with probability at least $1-\delta$. Its copy and computational complexity are polynomial in the length of the learning sequence, $1/\vartheta$, $1/\min\{\epsilon,\vartheta\}$, and $\log(1/\delta)$, and exponential only in the largest active register induced by the sequence.
\end{theorem}

The formal statement of \Cref{inf-thm:agnostic-tns-tomography} and its specialisations to contraction sequences and TTNs are given in \Cref{thm:agnostic-learning-sequence-tomography,cor:agnostic-contraction-complexity,cor:agnostic-ttn-tomography}.

\subsection{Technical overview}\label{subsec:technical-overview}

\paragraph{Entanglement rerouting.}
Suppose that an edge $e=\{x,y\}$ in a TNS with weighted graph $(G,w)$ is to be rerouted through a vertex $z\neq x,y$, as shown in \Cref{fig:rerouting-tns}. Multiply the dimensions of $\{x,z\}$ and $\{y,z\}$ by $w(e)$ and delete $e$. View the two enlarged indices as pairs $(i,j)$ and $(k,\ell)$, where $j,\ell\in[w(e)]$. The new tensor at $z$ retains the old tensor entries when $j=\ell$ and is zero otherwise. This enforces equality of the two components carrying the rerouted index, so summing over the enlarged indices reproduces the original contraction.

\begin{figure}[ht]
\centering
\begin{tikzpicture}[
  v/.style={circle,draw=black!80,fill=white,line width=0.7pt,minimum size=16pt,inner sep=1pt,font=\small},
  gedge/.style={line width=0.7pt,black!60},
  redge/.style={line width=1.2pt,figaccent},
  ghost/.style={line width=0.7pt,black!45,dashed},
  lab/.style={font=\small,figaccent},
  glab/.style={font=\small,black!70},
  note/.style={font=\small\itshape,black!70}
]
\begin{scope}
  \node[v] (x) at (0,0) {$x$};
  \node[v] (y) at (3.2,0) {$y$};
  \node[v] (z) at (1.6,1.5) {$z$};
  \draw[gedge] (x)--node[above left=-2pt,glab]{$w(\{x,z\})$}(z);
  \draw[gedge] (z)--node[above right=-2pt,glab]{$w(\{y,z\})$}(y);
  \draw[redge] (x)--node[below,lab]{$w(e)$}(y);
  \node[font=\small] at (-0.9,1.55) {(a)};
\end{scope}
\draw[-latex,line width=0.9pt,black!70] (4.5,0.65) -- (6.1,0.65)
  node[midway,below=2pt,note]{reroute $e$ via $z$};
\begin{scope}[xshift=8.3cm]
  \node[v] (x) at (0,0) {$x$};
  \node[v] (y) at (3.2,0) {$y$};
  \node[v] (z) at (1.6,1.5) {$z$};
  \draw[redge] (x)--node[above left=-2pt,lab]{$w(\{x,z\})\,w(e)$}(z);
  \draw[redge] (z)--node[above right=-2pt,lab]{$w(\{y,z\})\,w(e)$}(y);
  \draw[ghost] (x)--(y);
  \node[font=\small] at (-0.9,1.55) {(b)};
\end{scope}
\node[note,anchor=west,text width=13.2cm] at (-0.5,-1.15)
  {new tensor at $z$: read each enlarged index as a pair, keep the old value when
   the two components carrying the $e$-index agree, and set it to $0$ otherwise};
\end{tikzpicture}
\caption{\textbf{Entanglement rerouting.} The edge $e=\{x,y\}$ of a weighted
tensor network graph, panel (a), is deleted and its virtual dimension is routed
through $z$ by multiplying the weights of the two rerouting edges by $w(e)$, panel (b). The
tensor at $z$ is updated so that the overall contraction is unchanged,
see \Cref{inf-thm:entanglement-rerouting}.}
\label{fig:rerouting-tns}
\end{figure}
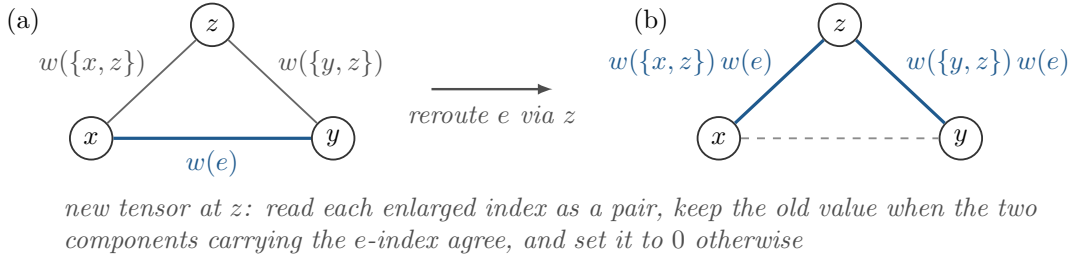

\paragraph{Tensor network representations from graph parameters.}
We describe the MPS construction first. The TTN construction is analogous, with a tree-cut decomposition replacing the linear ordering. Fix an ordering $v_1,\ldots,v_n$ and take $v_1-v_2-\cdots-v_n$ as the target path. Every original edge $\{v_a,v_b\}$ with $a<b-1$ is routed along the segment from $v_a$ to $v_b$. Its weight is multiplied into every path edge on that segment. The final bond dimension on $\{v_i,v_{i+1}\}$ is the product of the weights of the original edges crossing the corresponding prefix cut. If at most $c$ edges cross any prefix cut and every original weight is at most $\chi$, this bond dimension is at most $\chi^c$. The maximum prefix-cut size is the cutwidth of the ordering, and minimising over orderings gives $\cw(G)$.

For the TTN construction, the bags of a tree-cut decomposition become grouped physical subsystems. Each original edge is routed along the unique path between the bags containing its endpoints. Adhesion sizes bound the resulting bond dimensions, while torso sizes bound the dimensions of the grouped physical systems.

\paragraph{Learning tree tensor network states.}
Our starting point is the MPS tomography algorithm of \cite{Cramer_2010}, which proceeds from left to right along the known path, unitarily disentangling one qudit at a time. At each step, a single virtual edge separates the active block from the unprocessed suffix, so its reduced state has rank at most $\chi$. This rank bound allows the disentangling unitary to be learned from few copies.

We extend the same strategy to trees by processing vertices from the leaves towards the root. At a vertex $u$, the active subsystem consists of $u$ and the residual registers retained by its child subtrees. A single tree edge separates this subsystem from the remainder of the tree, so its reduced-state rank is again at most $\chi$. The degree $\Delta$ enters through the active-subsystem size rather than the cut size: there is one fresh qudit and at most $\Delta-1$ residual registers, each containing approximately $\log_d\chi$ qudits.

\paragraph{Learning general tensor network states.}
One approach to learning a general TNS is first to use entanglement rerouting to obtain an MPS or TTN representation, then apply the corresponding learner. This construction gives \Cref{inf-thm:general-tns-learning-blackbox}.

The direct approach behind \Cref{inf-thm:general-tns-learning-white-box} does not choose a fixed target representation. It processes the graph according to a learning sequence, a rooted schedule in which subsets of vertices are assembled into the full vertex set. At step $i$, the learner acts on the fresh vertices introduced at that step and the residual registers retained by its children. It then compresses this active subsystem to a residual register of approximately $|\cut_G(S_i)|\log_d\chi$ qudits. The standard rank bound across a TN cut gives reduced-state rank at most $\chi^{|\cut_G(S_i)|}$. The analysis uses only this rank bound and imposes no further structure on either side of the cut.

The learning complexity $\lc_{d,\chi}(G)$ records the largest combined size of an active subsystem and its retained residual register, minimised over learning sequences. \Cref{lem:contraction-to-learning-sequence} shows that every contraction sequence induces a learning sequence, yielding the upper bound in \Cref{inf-thm:general-tns-learning-white-box}.

\subsection{Related work}

\paragraph{Parameterised complexity for tensor networks.} 

Markov and Shi initiated the parameterised analysis of TN contraction by relating classical simulation costs to the treewidth of the associated graph~\cite{Markov_2008}. Arad and Landau showed that additive TN evaluation is $\mathsf{BQP}$-complete~\cite{arad2010quantum}. Together, these results motivate the study of structural restrictions that make TN problems tractable.

O'Gorman~\cite{o2019parameterization} formulated TN contraction explicitly within parameterised complexity and analysed alternative width parameters. Jakes-Schauer, Anekstein, and Wocjan~\cite{jakes2019carving} gave empirical evidence that carving-width is a useful measure of the memory required to contract planar tensor networks. Dudek, Due\~{n}as-Osorio, and Vardi~\cite{dudek2019efficient} related optimal contraction orders to carving decompositions and used tree-decomposition heuristics to apply TN methods to weighted model counting.

Graph parameters have also been used extensively in practical contraction-ordering methods, including benchmarks of treewidth-based strategies~\cite{dumitrescu2018benchmarking} and the exact and approximate contraction methods of Gray and collaborators~\cite{gray2021hyper,gray2024hyperoptimized}. Cheng et al.~\cite{cheng2025breakingtreewidthbarrierquantum} showed that decision-diagram representations can outperform purely treewidth-based contraction on some circuit families, suggesting that parameters beyond treewidth may describe simulability more accurately.

\paragraph{Learning simple tensor network states.}

The earliest results on efficient tomography for structured TN states are due to Landon-Cardinal, Liu, and Poulin~\cite{landoncardinal2010} and Cramer et al.~\cite{Cramer_2010}. They use local reduced states and local disentangling operations to reconstruct an MPS with constant bond dimension and known ordering from polynomially many copies. Cramer et al. give both a scheme based on local unitaries and a scheme based only on local measurements with more involved classical postprocessing. These results rely on the path topology of the underlying graph: cutting a single virtual edge separates the system into two parts, allowing the global state to be reconstructed from local information. Fewer results are known beyond MPS: For higher-dimensional TNs and TNs with loops, separators grow with system size, and no efficient learning algorithm is known. Existing computational hardness results suggest that none exists in general~\cite{schuch2007computational, wahl2023simulating}. Recent work has extended these results for MPS. Bakshi et al.~
\cite{bakshi2025learning} give a proper closest-product-state learner and an improper closest-MPS learner (in Appendix~B). Their analysis also gives bounds for propagating subspace errors, of which we use a sharper version in our agnostic analysis (see \Cref{lem:agnostic-projection-bound}). Lin, Chia, and Hung~\cite{lin2025efficientclosestmatrixproduct} improve the system-size dependence of MPS tomography from quintic to cubic via a clever parallelisation of the iterative disentangling. Our reduction to MPS learning invokes our sequential learner by default and their parallel learner when circuit depth matters (see \Cref{thm:blackbox-mps-learning} and the accompanying remark). In the agnostic setting, where the input state need not lie in the promised class, guarantees were known for the closest product state \cite{bakshi2025learning} and for the closest MPS \cite{bakshi2025learning, lin2025efficientclosestmatrixproduct}. Our results extend agnostic tomography to tensor network states on trees and on general graphs (see \Cref{rem:agnostic-comparison}). Regarding structure learning, Hashemizadeh et al.~\cite{hashemizadeh2020adaptive} propose an adaptive heuristic for learning TN topologies from data, but without provable guarantees on sample or computational complexity. To our knowledge, no prior work exhibited an efficient tomography algorithm with explicit guarantees for TTN states beyond the MPS special case. Also, the dependence of learning complexity on graph parameters of the underlying entanglement graph had not been previously investigated.

\subsection{Discussion and outlook}

Our work develops the interface between graph theory and tensor networks beyond classical simulation. Prior work~\cite{Markov_2008,cheng2025breakingtreewidthbarrierquantum} showed that parameterised complexity can provide useful structural descriptions of TN simulation. We show that the same perspective also gives a systematic way to analyse the efficiency of TN representations and the sample and computational complexity of learning them.
We next describe several open questions within this setting.

\paragraph{TNS structure learning.}
All our TNS learning results assume that the underlying graph is known. This assumption already appears for MPS: the learners of \cite{landoncardinal2010,Cramer_2010} require a known ordering of the subsystems and use that ordering in the iterative-disentangling procedure. Recent work has studied structure learning for classical graphical models, quantum channels, and quantum Hamiltonians~\cite{bresler2015efficientlylearningising,vuffray2016interactionscreening,klivans2017learninggraphicalmodels,rouzé2023efficientlearningstructureparameters,bluhm2026hamiltonianpropertytesting,bakshi2024structurelearning,zhao2025learningstructure,hu2025ansatzfree,lewis2026learningstructureopenquantum}. This motivates the corresponding problem for TNSs: given only that an unknown state admits a representation on some graph with tractable structure, recover a suitable graph from copies of the state.
\paragraph{Learning versus simulation.}
Efficient quantum learnability often seems to coincide with efficient classical simulability, such as for TNS classes with tractable structure as discussed in this work \cite{vidal2003efficient, Markov_2008, Cramer_2010}, for Clifford circuits \cite{gottesman1998heisenbergrepresentationquantumcomputers, aaronson2004improved, low2009learning} and stabilizer states \cite{aaronson2008identifyingstabilizerstates, montanaro2017learningstabilizerstatesbell, rocchetto2018stabiliser}, for Clifford+$T$ circuits with few $T$-gates \cite{bravyi2016improved, lai2022learningquantumcircuitsofsomeTgates} and their output states \cite{Grewal2025efficientlearningof, Leone2024learningtdoped}, for non-interacting-fermion states \cite{valiant2001quantum, terhal2002classical, aaronson2023efficienttomographynoninteractingfermion, bittel2025optimal}, for Gaussian bosonic unitaries \cite{fanizza2025efficientlearningbosonicgaussian} and states \cite{Bittel2025optimalestimatesof, mele2025learning, bittel2025energyindependenttomographygaussianstates}, and for bosonic or fermionic operations with few non-Gaussian gates \cite{ReardonSmith2024improvedsimulation, Dias2024classicalsimulation, dias2024classical, cudby2024learninggaussianoperations, cudby2025gaussiandecompositionmagicstates, iyer2025mildlyinteractingfermionicunitariesefficiently} and their output states \cite{mele2025efficient}.
While these separate results can be viewed as evidence for a connection between learnability and simulability, whether such a connection can be established in general remains an important open question in quantum learning theory \cite{yoganathan2019conditionclassicalsimulabilityimplies} and quantum machine learning \cite{cerezo2025does}.
Our results strongly hint at the potential of parameterised complexity in investigating such a connection for TNSs by considering which graph parameters govern the complexities of learning and simulations of TNs, respectively. 
Our results relate the two upper bounds through contraction complexity. This parameter determines the cost of classical simulation through the treewidth relation of \cite{Markov_2008} and, by \Cref{cor:lc-upper-bound}, bounds the exponent of our direct learner. It remains open to determine the optimal graph-dependent exponents for both tasks, to prove corresponding lower bounds for learning, and to identify graph families that are efficiently learnable but hard to simulate, or vice versa.

\paragraph{Learning and contraction complexity.}
\Cref{lem:contraction-to-learning-sequence} shows that every contraction sequence induces a learning sequence, and thus
\begin{equation}
    \lc_{d,\chi}(G)
    \leq
    3\max\left\{
        1,
        \left\lceil
            \CC(G)\log_d\chi
        \right\rceil
    \right\}.
\end{equation}
A learning sequence consisting of a single step with $S_1=F_1=V$ and
$I_1=\emptyset$ gives
\begin{equation}
    \lc_{d,\chi}(G)\leq n.
\end{equation}
Combining these bounds yields
\begin{equation}
    \lc_{d,\chi}(G)
    \leq
    \min\left\{
        n,
        3\max\left\{
            1,
            \left\lceil
                \CC(G)\log_d\chi
            \right\rceil
        \right\}
    \right\}.
\end{equation}

The contraction-complexity bound is not always tight. Consider $G=K_n$ with
$d=\chi=2$. The one-step learning sequence gives
$\lc_{2,2}(K_n)\leq n$. Conversely, let $i$ be a leaf of the dependency tree of
an arbitrary learning sequence, and write $s=|S_i|$. Then $I_i=\emptyset$, and we
have $S_i=F_i$. Moreover,
$|\cut_{K_n}(S_i)|=s(n-s)$, so this step contributes
\begin{equation}
    s+s(n-s)
    =
    n+(s-1)(n-s)
    \geq
    n
\end{equation}
to the learning complexity of the sequence.
It follows that
\begin{equation}
    \lc_{2,2}(K_n)=n.
\end{equation}
On the other hand, $\CC(K_n)=\Theta(n^2)$. Indeed, a current vertex
corresponding to a set $S\subseteq V$ has degree $|S|(n-|S|)$. For $n\geq4$,
consider the first set of size at least $n/3$ created by a contraction. Its two
constituent sets have size less than $n/3$, so its size is less than $2n/3$.
The corresponding current vertex therefore has degree $\Omega(n^2)$, while
$|S|(n-|S|)\leq n^2/4$ for every $S\subseteq V$. Hence
$\CC(K_n)=\Theta(n^2)$.

This separation reflects the different operations allowed by the two notions. A contraction sequence uses pairwise contractions and may pay for many edges between partially assembled parts of a dense region. A learning sequence may introduce the entire region in one step, paying instead for the fresh vertices and the cuts connecting the assembled sets to the rest of the graph. Contraction sequences provide a general construction of learning sequences, but not always an optimal one. Given this, we ask:

\begin{question}[Tightness of the contraction-complexity bound]
\label{q:lc-cc-relation}
For which graph families and parameter regimes do we have
\begin{equation}
    \lc_{d,\chi}(G)
    =
    \Theta\left(
        \max\left\{
            1,
            \left\lceil
                \CC(G)\log_d\chi
            \right\rceil
        \right\}
    \right)?
\end{equation}
More generally, how large can the gap be between optimal learning complexity and the upper bound based on contraction, and which graph properties determine it?
\end{question}

The exact relation between learning complexity and the optimal copy complexity beyond the bound in \Cref{inf-thm:general-tns-learning-white-box} also remains open.

\begin{question}[Optimal graph-dependent complexity of TNS tomography]
\label{q:optimal-tns-tomography}
What is the optimal copy complexity of TNS tomography as a function of the underlying graph? In particular:
\begin{enumerate}
    \item Is there a family of graphs $G_n$ and states in $\mathcal S_d(G_n,\chi)$ for which every tomography algorithm requires $d^{\Omega(\lc_{d,\chi}(G_n))}$ copies?
    \item Can an adaptive learner, which re-optimises the remaining learning sequence after each step using the information obtained so far, achieve copy complexity below $d^{\lc_{d,\chi}(G)}$ for some graph family or some inputs?
\end{enumerate}
\end{question}

\paragraph{TNS testing.}
There has been a recent surge in interest in property testing of quantum states (e.g., \cite{odonnell2015quantumspectrumtesting, montanaro2016survey, badescu2019quantumstatecertification, gross2021schur, grewal2024improvedstabilizerestimation, hinsche2025single-copystabilizertesting, arunachlam2025polynomialtimetolerant, bao2025toleranttestingstabilizerstates, iyer2025toleranttestingstabilizerstates, beckey2025producttestingsinglecopymeasurements, girardi2025gaussiantestingbosonicquantum, Flammia2024quantumchisquared, aliakbarpour2025adversariallyrobustquantumstate, caro2025testingclassicalpropertiesquantum}), unitaries (e.g., \cite{low2009learning, chen2023testing}), Hamiltonians (e.g., \cite{aharonov2022quantum,laborde2022quantum,bluhm2026hamiltonianpropertytesting,kallaugher2025hamiltonianlocalitytesting,gao2025quantumhamiltoniancertification,sinha2025improvedhamiltonianlearningsparsity, bluhm2025certifyinglearningquantumising, lee2025optimalcertificationconstantlocalhamiltonians}), and channels (e.g., \cite{bao2025testing}).
In particular, MPS and TTN testers have recently been developed that extend the product testing algorithm of Harrow and Montanaro~\cite{harrow2013testing, soleimanifar2022testing, lovitz2024nearlytightboundstesting}. 

Can these testing algorithms be extended to achieve testing for more general TN graph structures, with complexities governed by suitable graph parameters? 
We note that, while our results (\Cref{inf-thm:trafo-to-mps,inf-thm:trafo-to-ttn}) on transforming to MPS or TTN representations are immediately useful to learning, they unfortunately cannot be straightforwardly applied to lift MPS or TTN testers to general graphs, since one cannot in general reduce testing a property $\mathcal{P}$ to testing a ``super-''property $\mathcal{Q}$ with $\mathcal{P}\subseteq\mathcal{Q}$.
They do, however, transfer in a relaxed form. By \Cref{thm:reroute-to-ttn}, any
tester for a TTN class accepts every state of $\mathcal S_d(G,\chi)$, with
certainty for testers of perfect completeness such as that of
\cite{lovitz2024nearlytightboundstesting}, while its soundness guarantee refers
only to the larger class $\mathcal S_{d^{2\tcw(G)}}(\hat{\mathcal T},\chi^{2\tcw(G)})$. The
obstruction to full soundness comes from the underlying structure: grouping the vertices in a bag hides the internal edge structure that a tester for $\mathcal S_d(G,\chi)$ would need to verify.
\begin{question}[Testing tensor network structure]
\label{q:tns-testing}
Is there a tester for membership in $\mathcal S_d(G,\chi)$, with two-sided error
and copy complexity polynomial in $n$ and $\chi^{\tcw(G)}$, or in $n$ and
$\chi^{\cw(G)}$?
\end{question}
Further questions arise from this parameterised perspective. Can one test graph parameters such as cutwidth, treewidth, or tree-cutwidth of a general-graph TNS under a bond-dimension promise? For a fixed graph, or a fixed graph class, and a fixed bond dimension, what is the complexity of testing whether a state belongs to the corresponding TNS class?

\paragraph{Graph-theoretic insights for TN heuristics.}

Graph-theoretic ideas already appear, often implicitly, in numerical TN heuristics. In DMRG for quantum chemistry~\cite{schollwock2011density}, for example, the optimisation cost depends strongly on the orbital ordering. Existing heuristics use the Fiedler vector of a mutual-information or exchange matrix~\cite{wouters2015chemps2}, matrix bandwidth~\cite{rissler2006measuring}, or block entropies~\cite{legeza2003optimizing} to keep strongly correlated orbitals close.

\Cref{inf-thm:trafo-to-mps} gives a rigorous explanation for this practice: the MPS bond-dimension overhead is bounded by the cutwidth of the graph in the chosen ordering. Existing orbital-ordering heuristics can be interpreted as approximate cutwidth minimisers. Parameterised cutwidth algorithms~\cite{THILIKOS20051,THILIKOS200525} provide an alternative to greedy methods on graphs of small cutwidth.

A similar opportunity arises for TTN methods such as $\mathrm{T}^3$NS~\cite{gunst2018t3ns} and hierarchical Tucker decompositions~\cite{grasedyck2010hierarchical}. Tree topologies are currently chosen largely through physical intuition or clustering heuristics. \Cref{inf-thm:trafo-to-ttn} identifies tree-cutwidth as a relevant parameter and approximate tree-cut decompositions~\cite{Kim2018} as a natural starting point.

Modern TN contraction methods~\cite{dumitrescu2018benchmarking,zhou2020limits,gray2021hyper,ayral2023density,gray2024hyperoptimized} use heuristic searches for contraction orderings that approximate treewidth or contraction complexity. Our results show that related graph parameters also control representation transformations and learning.

This opens the possibility of a unified approach in which a single decomposition, guided by graph parameters, is reused across simulation, compression to a simpler representation (by \Cref{inf-thm:entanglement-rerouting}), and tomographic reconstruction. While our exact theorems do not by themselves
replace numerical heuristics (we think truncation will remain essential in practice), they identify the structural quantities that those heuristics implicitly optimise. We expect that integrating parameterised algorithms for cutwidth, treewidth, and tree-cutwidth into existing TN optimisation approaches will both improve the analysis of current methods and suggest concrete pathways for further practical improvements.

\section{Definitions and notation}
\subsection{Graph theory and graph parameters}\label{subsec:graph-theory}

We recall the graph-theoretic notions and parameters used throughout the paper. For more details on graph theory and parameterised complexity, we refer the reader to textbooks such as \cite{jungnickel2008graphs, cygan2015parameterized}.

\paragraph{Graph theory basics.}

A graph $G=(V,E)$ consists of a finite vertex set $V$ and an edge set
$E \subseteq \big\{ \{u,v\}~|~u,v\in V,\ u\not=v\big\}$. We usually take
$V=[n]$, and we write $V(G)$ and $E(G)$ when the graph in question needs to be
named explicitly. In figures, we draw each vertex as a circle, or as a filled dot
when vertex names are irrelevant, and each edge as a line joining its two
endpoints, as in \Cref{fig:example-graphs} below.

The degree of a vertex $v\in V$ is
$\deg(v)=\big|\{u~|~\{u,v\}\in E\}\big|$, and
$\Delta(G)=\max_{v\in V}\deg(v)$ is the maximum degree of $G$. We suppress the
argument and simply write $\Delta$ whenever the graph is clear from context.

We will also work with weighted graphs, in which every edge carries a positive
integer weight.

\begin{definition}\label{def:weighted-graph}
    A weighted graph is a pair $(G,w)$, where $G=(V,E)$ is a graph and
    $w:E\to[M]$ is a positive integer weight function for some
    $M\in\mathbb N$. When $G$ is clear from context, we refer simply to the
    weight function $w$.
\end{definition}
For convenience, we extend $w$ to all unordered pairs of vertices by setting
\begin{equation}
    w(\{u,v\})=1
\end{equation}
whenever $\{u,v\}\notin E$. Edges of weight one may also be present explicitly. We ignore edges of weight $1$ for calculating parameters, even when explicitly present, as they are equivalent to non-existent edges for our purposes. Unless stated otherwise, we will work with simple graphs. Parallel edges can be combined as described in \Cref{rem:multi-edges}.

A path from $u$ to $v$ consists of distinct vertices
\begin{equation}
    x_1=u,x_2,\ldots,x_k=v
\end{equation}
and edges $\{x_i,x_{i+1}\}$ for $i\in[k-1]$. We denote it by $P$, or by $P_{u,v}$ when the endpoints matter. Adding the edge $\{u,v\}$ to a path on at least three vertices produces a cycle, and a graph is acyclic if it contains no cycle. We allow disconnected graphs and isolated vertices. Statements involving contraction sequences assume connectedness explicitly and apply separately to the nontrivial connected components.

A tree is a connected acyclic graph. Equivalently, a tree is a graph in which
any two vertices are joined by a unique path, and a tree on $k$ vertices has exactly
$k-1$ edges. We usually denote trees by $T$, and we note that every path is a
tree. Designating one vertex $r\in V$ as the root turns $T$ into a rooted tree.
When no root is specified, we choose one as convenient.
For a vertex $u$ other than the root, the parent of $u$ is the unique neighbour of
$u$ on the path $P_{u,r}$, and $u$ is a child of its parent. The root has no
parent. We write $\operatorname{ch}(u)$ for the set of children of $u$, and a
vertex without children is called a leaf. The depth of a vertex $v$, denoted
$\depth(v)$, is the number of edges on the path $P_{r,v}$, with $\depth(r)=0$,
and the height of $T$ is $h=\max_{v\in V}\depth(v)+1$. Finally, the subtree
rooted at $u$, denoted $T_u$, is the induced subgraph on the vertex set
$V(T_u)=\{v\in V(T)~|~u\in V(P_{v,r})\}$, that is, on all vertices whose path to
the root passes through $u$.

Next, we recall the notion of a cut.

\begin{definition}
    Let $G=(V,E)$ be a graph and let $A,B\subseteq V$ be disjoint. We define
    \begin{equation}
        \cut_G(A,B)
        =
        \big\{\{x,y\}\in E~|~x\in A,\ y\in B\big\},
        \qquad
        \cut_G(A)=\cut_G(A,V\setminus A).
    \end{equation}
    When the graph is clear from context, we write $\cut(A,B)$ and $\cut(A)$.
\end{definition}

Deleting the edge between a vertex $v\neq r$ and its parent separates $V(T_v)$
from the rest of the tree, so $|\cut(V(T_v))|=1$ for every $v\neq r$, while
$\cut(V(T_r))=\emptyset$.

We next introduce treewidth, contraction complexity, cutwidth, and tree-cutwidth, using the graphs in \Cref{fig:example-graphs} as recurring examples.
\begin{figure}[h]
    \centering
\begin{tikzpicture}[
  dot/.style={circle,fill=black!75,inner sep=1.7pt},
  gedge/.style={line width=0.7pt,black!70},
  plab/.style={font=\small,black!80}
]
\begin{scope}[xshift=0cm]
  \node[dot] (c) at (0,0) {};
  \foreach \a in {0,45,...,315}{
    \node[dot] (s\a) at (\a:0.95) {};
    \draw[gedge] (c)--(s\a);
  }
  \node[plab] at (0,-1.55) {star};
\end{scope}
\begin{scope}[xshift=3.1cm]
  \node[dot] (r) at (0,0.95) {};
  \node[dot] (a) at (-0.95,0.15) {};
  \node[dot] (b) at (-0.32,0.15) {};
  \node[dot] (c) at (0.32,0.15) {};
  \node[dot] (d) at (0.95,0.15) {};
  \node[dot] (a1) at (-1.25,-0.65) {};
  \node[dot] (a2) at (-0.65,-0.65) {};
  \node[dot] (c1) at (0.32,-0.65) {};
  \node[dot] (d1) at (0.75,-0.65) {};
  \node[dot] (d2) at (1.25,-0.65) {};
  \draw[gedge] (r)--(a) (r)--(b) (r)--(c) (r)--(d);
  \draw[gedge] (a)--(a1) (a)--(a2) (c)--(c1) (d)--(d1) (d)--(d2);
  \node[plab] at (0,-1.55) {tree};
\end{scope}
\begin{scope}[xshift=6.3cm,yshift=0.9cm]
  \foreach \r in {0,...,3}{
    \foreach \c in {0,...,3}{
      \node[dot] (g\r\c) at ({\c*0.6-0.9},{-\r*0.6}) {};
    }}
  \foreach \r in {0,...,3}{
    \foreach \c in {0,...,2}{
      \pgfmathtruncatemacro{\cc}{\c+1}
      \draw[gedge] (g\r\c)--(g\r\cc);
    }}
  \foreach \r in {0,...,2}{
    \foreach \c in {0,...,3}{
      \pgfmathtruncatemacro{\rr}{\r+1}
      \draw[gedge] (g\r\c)--(g\rr\c);
    }}
  \node[plab] at (0,-2.45) {grid};
\end{scope}
\begin{scope}[xshift=9.0cm]
  \foreach \k in {0,...,5}{
    \node[dot] (k\k) at ({90+\k*60}:0.9) {};
  }
  \foreach \k in {0,...,5}{
    \foreach \l in {0,...,5}{
      \ifnum\l>\k
        \draw[gedge] (k\k)--(k\l);
      \fi
    }}
  \node[plab] at (0,-1.55) {clique};
\end{scope}
\begin{scope}[xshift=10.8cm]
  \node[dot] (t1) at (0,0.45) {};
  \node[dot] (t2) at (0.45,-0.45) {};
  \node[dot] (t3) at (0.9,0.45) {};
  \node[dot] (t4) at (1.35,-0.45) {};
  \node[dot] (t5) at (1.8,0.45) {};
  \node[dot] (t6) at (2.25,-0.45) {};
  \node[dot] (t7) at (2.7,0.45) {};
  \draw[gedge] (t1)--(t2)--(t3)--(t4)--(t5)--(t6)--(t7);
  \draw[gedge] (t1)--(t3) (t2)--(t4) (t3)--(t5) (t4)--(t6) (t5)--(t7);
  \node[plab] at (1.35,-1.55) {2-tree};
\end{scope}
\end{tikzpicture}
    \caption{\textbf{Recurring example graphs.} From left to right: a star, a tree,
a square grid, a clique, and a $2$-tree (in fact, the example graph shown is a $2$-path). These families reappear throughout the
paper to illustrate graph parameters (\Cref{tab:parameter-examples}) and tensor
network classes.}
    \label{fig:example-graphs}
\end{figure}
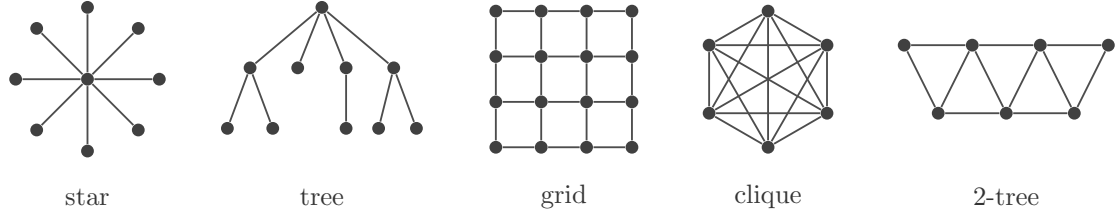
\paragraph{Treewidth and contraction complexity.}
Treewidth measures how closely a graph resembles a tree. Its definition uses a tree decomposition, which partitions the vertices into overlapping bags indexed by a tree.
We use $\mathcal T$ for a decomposition tree whose vertices index bags, and $T$ for a tree that is itself used as a tensor-network graph.

\begin{definition}[Tree decomposition {\cite{RobertsonSymour1986}}]
    A {\em tree decomposition} of a graph $G=(V,E)$ is a tree $\mathcal{T}$, together with a mapping $t\mapsto B_t$ from vertices $t\in V(\mathcal{T})$ to subsets $B_t \subseteq V(G)$, called bags, such that the following conditions hold:
    \begin{itemize}
        \item $\bigcup_{t\in V(\mathcal{T})} B_t = V(G)$, that is, each vertex appears in at least one bag.
        \item For every edge $\{u,v\}\in E(G)$, there is a vertex $t\in V(\mathcal{T})$ such that $\{u,v\}\subseteq B_t$, that is, at least one bag contains both endpoints of every edge.
        \item For every $u\in V(G)$, the set $\{t\in V(\mathcal{T})~|~u\in B_t\}$ forms a connected subtree of $\mathcal{T}$.
    \end{itemize}
    The {\em width} of a tree decomposition is $\max_{t\in V(\mathcal{T})}|B_t|-1$. The {\em treewidth} of $G$ is the minimum width over all tree decompositions,
    \begin{equation}
        \tw(G)
        =
        \min_{\mathcal{T}}
        \left(
            \max_{t\in V(\mathcal{T})}|B_t|-1
        \right).
    \end{equation}
\end{definition}
Finding an optimal tree decomposition of $G$ can be done in time $2^{O(\tw^2)}n^{O(1)}$, as shown in \cite{korhonen2023improvedparameterizedalgorithmtreewidth}.

\begin{example}
    We now illustrate this definition with some simple examples (see \Cref{fig:example-graphs}):
    \begin{enumerate}
        \item The treewidth of a tree with at least two vertices is equal to $1$, independently of its size. In particular, the treewidth of the star is $1$.

        \item The treewidth of the clique $K_n$ is $n-1$.

        \item The treewidth of a square lattice with $n^2$ vertices is $n$. Note that such a graph is the entanglement graph of a Projected Entangled Pair State (PEPS). 
    \end{enumerate}
\end{example}

We next introduce contraction complexity, a graph parameter used to analyse the cost of tensor network contraction~\cite{Markov_2008}. It is closely related to treewidth.

\begin{definition}[Contraction complexity {\cite[Definition 4.1]{Markov_2008}}]\label{def:contraction-complexity}
    A contraction sequence starts from $G$ and repeatedly chooses a non-loop edge of the current multigraph and identifies its endpoints, until every connected component of $G$ has been contracted to a single vertex. Parallel edges are retained and counted with multiplicity, while loops created by an identification are deleted. The complexity of the sequence is the maximum degree, counted with multiplicity, of any vertex created by a contraction. Thus the degrees of the initial vertices are not included unless a vertex later appears as the result of a contraction. The {\em contraction complexity} of $G$, denoted by $\cc(G)$, is the minimum of this quantity over all contraction sequences. If no contraction is required, we set $\cc(G)=0$.
\end{definition}

We can bound $\cc(G)$ using maximum degree and treewidth.

\begin{theorem}[Bounds on contraction complexity {\cite[Theorem 4.5]{Markov_2008}}] \label{thm:cc-vs-tw-delta}
    For any graph $G$ with maximum degree $\Delta=\Delta(G)$, we have 
    \begin{equation}
        (\tw(G)-1)/2\le \cc(G) \le \Delta(\tw(G)+1)-1\, .
    \end{equation}
\end{theorem}
Markov and Shi also prove the exact identity $\cc(G)=\tw(L(G))$ and show how to construct a contraction sequence of the same width from a tree decomposition of $L(G)$~\cite{Markov_2008}. Combined with the exact treewidth algorithm of \cite{korhonen2023improvedparameterizedalgorithmtreewidth}, this gives an optimal contraction sequence in time $2^{O(\cc(G)^2)}|E|^{O(1)}$.


\paragraph{Cutwidth and tree-cutwidth.}
Cutwidth~\cite{doi:10.1137/0606026} is the minimum, over linear orderings of the vertices, of the largest edge cut between a prefix and the remaining vertices.

\begin{definition}[Cutwidth]
    The cutwidth of a bijection $\pi:V\rightarrow [n]$ is defined as 
    \begin{equation}
        \cw(\pi)=\max_{i\in[n]} |\cut(\{v\in V~|~\pi(v)\leq i\})|.
    \end{equation}
    The cutwidth of $G=(V,E)$ is defined as $\cw(G) = \min_{\pi} \cw(\pi)$.
\end{definition}
We often call $\pi$ a linear ordering. 
Finding an optimal ordering can be done in time $2^{O(\cw^2)}n$, as shown in \cite{THILIKOS20051, THILIKOS200525}.

Tree-cutwidth~\cite{wollan2015structure} is a tree-structured analogue of cutwidth. Its definition uses a near partition.

\begin{definition}
A \emph{near partition} of $X$ is a pairwise disjoint family $X_1,\ldots,X_k\subseteq X$ whose union is $X$, allowing for the possibility of $X_i$ to be empty.
\end{definition}

We next define a tree-cut decomposition of a graph $G$.

\begin{definition}[Tree-cut decomposition]
    A \emph{tree-cut decomposition} of $G$ is a pair $(\mathcal{T},\mathcal{X})$ where $\mathcal{T}$ is a rooted tree, with root $r$, and $\mathcal{X} = \{X_t \subseteq V(G): t\in V(\mathcal{T})\}$ is a near partition of $V(G)$. The sets $X_t$ are called bags.
\end{definition}

Tree-cutwidth is now defined using two terms, adhesion and torso size. These in turn are defined with respect to a tree-cut decomposition $(\mathcal{T},\mathcal{X})$.

\begin{definition}[Adhesion]
Let $(\mathcal T,\mathcal X)$ be a tree-cut decomposition of $G$. Let $t\neq r$ have parent $u$, and let $e=\{u,t\}$. Deleting $e$ separates $\mathcal T$ into components $\mathcal T_u$ and $\mathcal T_t$ containing $u$ and $t$, respectively. The \emph{adhesion} $\adh(t)$ is the number of edges of $G$ joining vertices in bags on opposite sides of this separation:
    \begin{equation}
        \adh(t) = \left|\cut\left(\bigcup_{b\in V(T_u)} X_b,\bigcup_{b\in V(T_t)} X_b \right)\right| \, .
    \end{equation}
    For the root $r$, which has no parent, we set $\adh(r)=0$.
\end{definition}

\begin{definition}[Torso size]
\label{def:torso-size}
Let $(\mathcal{T},\mathcal X)$ be a tree-cut decomposition of a graph $G$, where
$\mathcal X=\{X_t:t\in V(\mathcal{T})\}$. Fix $t\in V(\mathcal{T})$, and let
$T_1,\ldots,T_\ell$ be the connected components of $\mathcal{T}-t$.

The torso of $(\mathcal{T},\mathcal X)$ at $t$, denoted $H_t$, is the multigraph obtained from
$G$ as follows: For every $i\in[\ell]$, contract the vertex set
\begin{equation}
    \bigcup_{s\in V(T_i)} X_s
\end{equation}
to a single vertex $z_i$, preserving parallel edges and deleting loops. The vertices in
$X_t$ are called the core vertices of $H_t$, and the vertices
$z_1,\ldots,z_\ell$ are called the peripheral vertices of $H_t$.

The $3$-centre of $H_t$, denoted $\widetilde H_t$, is obtained from $H_t$ by repeatedly
suppressing peripheral vertices of degree at most $2$, while never suppressing core
vertices. More precisely, while there exists a peripheral vertex $z\notin X_t$ with
$\deg(z)\leq 2$, perform one of the following operations:
\begin{itemize}
    \item If $\deg(z)=0$ or $\deg(z)=1$, delete $z$ and all edges incident to $z$.
    \item If $\deg(z)=2$, with incident neighbours $a$ and $b$, delete $z$ and its
    incident edges and, if $a\neq b$, add an edge $\{a,b\}$, preserving multiplicity.
\end{itemize}
The torso size of $t$ is
\begin{equation}
    \operatorname{tor}(t)=|V(\widetilde H_t)|.
\end{equation}
\end{definition}

The tree-cutwidth is then defined as follows:

\begin{definition}[Tree-cutwidth]
    Let $(\mathcal{T},\mathcal{X})$ be a tree-cut decomposition of $G$. The \emph{width} of such a tree-cut decomposition is $\max_{t\in V(\mathcal{T})}\{\adh(t),\tor(t)\}$. 
    The \emph{tree-cutwidth} of $G$, denoted $\tcw(G)$, is the minimum width of all tree-cut decompositions of $G$.
\end{definition}
A tree-cut decomposition of $G$ of width at most $2\tcw(G)$ can be computed in time $2^{O(\tcw^2\log\tcw)}n^2$, as shown in \cite{Kim2018}. 

\begin{remark}
\label{rem:torso-size-bounds-bag-size}
The vertices in the central bag $X_t$ are never suppressed when forming the $3$-centre $\widetilde H_t$, so$|X_t|\leq \operatorname{tor}(t)$ for every $t\in V(\mathcal T)$. In particular, every bag in a width-$k$ tree-cut decomposition has size at most $k$.
\end{remark}

\Cref{tab:parameter-examples} summarises the scaling of these parameters on the recurring graph families.

\begin{table}[ht]
\centering
\begin{tabular}{lll}
\hline
Parameter & Small examples & Large examples \\
\hline
maximum degree & paths, grids, binary trees & stars, cliques \\
treewidth & paths, trees, cycles, $k$-trees ($k$ constant) & grids, cliques \\
cutwidth & paths, cycles, bounded-degree trees & stars, cliques, grids \\
tree-cutwidth & paths, trees, cycles & grids, cliques \\
\hline
\end{tabular}
\caption{\textbf{Scaling of graph parameters on example families.} A family is
listed as small if the parameter is constant or logarithmic in the number of
vertices, and as large otherwise. The families are drawn in
\Cref{fig:example-graphs}.}
\label{tab:parameter-examples}
\end{table}

\subsection{Tensor network notation}\label{subsec:TN-notation}

A pure $n$-qudit state $\ket{\psi}\in(\mathbb{C}^d)^{\otimes n}$ can be described by a tensor $(\psi_{i_1,\ldots,i_n})_{i_1,\ldots,i_n=1}^d$ of complex computational basis amplitudes via $\ket{\psi}=\sum_{i_1,\ldots,i_n=1}^d \psi_{i_1,\ldots,i_n}\ket{i_1,\ldots,i_n}$. As this representation uses a tensor with $d^n$ complex numbers, constrained only by the normalisation condition $\braket{\psi}{\psi}=1$, its size is exponential in the system size $n$. However, for many physically relevant states, the tensor $(\psi_{i_1,\ldots,i_n})_{i_1,\ldots,i_n=1}^d$ admits a more efficient representation. 

We will consider \emph{tensor networks} \cite{cirac2021MPS-PEPS-review} as follows. Let $G=([n],E)$ be a graph and let $w:E\to\mathbb N$ be a positive integer weight function. We say that a tensor $(\psi_{i_1,\ldots,i_n})_{i_1,\ldots,i_n=1}^d$ is compatible with $(G,w)$ if, for every vertex $v\in[n]$ with neighbourhood $N(v)=\{u\in[n]~|~\{u,v\}\in E\}=\{u_1,\ldots,u_{\deg(v)}\}$, there exists a tensor $\varphi^{[v]} = \left(\varphi_{i_v}^{j_1^{[v]},\ldots,j_{\deg(v)}^{[v]}}\right)$ with indices $j_k^{[v]}\in [w(\{u_k,v\})]$ and $i_v\in[d]$, such that the following holds. If $E=\{e_1,\ldots,e_m\}$, where $m=|E|$, then for all $i_1,\ldots,i_n\in[d]$,
\begin{equation}\label{eq:general-tn-state}
    \psi_{i_1,\ldots,i_n}
    = \sum_{j_{[e_1]} =1 }^{w(e_1)} \ldots \sum_{j_{[e_{m}]} =1 }^{w(e_{m})}  \varphi^{[1]}_{i_1}\left(\gets j_{[e_1]},\ldots,j_{[e_{m}]}\right)\ldots \varphi^{[n]}_{i_n}\left(\gets j_{[e_1]},\ldots,j_{[e_{m}]}\right) \, .
\end{equation}
Here, we use the notation $\varphi^{[v]}\left(\gets j_{[e_1]},\ldots,j_{[e_{m}]}\right)$ for the version of $\varphi^{[v]}$ in which every index $j_k^{[v]}$ is set to $j_{[e]}$ for $e= \{u_k,v\}$.
To express \Cref{eq:general-tn-state} in words, we say that our amplitude tensor is obtained by contracting the tensor network:
every edge carries one shared summation index of range equal to its weight, every
vertex tensor is evaluated with its virtual slots set to the indices of its incident
edges, and the amplitude is the sum over all virtual indices.
\Cref{fig:tn-contraction} illustrates the contraction for the smallest nontrivial
example.

\begin{figure}[ht]
    \centering
\begin{tikzpicture}[
  gv/.style={circle,draw=black!80,fill=white,line width=0.7pt,minimum size=13pt,inner sep=0.5pt,font=\scriptsize},
  tens/.style={circle,draw=black!80,fill=figaccent!10,line width=0.8pt,minimum size=22pt,inner sep=1pt,font=\small},
  gedge/.style={line width=0.7pt,black!70},
  bond/.style={line width=1.1pt,figaccent},
  leg/.style={line width=0.7pt,black!75},
  lab/.style={font=\small,figaccent},
  note/.style={font=\small\itshape,black!70}
]
\node[gv] (u1) at (0,0) {$1$};
\node[gv] (u2) at (1.4,0) {$2$};
\node[gv] (u3) at (2.8,0) {$3$};
\draw[gedge] (u1)--node[above,font=\small]{$w(e_1)$}(u2);
\draw[gedge] (u2)--node[above,font=\small]{$w(e_2)$}(u3);
\node[note] at (1.4,-0.85) {weighted graph $(G,w)$};
\draw[-latex,line width=0.9pt,black!70] (3.9,0) -- (5.0,0);
\begin{scope}[xshift=6.0cm]
  \node[tens] (T1) at (0,0) {$\varphi^{[1]}$};
  \node[tens] (T2) at (1.9,0) {$\varphi^{[2]}$};
  \node[tens] (T3) at (3.8,0) {$\varphi^{[3]}$};
  \draw[bond] (T1)--node[above,lab]{$j_{[e_1]}$}(T2);
  \draw[bond] (T2)--node[above,lab]{$j_{[e_2]}$}(T3);
  \draw[leg] (T1)--++(0,-0.85) node[below,font=\small]{$i_1$};
  \draw[leg] (T2)--++(0,-0.85) node[below,font=\small]{$i_2$};
  \draw[leg] (T3)--++(0,-0.85) node[below,font=\small]{$i_3$};
\end{scope}
\node[anchor=west,font=\small] at (0.0,-2.35)
  {$\displaystyle
    \psi_{i_1,i_2,i_3}
    \;=\;
    \sum_{j_{[e_1]}=1}^{w(e_1)}\;
    \sum_{j_{[e_2]}=1}^{w(e_2)}
    \varphi^{[1]}_{i_1}\!\left(\gets j_{[e_1]}\right)\,
    \varphi^{[2]}_{i_2}\!\left(\gets j_{[e_1]},j_{[e_2]}\right)\,
    \varphi^{[3]}_{i_3}\!\left(\gets j_{[e_2]}\right)$};
\end{tikzpicture}
    \caption{\textbf{Contracting a tensor network.} Left: a weighted path graph $G$
on three vertices with edges $e_1=\{1,2\}$ and $e_2=\{2,3\}$. Right: compatible
tensors $\varphi^{[v]}$, one per vertex, with a shared virtual index $j_{[e]}$ of
range $w(e)$ for each edge $e$ and one physical index $i_v$ per vertex. The
amplitude $\psi_{i_1,i_2,i_3}$ is obtained by summing over all virtual indices, as
in the displayed instance of \Cref{eq:general-tn-state}. }
    \label{fig:tn-contraction}
\end{figure}

A pure $n$-qudit state $|\psi\rangle$ admits a TN representation on $(G,w)$ if its amplitude tensor in the computational basis is compatible with $(G,w)$. We define the bond dimension of this representation by
\begin{equation}
    \chi(w):=\max\left(\{1\}\cup\{w(e):e\in E\}\right),
\end{equation}
so an edgeless representation has bond dimension $1$.

\begin{definition}[Tensor-network state classes]\label{def:general-tn-state}
    Let $G=([n],E)$ be a graph and let $w:E\to\mathbb N$ be a positive integer
    weight function. We define $\mathcal S_d(G,w)$ to be the class of pure
    $n$-qudit states that admit a TN representation whose virtual index on each
    edge $e$ has dimension $w(e)$. For $\chi\in\mathbb N$, we define the class
    of TNSs on $G$ with bond dimension at most $\chi$ by
    \begin{equation}
        \mathcal S_d(G,\chi)
        :=
        \bigcup_{w:E\to[\chi]}\mathcal S_d(G,w).
        \label{eq:tn-class-bond-dimension}
    \end{equation}
    For $E=\emptyset$, the union contains the unique empty weight function.
\end{definition}
For a fixed weight function $w$, removing an edge $e\in E$ with $w(e)=1$
does not change the represented state class. When graph parameters are associated
with a fixed weighted representation $(G,w)$, we evaluate them after removing all
edges of weight one. For the class $\mathcal S_d(G,\chi)$, the graph $G$ remains
fixed, since the edge dimensions vary over the weight functions in
\eqref{eq:tn-class-bond-dimension}. For the class $\mathcal S_d(G,\chi)$, the graph $G$
remains part of the class definition because the edge dimensions are not fixed.

The notation $\mathcal S_d(G,w)$ specifies the individual edge dimensions, while $\mathcal S_d(G,\chi)$ specifies only an upper bound. A state in $\mathcal S_d(G,w)$ may use a smaller effective dimension on some edges, since unused virtual levels can be padded with zero tensor entries.

\begin{example}
For the recurring graphs of \Cref{fig:example-graphs}, if $G$ is a path, then
$\mathcal S_d(G,\chi)$ is the class of matrix product states with physical
dimension $d$ and bond dimension at most $\chi$. If $G$ is a tree, then
$\mathcal S_d(G,\chi)$ is the corresponding class of tree tensor network
states. If $G$ is the square grid, then $\mathcal S_d(G,\chi)$ is the class
of PEPS with open boundary conditions and bond dimension at most $\chi$.

For the clique $K_n$ with $n\geq2$, already $\chi=d$ suffices to represent every
pure $n$-qudit state. To see this, choose vertex $1$ as a centre. The edge
$\{1,k\}$ carries the physical index of qudit $k$ to the centre, the tensor at
vertex $k>1$ enforces equality between its physical index and this virtual index,
and the tensor at vertex $1$ stores the full amplitude tensor
$\psi_{i_1,\ldots,i_n}$. All clique edges not incident to vertex $1$ are fixed
to one distinguished virtual value. The resulting contraction equals
$\psi_{i_1,\ldots,i_n}$ and uses edge dimension at most $d$. 
The same construction works for a star on $n\geq2$ vertices.
\end{example}

\begin{remark}[Multi-edges and grouping]
\label{rem:multi-edges}
One could allow parallel virtual edges between the same pair of vertices.  Any
family of parallel edges can be grouped into a single edge whose weight is the
product of their weights, without changing the represented class of states,
since a tuple of virtual indices can be encoded as one index of the product
range \cite{bridgeman2017handwaving}.  We work with simple weighted graphs throughout. Grouping is used implicitly whenever rerouting creates an
edge parallel to an existing one (\Cref{section:rerouting}) and when parallel
edges are counted during contractions (\Cref{subsec:contraction-to-learning}).
\end{remark}

We will repeatedly use the following standard rank bound across a TN cut.

\begin{claim}[Rank bound across a TN cut]
\label{claim:cut-rank-bound}
Let $G=([n],E)$ be a graph, let $w:E\to\mathbb N$, and let
$|\psi\rangle\in\mathcal S_d(G,w)$. For every subset $S\subseteq[n]$, we have
\begin{equation}
    \operatorname{rank}
    \left(
        \operatorname{tr}_{[n]\setminus S}
        \left[
            |\psi\rangle\langle\psi|
        \right]
    \right)
    \leq
    \prod_{e\in \cut_G(S)} w(e).
\end{equation}
Here and throughout, the empty product is $1$. In particular, if $w(e)\leq\chi$ for every edge, then
\begin{equation}
    \operatorname{rank}
    \left(
        \operatorname{tr}_{[n]\setminus S}
        \left[
            |\psi\rangle\langle\psi|
        \right]
    \right)
    \leq
    \chi^{|\cut_G(S)|}.
\end{equation}
\end{claim}

\begin{proof}
Let
\begin{equation}
    \cut_G(S)=\{e_1,\ldots,e_m\}.
\end{equation}
Contract all tensors whose physical legs lie in $S$, leaving open the virtual indices
corresponding to edges in $\cut_G(S)$. This gives vectors
\begin{equation}
    |\Phi_{\alpha_1,\ldots,\alpha_m}\rangle_S
\end{equation}
on the Hilbert space of $S$, where
\begin{equation}
    \alpha_j\in[w(e_j)]
\end{equation}
for every $j\in[m]$. Similarly, contracting all tensors whose physical legs lie in
$[n]\setminus S$ gives vectors
\begin{equation}
    |\Gamma_{\alpha_1,\ldots,\alpha_m}\rangle_{[n]\setminus S}
\end{equation}
on the complementary Hilbert space. Thus the global state can be written as
\begin{equation}
    |\psi\rangle
    =
    \sum_{\alpha_1=1}^{w(e_1)}
    \cdots
    \sum_{\alpha_m=1}^{w(e_m)}
    |\Phi_{\alpha_1,\ldots,\alpha_m}\rangle_S
    \otimes
    |\Gamma_{\alpha_1,\ldots,\alpha_m}\rangle_{[n]\setminus S}.
\end{equation}
The Schmidt rank of $|\psi\rangle$ across the bipartition
$S\mid [n]\setminus S$ is at most
\begin{equation}
    \prod_{e\in\cut_G(S)}w(e).
\end{equation}
The rank of the reduced state on $S$ is equal to this Schmidt rank, which proves the
first claim. The second claim follows immediately from $w(e)\leq\chi$.
\end{proof}

\section{Entanglement rerouting for TN representations}\label{section:rerouting}

We now show how the underlying graph of a TN
representation can be modified while preserving the class of states that can be
represented. The basic
operation is local: an edge $\{x,y\}$ can be removed and its virtual index can be
routed through a third vertex $z$, at the cost of increasing the virtual
dimensions of the two edges $\{x,z\}$ and $\{y,z\}$. We call this operation
\emph{entanglement rerouting}. 

We first prove the rerouting move for a single edge in \Cref{thm:entanglement-rerouting}.
The rest of the section then applies this move iteratively. In \Cref{subsec:reroute-mps},
we reroute all edges of a general graph onto a path, obtaining an MPS representation
whose bond dimension is controlled by cutwidth. In \Cref{subsec:reroute-ttn}, we reroute all edges onto a tree, obtaining
TTN representations with bounds in terms of tree-cutwidth. Thus, the local operation in \Cref{thm:entanglement-rerouting}
is the basic mechanism behind all representation transformations in this section.

Throughout this section, we use the convention from \Cref{subsec:TN-notation} that
an absent edge is equivalent to an edge of virtual dimension $1$. Thus, when we write
$w(\{u,v\})=1$ for a pair $\{u,v\}\notin E$, this simply means that we may add a
formal virtual leg of dimension one without changing the represented state.

\subsection{Rerouting single edges in a TN}

We begin with the basic local move. Suppose a TN representation contains an edge
$\{x,y\}$ of virtual dimension $q$. If $z$ is any other vertex, enlarge the virtual legs connecting $x$ to $z$ and
$y$ to $z$ so that each also carries the value of the virtual index on $\{x,y\}$.
The tensor at $z$ is then defined to enforce equality between these two components.
This removes the direct edge $\{x,y\}$, and multiplies the virtual dimensions of
$\{x,z\}$ and $\{y,z\}$ by $q$. In this way, the represented physical state is unchanged. The following theorem states this precisely. We call it a total rerouting statement: the
entire virtual index on $\{x,y\}$ is moved through $z$. We do not need partial
rerouting (where only some of the virtual dimension $q$ of the edge $\{x,y\}$ is rerouted through a third vertex $z$) in the sequel, so we formulate only the total version.

 \begin{theorem}[Total entanglement rerouting: Formal statement of \Cref{inf-thm:entanglement-rerouting}]
\label{thm:entanglement-rerouting}
Let $G=([n],E)$ be a graph, let $w:E\to\mathbb N$, and extend $w$ to all unordered
pairs of distinct vertices by setting $w(f)=1$ whenever $f\notin E$.
Let $|\psi\rangle\in \mathcal S_d(G,w)$ and let $e=\{x,y\}\in E$. Assume that there is a vertex $z\in[n]\setminus\{x,y\}$, and set
\begin{equation}
    q := w(e).
\end{equation}
Define
\begin{equation}
    E' := (E\setminus\{\{x,y\}\})\cup\{\{x,z\},\{y,z\}\},
\end{equation}
and a weight function $w':E'\to\mathbb N$ by
\begin{equation}
w'(f)=
\begin{cases}
q\,w(\{x,z\}), &  \textrm{if }f=\{x,z\},\\
q\,w(\{y,z\}), &  \textrm{if }f=\{y,z\},\\
w(f), & \text{otherwise}.
\end{cases}
\end{equation}
Let $G'=([n],E')$. Then
\begin{equation}
    |\psi\rangle\in \mathcal S_d(G',w').
\end{equation}
Moreover, the dimension of the largest edge in the new representation is
\begin{equation}
    \chi'
    :=
    \max\left(\{1\}\cup\{w'(f):f\in E'\}\right).
\end{equation}
\end{theorem}

Equivalently, the edge $\{x,y\}$ can be removed and its virtual index routed through $z$ by multiplying the dimensions of $\{x,z\}$ and $\{y,z\}$ by $w(\{x,y\})$.

\begin{proof}
Fix a tensor network representation of $|\psi\rangle$ with underlying weighted graph $(G,w)$. Let
\begin{equation}
    q = w(\{x,y\}),\qquad
    a = w(\{x,z\}),\qquad
    b = w(\{y,z\}),
\end{equation}
where $a=1$ or $b=1$ if the corresponding edge is absent from $G$.
We write the tensors at $x,y,z$ as
$A^{i_x}_{\mathbf r_x,\alpha,\gamma},
B^{i_y}_{\mathbf r_y,\beta,\gamma},
M^{i_z}_{\mathbf r_z,\alpha,\beta},
$
where
$   \alpha\in[a],\beta\in[b], \gamma\in[q]$.
Here, $\gamma$ is the virtual index on the edge $\{x,y\}$, $\alpha$ is the virtual index on $\{x,z\}$, $\beta$ is the virtual index on $\{y,z\}$, and the multi-indices $\mathbf r_x,\mathbf r_y,\mathbf r_z$ collect all other virtual legs
incident to $x,y,z$, respectively. If $\{x,z\}$ or $\{y,z\}$ is absent,
the corresponding index has dimension $1$.

We now define new tensors $\widetilde A,\widetilde B,\widetilde M$. The new
edge $\{x,z\}$ has dimension $aq$, and we write its index as a pair
$   (\alpha,\gamma_x)\in [a]\times[q].$
Similarly, the new edge $\{y,z\}$ has dimension $bq$, and we write its
index as a pair $(\beta,\gamma_y)\in [b]\times[q]$. 
Define 
\begin{align*}
    \widetilde A^{i_x}_{\mathbf r_x,(\alpha,\gamma_x)}  &:=  A^{i_x}_{\mathbf r_x,\alpha,\gamma_x}\, ,\\
    \widetilde B^{i_y}_{\mathbf r_y,(\beta,\gamma_y)}
    &:=
    B^{i_y}_{\mathbf r_y,\beta,\gamma_y}\,,\\
    \widetilde M^{i_z}_{\mathbf r_z,(\alpha,\gamma_x),(\beta,\gamma_y)}&:=
    \mathbf 1_{\gamma_x=\gamma_y}\,
    M^{i_z}_{\mathbf r_z,\alpha,\beta}\, .
\end{align*}
All other tensors in the network are unchanged. 

It remains to check that the overall contraction (and thus the state) is unchanged. Contract all tensors
other than those at $x,y,z$, and denote the resulting tensor by $ R^{\mathbf i_{\mathrm{rest}}}_{\mathbf r_x,\mathbf r_y,\mathbf r_z}$. The original tensor amplitudes can be written as
\begin{equation*}
    \psi_{\mathbf i}=
\sum_{\mathbf r_x,\mathbf r_y,\mathbf r_z}
\sum_{\alpha=1}^{a}
\sum_{\beta=1}^{b}
\sum_{\gamma=1}^{q}
A^{i_x}_{\mathbf r_x,\alpha,\gamma}
B^{i_y}_{\mathbf r_y,\beta,\gamma}
M^{i_z}_{\mathbf r_z,\alpha,\beta}
R^{\mathbf i_{\mathrm{rest}}}_{\mathbf r_x,\mathbf r_y,\mathbf r_z}\, .
\end{equation*}
The amplitudes of the modified network are 
\begin{equation*}
    \widetilde\psi_{\mathbf i} = \sum_{\mathbf r_x,\mathbf r_y,\mathbf r_z}
\sum_{\alpha=1}^{a}
\sum_{\gamma_x=1}^{q}
\sum_{\beta=1}^{b}
\sum_{\gamma_y=1}^{q}
\widetilde A^{i_x}_{\mathbf r_x,(\alpha,\gamma_x)}
\widetilde B^{i_y}_{\mathbf r_y,(\beta,\gamma_y)}
\widetilde M^{i_z}_{\mathbf r_z,(\alpha,\gamma_x),(\beta,\gamma_y)}
R^{\mathbf i_{\mathrm{rest}}}_{\mathbf r_x,\mathbf r_y,\mathbf r_z}\, .
\end{equation*}
Substituting the definitions of the modified tensors gives 
\begin{align*}
    \widetilde\psi_{\mathbf i} 
    &=
    \sum_{\mathbf r_x,\mathbf r_y,\mathbf r_z}
    \sum_{\alpha=1}^{a}
    \sum_{\gamma_x=1}^{q}
    \sum_{\beta=1}^{b}
    \sum_{\gamma_y=1}^{q}
    A^{i_x}_{\mathbf r_x,\alpha,\gamma_x}
    B^{i_y}_{\mathbf r_y,\beta,\gamma_y}
    \mathbf 1_{\gamma_x=\gamma_y}
    M^{i_z}_{\mathbf r_z,\alpha,\beta}
    R^{\mathbf i_{\mathrm{rest}}}_{\mathbf r_x,\mathbf r_y,\mathbf r_z}\\
    &= \sum_{\mathbf r_x,\mathbf r_y,\mathbf r_z}
    \sum_{\alpha=1}^{a}
    \sum_{\beta=1}^{b}
    \sum_{\gamma=1}^{q}
    A^{i_x}_{\mathbf r_x,\alpha,\gamma}
    B^{i_y}_{\mathbf r_y,\beta,\gamma}
    M^{i_z}_{\mathbf r_z,\alpha,\beta}
    R^{\mathbf i_{\mathrm{rest}}}_{\mathbf r_x,\mathbf r_y,\mathbf r_z} =\psi_{\mathbf i}\, ,
\end{align*}
The indicator $\mathbf 1_{\gamma_x=\gamma_y}$ enforces $\gamma_x=\gamma_y$, so $\widetilde\psi_{\mathbf i}=\psi_{\mathbf i}$. Thus the modified tensor network represents the same state and has weighted graph $(G',w')$. The expression for $\chi'$ follows from the definition of the maximum bond dimension.
\end{proof}

 Although \Cref{thm:entanglement-rerouting} reroutes an edge through a single
intermediate vertex, it can be applied repeatedly to reroute an edge along a path,
as illustrated in \Cref{fig:reroute-via-path}.
For example, suppose we want to remove an edge $\{u,x\}$ and route it along the
path $u-v_1-v_2-\cdots-v_j-x.$
We first reroute $\{u,x\}$ through $v_1$. This removes $\{u,x\}$ and creates,
or increases the weights of, $\{u,v_1\}$ and $\{v_1,x\}$. We then reroute
$\{v_1,x\}$ through $v_2$, and continue in this way until the virtual index has
been routed along the whole path. Thus, a direct virtual edge can be replaced by a
chain of virtual edges, with the original virtual dimension multiplied into every edge
of the chosen route.

This path-rerouting perspective is what we use in the next subsections. To transform a
general TN graph into a target graph, such as a path or a tree, we route each edge that
is not present in the target graph along an appropriate path in the target graph. The
resulting bond dimension is controlled by the maximum number of original edges whose
indices are routed through any one target edge.
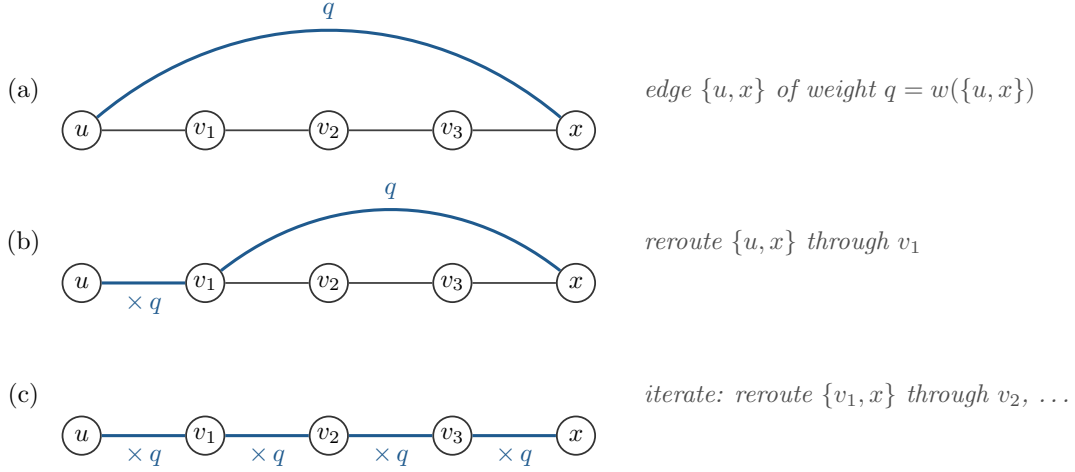
\begin{figure}
    \centering
\resizebox{0.98\textwidth}{!}{%
\begin{tikzpicture}[
  v/.style={circle,draw=black!80,fill=white,line width=0.7pt,minimum size=15pt,inner sep=1pt,font=\small},
  pedge/.style={line width=0.7pt,black!75},
  redge/.style={line width=1.2pt,figaccent},
  lab/.style={font=\small,figaccent},
  note/.style={font=\small\itshape,black!70}
]
\begin{scope}[yshift=0cm]
  \node[v] (u) at (0,0) {$u$};
  \node[v] (v1) at (1.7,0) {$v_1$};
  \node[v] (v2) at (3.4,0) {$v_2$};
  \node[v] (v3) at (5.1,0) {$v_3$};
  \node[v] (x) at (6.8,0) {$x$};
  \draw[pedge] (u)--(v1)--(v2)--(v3)--(x);
  \draw[redge] (u) to[bend left=40] node[above,lab]{$q$} (x);
  \node[note,anchor=west] at (7.6,0.55) {edge $\{u,x\}$ of weight $q=w(\{u,x\})$};
\end{scope}
\begin{scope}[yshift=-2.1cm]
  \node[v] (u) at (0,0) {$u$};
  \node[v] (v1) at (1.7,0) {$v_1$};
  \node[v] (v2) at (3.4,0) {$v_2$};
  \node[v] (v3) at (5.1,0) {$v_3$};
  \node[v] (x) at (6.8,0) {$x$};
  \draw[redge] (u)--node[below,lab]{$\times\,q$}(v1);
  \draw[pedge] (v1)--(v2)--(v3)--(x);
  \draw[redge] (v1) to[bend left=38] node[above,lab]{$q$} (x);
  \node[note,anchor=west] at (7.6,0.55) {reroute $\{u,x\}$ through $v_1$};
\end{scope}
\begin{scope}[yshift=-4.2cm]
  \node[v] (u) at (0,0) {$u$};
  \node[v] (v1) at (1.7,0) {$v_1$};
  \node[v] (v2) at (3.4,0) {$v_2$};
  \node[v] (v3) at (5.1,0) {$v_3$};
  \node[v] (x) at (6.8,0) {$x$};
  \draw[redge] (u)--node[below,lab]{$\times\,q$}(v1);
  \draw[redge] (v1)--node[below,lab]{$\times\,q$}(v2);
  \draw[redge] (v2)--node[below,lab]{$\times\,q$}(v3);
  \draw[redge] (v3)--node[below,lab]{$\times\,q$}(x);
  \node[note,anchor=west] at (7.6,0.55) {iterate: reroute $\{v_1,x\}$ through $v_2$, \dots};
\end{scope}
\node[font=\small] at (-0.8,0.55) {(a)};
\node[font=\small] at (-0.8,-1.55) {(b)};
\node[font=\small] at (-0.8,-3.65) {(c)};
\end{tikzpicture}%
}
\caption{\textbf{Rerouting along a path.} An edge $\{u,x\}$ can be rerouted through
successive intermediate vertices. The example shows three such vertices. Each step
removes the current long edge and multiplies the weights of the two carrying edges by
$q$.}
    \label{fig:reroute-via-path}
\end{figure}

We use the following direct consequence of \Cref{thm:entanglement-rerouting}.

\begin{corollary}\label{cor:rerouting-set-inclusion}
Let $G=([n],E)$ be a graph, let $w:E\to\mathbb N$, let $e=\{x,y\}\in E$, and let
$z\in[n]\setminus\{x,y\}$. Let $(G',w')$ be obtained by totally rerouting $e$ through $z$ as in \Cref{thm:entanglement-rerouting}. Then
\begin{equation}
    \mathcal S_d(G,w)\subseteq \mathcal S_d(G',w').
\end{equation}
If only the uniform bound $w(f)\leq\chi$ is known, then every new edge weight is at most $\chi^2$, and hence
\begin{equation}
    \mathcal S_d(G,\chi)\subseteq \mathcal S_d(G',\chi^2).
\end{equation}
\end{corollary}
%
%
Since a single rerouting step gives an inclusion of TNS classes, any finite
sequence of rerouting steps gives an inclusion from the original TNS class into the
TNS class associated with the final weighted graph. In particular, if we can reroute a
general graph onto a path, we obtain an MPS representation. If we can reroute it onto
a tree, we obtain a TTN representation. The price paid for simplifying the graph is an
increase in the virtual dimensions of the edges that carry rerouted indices.

\subsection{Rerouting to an MPS}\label{subsec:reroute-mps}

We now iteratively apply the single-edge rerouting move from \Cref{thm:entanglement-rerouting}
to transform an arbitrary TN graph into a path, and hence into an MPS representation.
The only choice involved is a linear ordering of the vertices. Once an ordering
$\pi:V\to[n]$ has been fixed, we write
$v_i=\pi^{-1}(i)$ and use the path
$v_1-v_2-\cdots-v_n$ as the target MPS graph.

The construction is conceptually simple. If an original edge $\{v_a,v_b\}$ has
$a<b-1$, we reroute its virtual index along the path
$v_a-v_{a+1}-\cdots-v_b$. This edge contributes a multiplicative factor
$w(\{v_a,v_b\})$ to each path edge $\{v_i,v_{i+1}\}$ with $a\leq i<i+1\leq b$. Thus, the
bond dimension on the path edge $\{v_i,v_{i+1}\}$ is controlled by the number of
original edges crossing the cut
$\{v_1,\ldots,v_i\}\mid\{v_{i+1},\ldots,v_n\}$.
We call these cuts the \emph{prefix cuts} of the ordering.
The maximum size of a prefix cut is the cutwidth of the ordering. Minimising this maximum over all orderings gives $\cw(G)$.

\begin{theorem}[Rerouting to an MPS]
\label{thm:reroute-to-mps}
Let $G=(V,E)$ be a graph with $V=[n]$, let $w:E\to\mathbb N$, and let $\pi:V\to[n]$ be a linear ordering. Write $v_i=\pi^{-1}(i)$ for $i\in[n]$, and let $P_n^\pi$ be the labelled path on vertices $v_1,\ldots,v_n$. Define its edge weights by
\begin{equation}
    w_P(\{v_i,v_{i+1}\})
    =
    \prod_{e\in \cut_G(\{v_1,\ldots,v_i\})} w(e),
    \qquad i\in[n-1],
\end{equation}
where the empty product is $1$. Then
\begin{equation}
    \mathcal S_d(G,w)
    \subseteq
    \mathcal S_d(P_n^\pi,w_P).
    \label{eq:weighted-mps-rerouting}
\end{equation}
Consequently, if only the uniform bound $w(e)\leq\chi$ is specified, then
\begin{equation}
    \mathcal S_d(G,\chi)
    \subseteq
    \mathcal S_d(P_n^\pi,\chi^{\operatorname{Cut\text{-}Width}(\pi)}).
\end{equation}
In particular, for an ordering $\pi^\star$ of minimum cutwidth,
\begin{equation}
    \mathcal S_d(G,\chi)
    \subseteq
    \mathcal S_d(P_n^{\pi^\star},\chi^{\operatorname{Cut\text{-}Width}(G)}).
\end{equation}

Moreover, for a fixed ordering $\pi$, this transformation can be implemented using at
most
\begin{equation}
    \sum_{i=1}^{n-1}
    \left|\cut_G(\{v_1,\ldots,v_i\})\right|
    \leq
    (n-1)\operatorname{Cut\text{-}Width}(\pi)
\end{equation}
single-edge rerouting steps. If no ordering is given, we can first compute an optimal
cutwidth ordering in time $2^{O(\operatorname{Cut\text{-}Width}(G)^2)}n$.
\end{theorem}

\begin{proof}
Fix a linear ordering $\pi$ and write $v_i=\pi^{-1}(i)$. We first describe the
rerouting procedure. Consider an original edge $e=\{v_a,v_b\}\in E$ with $a<b$.
If $b=a+1$, then $e$ is already an edge of the target path and no rerouting is needed.
If $b>a+1$, we reroute $e$ through $v_{a+1}$. This removes $\{v_a,v_b\}$ and creates,
or increases the weights of, the two edges $\{v_a,v_{a+1}\}$ and $\{v_{a+1},v_b\}$.
We then reroute $\{v_{a+1},v_b\}$ through $v_{a+2}$, and continue in this way until
the virtual index has been routed along the full path
\begin{equation}
    v_a-v_{a+1}-\cdots-v_b.
\end{equation}
Thus the original edge $e$ contributes its weight $w(e)$ to every path edge
$\{v_i,v_{i+1}\}$ with $a\leq i<b$.

After all original edges have been processed, every nontrivial virtual edge lies on the target path. Each rerouting preserves the represented state by \Cref{thm:entanglement-rerouting}, so every state in $\mathcal S_d(G,w)$ is represented on $(P_n^\pi,w_P)$.

It remains to bound the final path weights. Fix $i\in[n-1]$. By the construction above,
an original edge $e=\{v_a,v_b\}$ with $a<b$ contributes to the final edge
$\{v_i,v_{i+1}\}$ if and only if $a\leq i<b$. This is precisely the condition that
$e$ crosses the cut
$\{v_1,\ldots,v_i\}\mid\{v_{i+1},\ldots,v_n\}$ in the original graph. Hence
\begin{equation}
    w_P(\{v_i,v_{i+1}\})
    =
    \prod_{e\in \cut_G(\{v_1,\ldots,v_i\})} w(e).
\end{equation}
Equation \eqref{eq:weighted-mps-rerouting} now follows from the expression for
$w_P$. Suppose next that $|\psi\rangle\in\mathcal S_d(G,\chi)$. By definition,
there exists a weight function $w_\psi:E\to[\chi]$ such that
$|\psi\rangle\in\mathcal S_d(G,w_\psi)$. Hence
\begin{align}
    w_P(\{v_i,v_{i+1}\})
    &\leq
    \chi^{|\cut_G(\{v_1,\ldots,v_i\})|} \\
    &\leq
    \chi^{\operatorname{Cut\text{-}Width}(\pi)}.
\end{align}
This proves
\begin{equation}
    \mathcal S_d(G,\chi)
    \subseteq
    \mathcal S_d
    \left(
        P_n^\pi,
        \chi^{\operatorname{Cut\text{-}Width}(\pi)}
    \right).
\end{equation}
Optimising over $\pi$ gives the bound in terms of
$\operatorname{Cut\text{-}Width}(G)$.

Finally, we bound the number of rerouting steps. An original edge
$e=\{v_a,v_b\}$ with $a<b$ requires at most $b-a-1\leq b-a$ rerouting steps.
The total number of rerouting steps is at most
\begin{align}
    \sum_{\{v_a,v_b\}\in E,\ a<b} (b-a)
    &=
    \sum_{i=1}^{n-1}
    \left|\cut_G(\{v_1,\ldots,v_i\})\right| \\
    &\leq
    (n-1)\operatorname{Cut\text{-}Width}(\pi).
\end{align}
If $\pi$ is not supplied, we first compute an optimal cutwidth ordering using the
standard algorithm. This takes time
$2^{O(\operatorname{Cut\text{-}Width}(G)^2)}n$ (see \Cref{subsec:graph-theory}), after which the above rerouting procedure
is applied.
\end{proof}

\begin{algorithm}
\caption{Reroute to an MPS}
\label{alg:reroute-to-mps}
\textbf{Input:} A graph $G=(V,E)$ with $V=[n]$, an edge-dimension function $w:E\to\mathbb N$, and optionally a linear
ordering $\pi:V\to[n]$.

\textbf{Output:} A weighted labelled path $P_n^\pi$ and a sequence of rerouting steps
transforming $G$ into $P_n^\pi$.
\begin{algorithmic}[1]
\If{$\pi$ is not provided}
    \State Compute an optimal cutwidth ordering $\pi$ of $G$.
\EndIf
\State Let $v_i\leftarrow \pi^{-1}(i)$ for every $i\in[n]$.
\State $G'\leftarrow G$.
\For{$i=1$ to $n-2$}
    \For{$j=i+2$ to $n$}
        \If{$\{v_i,v_j\}\in E(G')$}
            \State Reroute the edge $\{v_i,v_j\}$ through $v_{i+1}$ using
            \Cref{thm:entanglement-rerouting}, and update $G'$.
        \EndIf
    \EndFor
\EndFor
\State Retain every target path edge formally, assigning weight $1$ to a trivial virtual leg.
\State \Return the weighted path $P_n^\pi$ and the recorded rerouting sequence.
\end{algorithmic}
\end{algorithm}

\Cref{alg:reroute-to-mps} describes the corresponding transformation of the weighted graph. Given a tensor network representation, each rerouting step is implemented by the tensor update in \Cref{thm:entanglement-rerouting}.
The algorithm may group several virtual indices into a single weighted edge whenever
they have the same endpoints. This is equivalent to keeping parallel virtual legs and then
combining their dimensions by multiplication.

\begin{example}[Square-grid PEPS to MPS rerouting]
\label{ex:grid-peps-to-mps}
Let $G_{m,m}$ be the open-boundary $m\times m$ square grid, which has
$m^2$ vertices. Consider a PEPS-like tensor network state with underlying graph
$G_{m,m}$, physical dimension $d$, and bond dimension $\chi$.

Order the vertices row by row as illustrated in \Cref{fig:peps-rerouting}, i.e., label the vertex in row $r$ and column $c$ by
$(r,c)$ and set
\begin{equation}
    \pi(r,c)=(r-1)m+c.
\end{equation}
For every prefix cut in this ordering, at most $m$ vertical grid edges cross the cut, and
at most one horizontal grid edge crosses the cut. Hence,
\begin{equation}
    \operatorname{Cut\text{-}Width}(\pi)\leq m+1.
\end{equation}
By \Cref{thm:reroute-to-mps}, we obtain
\begin{equation}
    \mathcal S_d(G_{m,m},\chi)
    \subseteq
    \mathcal S_d(P_{m^2}^{\pi},\chi^{m+1}).
\end{equation}
Thus an $m\times m$ PEPS with physical dimension $d$ and bond dimension $\chi$ can
be represented as an MPS on the same $m^2$ qudits, with the same physical dimension
$d$, and with MPS bond dimension at most $\chi^{m+1}$.

This should be distinguished from blocking each row into a single supersite. Row
blocking would produce an MPS on only $m$ sites, but with physical dimension $d^m$.
The rerouting construction above does not block physical systems. It preserves the
original $m^2$ local physical subsystems, and the local physical dimension remains $d$. The
change in representation affects only the bond dimension of the MPS.

\begin{figure}[ht]
    \centering
\begin{tikzpicture}[
  v/.style={circle,draw=black!80,fill=white,line width=0.7pt,minimum size=14pt,inner sep=0.5pt,font=\scriptsize},
  gedge/.style={line width=0.7pt,black!70},
  cedge/.style={line width=1.4pt,figaccent},
  lab/.style={font=\small,figaccent},
  note/.style={font=\small\itshape,black!70}
]
\fill[figaccent!12,rounded corners=6pt] (-0.45,0.45) rectangle (3.75,-0.45);
\fill[figaccent!12,rounded corners=6pt] (-0.45,-0.65) rectangle (1.55,-1.55);
\foreach \r in {1,...,4}{
  \foreach \c in {1,...,3}{
    \draw[gedge] ({\c-1},{-(\r-1)*1.1}) -- ({\c},{-(\r-1)*1.1});
  }}
\foreach \r in {1,...,3}{
  \foreach \c in {1,...,4}{
    \draw[gedge] ({\c-1},{-(\r-1)*1.1}) -- ({\c-1},{-\r*1.1});
  }}
\draw[cedge] (2,0)--(2,-1.1);
\draw[cedge] (3,0)--(3,-1.1);
\draw[cedge] (0,-1.1)--(0,-2.2);
\draw[cedge] (1,-1.1)--(1,-2.2);
\draw[cedge] (1,-1.1)--(2,-1.1);
\foreach \r in {1,...,4}{
  \foreach \c in {1,...,4}{
    \pgfmathtruncatemacro{\i}{(\r-1)*4+\c}
    \node[v] at ({\c-1},{-(\r-1)*1.1}) {\i};
  }}
\node[note,anchor=west] at (4.3,-0.2) {row-major ordering $\pi(r,c)=(r-1)m+c$};
\node[note,anchor=west] at (4.3,-0.9) {shaded: prefix $\{v_1,\dots,v_6\}$};
\node[note,anchor=west,text width=6.2cm] at (4.3,-1.9)
  {thick: $\operatorname{cut}(\{v_1,\dots,v_6\})$, at most $m$ vertical edges and one horizontal edge};
\draw[-latex,line width=0.9pt,black!70] (1.5,-3.85) -- (1.5,-4.75)
  node[midway,right,note]{reroute along the ordering};
\begin{scope}[yshift=-5.6cm]
  \foreach [count=\k] \i in {1,2,3,4,5,6,7}{
    \node[v] (p\k) at ({(\k-1)*0.95},0) {\i};
  }
  \node (dots) at (7.05,0) {$\cdots$};
  \node[v] (p16) at (7.9,0) {16};
  \foreach \k [evaluate=\k as \kk using int(\k+1)] in {1,...,6}{
    \ifnum\k=6
      \draw[cedge] (p\k)--node[below=3pt,lab]{$\le\chi^{\,m+1}$}(p\kk);
    \else
      \draw[gedge] (p\k)--(p\kk);
    \fi
  }
  \draw[gedge] (p7)--(dots);
  \draw[gedge] (dots)--(p16);
\end{scope}
\node[font=\small] at (-1.0,0.45) {(a)};
\node[font=\small] at (-1.0,-5.15) {(b)};
\end{tikzpicture}
\caption{\textbf{Rerouting a square grid to a path, shown for $m=4$.} Panel (a): the row-major
ordering $\pi(r,c)=(r-1)m+c$ of the $m\times m$ grid, with the prefix
$\{v_1,\ldots,v_6\}$ shaded and its prefix cut highlighted, consisting of $m$
vertical edges and one horizontal edge. Panel (b): the resulting path on $m^2$
vertices. Each MPS bond carries the product of the virtual dimensions of the grid
edges crossing the corresponding prefix cut, here at most $\chi^{m+1}$ on the bond
$\{v_6,v_7\}$.}
    \label{fig:peps-rerouting}
\end{figure}
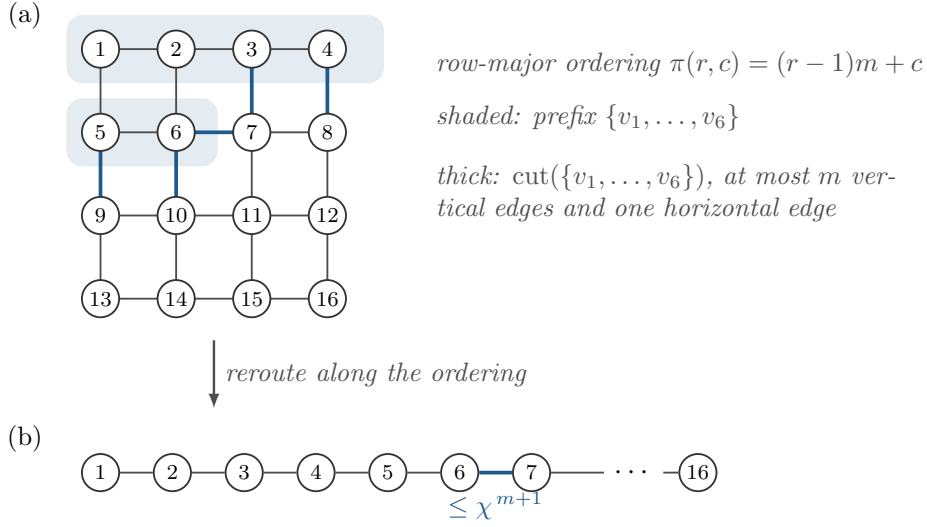
\end{example}

\subsection{Rerouting to a TTN}\label{subsec:reroute-ttn}

We next construct a TTN representation. Unlike the MPS construction of
\Cref{thm:reroute-to-mps}, this construction groups physical systems according
to the bags of a tree-cut decomposition.

Let $(\mathcal T,\mathcal X)$ be a tree-cut decomposition of $G$, where$\mathcal X=\{X_t:t\in V(\mathcal T)\}$.
We group the qudits in each bag $X_t$ into one physical subsystem, whose local
dimension is $d^{|X_t|}$, and use $\mathcal T$ as the target TTN graph. Each
original edge whose endpoints lie in different bags is routed along the unique
path between those bags in $\mathcal T$. A tree edge then carries the virtual
indices of the original edges crossing the associated cut.

Tree-cutwidth bounds two quantities. The torso sizes bound the bag sizes and
hence the local dimensions of the grouped subsystems. The adhesions bound the
number of virtual indices routed through each tree edge and hence the resulting
bond dimensions.

Since tree-cut decompositions are defined using near partitions, some bags may be
empty. Empty bags do not correspond to physical subsystems. Before constructing the
TTN, we remove them by contracting them into neighbouring bags, as formalised in the
next lemma.

\begin{lemma}[Removing empty bags]
\label{lem:remove-empty-bags}
Let $(\mathcal{T},\mathcal X)$ be a tree-cut decomposition of a graph $G$ of width $k$, where
$\mathcal X=\{X_t:t\in V(\mathcal{T})\}$ is a near partition of $V(G)$. Then one can obtain a
tree $\hat{\mathcal{T}}$ and a partition
$\hat{\mathcal X}=\{\hat X_t:t\in V(\hat{\mathcal{T}})\}$ of $V(G)$ into nonempty
bags such that:
\begin{enumerate}
    \item Every bag $\hat X_t$ is one of the nonempty bags of the original
    decomposition.
    \item For every edge $f\in E(\hat{\mathcal{T}})$, the cut of $V(G)$ induced by
    $\hat{\mathcal{T}}-f$ is equal to the cut induced by some edge of $\mathcal{T}$.
    \item Every bag has size at most $k$, and every cut induced as in 2. has size at most
    $k$.
\end{enumerate}
\end{lemma}

Removing empty bags can increase torso sizes and hence the width of the resulting decomposition. The lemma shows, however, that it preserves the original bounds on bag sizes and on cuts induced by decomposition-tree edges.

\begin{proof}
Starting from $(\mathcal{T},\mathcal X)$, repeatedly choose a vertex $t\in V(\mathcal{T})$ with
$X_t=\emptyset$ and contract a tree edge incident to $t$. If the chosen edge is
$\{s,t\}$, the new bag assigned to the contracted vertex is $X_s\cup X_t=X_s$. Thus,
contracting an empty bag does not change any nonempty bag.

We claim that this operation does not increase any relevant adhesion. Indeed, every
tree edge that remains after the contraction corresponds to an edge of the previous tree.
The two connected components obtained by deleting that edge may have changed by the
presence or absence of the empty bag, but the union of the original graph vertices on
each side is unchanged because $X_t=\emptyset$. Hence the induced cut in $G$ is the same
as before. Iterating this argument proves that every edge of the final tree
$\hat{\mathcal{T}}$ induces a cut that already appeared as an adhesion cut in the original
tree-cut decomposition.

After all empty bags have been removed, the remaining bags form a partition of $V(G)$
into nonempty sets. Since the original decomposition has width $k$, all adhesion cuts
have size at most $k$. Moreover, using the standard convention for tree-cut torsos that
the vertices of the central bag are retained in the torso, we have
\begin{equation}
    |X_t|\leq \operatorname{tor}(t)\leq k
\end{equation}
for every original nonempty bag $X_t$. Every final bag also has size at most
$k$.
\end{proof}

The proof of \Cref{thm:entanglement-rerouting} does not require equal local physical dimensions, so it applies after grouping the systems in each bag. Parallel virtual legs between grouped tensors are combined as in \Cref{rem:multi-edges}.

\begin{theorem}[Rerouting to a TTN]
\label{thm:reroute-to-ttn}
Let $G=([n],E)$ be a graph, let $w:E\to\mathbb N$, and let $|\psi\rangle\in\mathcal S_d(G,w)$. Let $(\mathcal{T},\mathcal X)$ be a tree-cut decomposition of $G$ of width $k\geq1$. Let $(\hat{\mathcal{T}},\hat{\mathcal X})$ be the tree-indexed partition obtained from $(\mathcal{T},\mathcal X)$ by removing empty bags as in \Cref{lem:remove-empty-bags}. Write $m=|V(\hat{\mathcal{T}})|$.

Then, after grouping all qudits in the same bag of $\hat{\mathcal X}$ into one physical subsystem, $|\psi\rangle$ has a TTN representation on the weighted tree $(\hat{\mathcal{T}},w_{\hat{\mathcal{T}}})$. More explicitly, for every edge $f\in E(\hat{\mathcal{T}})$, let $A_f\subseteq[n]$ be the union of the bags on one side of the cut $\hat{\mathcal{T}}-f$ and define
\begin{equation}
    w_{\hat{\mathcal{T}}}(f)
    =
    \prod_{e\in \cut_G(A_f)} w(e),
\end{equation}
where the empty product is $1$. Under the canonical identification
\begin{equation}
    \bigotimes_{v\in[n]}\mathbb C^d
    \cong
    \bigotimes_{t\in V(\hat{\mathcal{T}})}(\mathbb C^d)^{\otimes |\hat X_t|},
\end{equation}
and after padding the grouped local Hilbert spaces to dimension $d^k$ if needed, we obtain
\begin{equation}
    |\psi\rangle_{\hat{\mathcal X}}
    \in
    \mathcal S_{d^k}(\hat{\mathcal{T}},w_{\hat{\mathcal{T}}}).
    \label{eq:weighted-ttn-rerouting}
\end{equation}
If $w(e)\leq\chi$ for every original edge, then
\begin{equation}
    w_{\hat{\mathcal{T}}}(f)\leq\chi^k
\end{equation}
for every $f\in E(\hat{\mathcal{T}})$. Hence, after the same grouping and
padding, every state in $\mathcal S_d(G,\chi)$ has a representation in
$\mathcal S_{d^k}(\hat{\mathcal{T}},\chi^k)$.
Moreover,
\begin{equation}
    \left\lceil \frac{n}{k}\right\rceil
    \leq
    m
    \leq
    n.
\end{equation}

In particular, if $(\mathcal{T},\mathcal X)$ is an optimal tree-cut decomposition, then one may
take $k=\operatorname{Tree\text{-}Cut\text{-}Width}(G)$. 

If no decomposition is supplied,
one may first compute a tree-cut decomposition of width at most
$2\operatorname{Tree\text{-}Cut\text{-}Width}(G)$ in time
$2^{O(\operatorname{Tree\text{-}Cut\text{-}Width}(G)^2
\log \operatorname{Tree\text{-}Cut\text{-}Width}(G))}n^2$, and the same statement holds
with $k=2\operatorname{Tree\text{-}Cut\text{-}Width}(G)$.
\end{theorem}

\begin{proof}
Let $(\hat{\mathcal{T}},\hat{\mathcal X})$ be obtained from $(\mathcal{T},\mathcal X)$ by
\Cref{lem:remove-empty-bags}. By that lemma, the bags of
$\hat{\mathcal X}$ form a partition of $[n]$ into nonempty sets, each of size at
most $k$, and every cut of $G$ induced by an edge of $\hat{\mathcal{T}}$ has size at most $k$.

We first group the physical systems according to the bags. For each
$t\in V(\hat{\mathcal T})$, we contract the subnetwork induced by $\hat X_t$,
including every virtual edge whose endpoints both lie in $\hat X_t$. The
resulting tensor has one physical index of dimension $d^{|\hat X_t|}$ and one
virtual leg for each original edge leaving the bag. Since $|\hat X_t|\leq k$,
we embed every grouped physical space into a space of dimension $d^k$. We
combine parallel virtual legs as in \Cref{rem:multi-edges} and then apply
\Cref{thm:entanglement-rerouting}. Grouping and padding preserve the represented
state under the canonical Hilbert-space identification above.

It remains to transform the virtual graph between the grouped tensors into the tree
$\hat{\mathcal{T}}$. Consider an original edge $e=\{u,v\}\in E$. Let $s,t\in V(\hat{\mathcal{T}})$ be
the unique bag labels such that $u\in\hat X_s$ and $v\in\hat X_t$. If $s=t$,
then $e$ was already absorbed into the local tensor at $s$. If $s\neq t$, let
\begin{equation}
    x_1=s,x_2,\ldots,x_\ell=t
\end{equation}
be the unique path from $s$ to $t$ in $\hat{\mathcal T}$. If $\ell=2$, the
edge already lies on $\hat{\mathcal T}$ and no rerouting is required. If
$\ell\geq3$, we route the virtual index of $e$ along this path by repeated
applications of \Cref{thm:entanglement-rerouting}. First reroute the edge from
$x_1$ to $x_\ell$ through $x_2$, then reroute the remaining edge from $x_2$
to $x_\ell$ through $x_3$, and continue through the intermediate vertices
$x_2,\ldots,x_{\ell-1}$. Each step preserves the represented state.

After performing this procedure for every original inter-bag edge, all nontrivial virtual legs lie on edges of $\hat{\mathcal{T}}$. The formal tree edges are retained even when their final virtual dimension is one. Fix a tree edge $f\in E(\hat{\mathcal{T}})$, and
let $A_f$ be the union of the bags on one side of $\hat{\mathcal{T}}-f$. An original edge contributes its virtual dimension to $f$ exactly when its endpoints lie on opposite sides of this cut. The final weight of $f$ is
\begin{equation}
    w_{\hat{\mathcal{T}}}(f)
    =
    \prod_{e\in \cut_G(A_f)} w(e).
\end{equation}
By \Cref{lem:remove-empty-bags}, the cut $\cut_G(A_f)$ has size at most
$k$. Since every original edge weight is at most $\chi$, we obtain
\begin{equation}
    w_{\hat{\mathcal{T}}}(f)
    \leq
    \chi^{|\cut_G(A_f)|}
    \leq
    \chi^k.
\end{equation}
Thus the resulting TTN has bond dimension at most $\chi^k$.

Finally, since the bags form a partition of $[n]$ into nonempty sets and every bag has
size at most $k$, the number $m$ of bags satisfies
\begin{equation}
    \left\lceil \frac{n}{k}\right\rceil
    \leq
    m
    \leq
    n.
\end{equation}
This proves the required representation. The algorithmic statement follows by using the
standard approximation algorithm for tree-cut decompositions, see \Cref{subsec:graph-theory}.
\end{proof}

\begin{algorithm}
\caption{Reroute to a TTN}
\label{alg:reroute-to-ttn}
\textbf{Input:} A graph $G=([n],E)$, an edge-dimension function $w:E\to\mathbb N$, and optionally a tree-cut decomposition
$(\mathcal{T},\mathcal X)$ of $G$.

\textbf{Output:} A tree $\hat{\mathcal{T}}$, a partition
$\hat{\mathcal X}=\{\hat X_t:t\in V(\hat{\mathcal{T}})\}$ of $[n]$, edge weights
$w_{\hat{\mathcal{T}}}$ on $\hat{\mathcal{T}}$, and a sequence of rerouting steps realising the
corresponding TTN representation.
\begin{algorithmic}[1]
\If{$(\mathcal{T},\mathcal X)$ is not provided}
    \State Compute a tree-cut decomposition $(\mathcal{T},\mathcal X)$ of width
    $\leq 2\operatorname{Tree\text{-}Cut\text{-}Width}(G)$.
\EndIf
\While{there exists $t\in V(\mathcal{T})$ with $X_t=\emptyset$}
    \State Choose a neighbour $s$ of $t$ in $\mathcal{T}$.
    \State Contract the tree edge $\{s,t\}$ in $\mathcal{T}$.
    \State Assign the bag $X_s\cup X_t$ to the new contracted vertex.
\EndWhile
\State Let $(\hat{\mathcal{T}},\hat{\mathcal X})$ be the resulting tree-indexed partition.
\State Set $w_{\hat{\mathcal{T}}}(f)\leftarrow 1$ for every $f\in E(\hat{\mathcal{T}})$.
\ForAll{$t\in V(\hat{\mathcal{T}})$}
    \State Group all physical systems in $\hat X_t$ into one physical subsystem.
    \State Absorb all original edges with both endpoints in $\hat X_t$ into the
    local tensor at $t$.
\EndFor
\ForAll{$e=\{u,v\}\in E(G)$}
    \State Let $s,t\in V(\hat{\mathcal{T}})$ be the unique vertices such that
    $u\in \hat X_s$ and $v\in \hat X_t$.
    \If{$s\neq t$}
        \State Let $x_1=s,x_2,\ldots,x_k=t$ be the unique path from $s$ to $t$
in $\hat{\mathcal T}$.
\For{$j=1$ to $k-1$}
    \State
    $w_{\hat{\mathcal T}}(\{x_j,x_{j+1}\})
    \leftarrow
    w_{\hat{\mathcal T}}(\{x_j,x_{j+1}\})w(e)$.
\EndFor
\If{$k\geq3$}
    \State Record the rerouting sequence through
    $x_2,\ldots,x_{k-1}$ using
    \Cref{thm:entanglement-rerouting}.
\EndIf
    \EndIf
\EndFor
\State Retain every edge of $\hat{\mathcal{T}}$ formally, including edges of weight one.
\State \Return $(\hat{\mathcal{T}},\hat{\mathcal X},w_{\hat{\mathcal{T}}})$ and the recorded rerouting sequence.
\end{algorithmic}
\end{algorithm}

\Cref{alg:reroute-to-ttn} computes the final tree weights without updating the
tensors after each rerouting step. An original edge contributes exactly to the
tree edges on the path between the bags containing its endpoints. Parallel
virtual legs are combined as in \Cref{rem:multi-edges}.

\begin{remark}[Using the representation in tomography]
\label{rem:ttn-rerouting-for-tomography}
The resulting TTN has $m=|V(\hat{\mathcal T})|\leq n$ grouped physical
subsystems. After padding, each local physical space has dimension $d^k$, and
the bond dimension is at most $\chi^k$. When this representation is used for
tomography, the relevant tree degree is $\Delta(\hat{\mathcal T})$.
\end{remark}

\section{Tensor network state tomography}
\label{section:tomography}

We now turn from representation questions to tomography. Throughout this section,
learning a state means producing a classical description $\hat\rho$ of an unknown
state $\rho$ such that $\|\rho-\hat\rho\|_1$ is small. We call this quantity the
trace norm error. An $\epsilon$-accurate reconstruction satisfies
$\|\rho-\hat\rho\|_1\leq\epsilon$. The corresponding trace distance is
$D(\rho,\hat\rho):=\frac{1}{2}\|\rho-\hat\rho\|_1$.

The tensor network states considered in this paper are pure states. However, the
tomography procedures below repeatedly learn reduced states of subsystems, and these
reduced states are generally mixed. We now formulate the basic tomography task for arbitrary density operators.

\begin{definition}[State tomography]
\label{def:state-tomography}
Let $\mathcal S$ be a known set of $n$-qudit states. The state tomography problem for
$\mathcal S$ with accuracy parameter $\epsilon$ and confidence parameter $\delta$ is the
following task: given i.i.d. copies of an unknown state $\rho\in\mathcal S$, output a
classical description of a state $\hat\rho$ such that
\begin{equation}
    \|\rho-\hat\rho\|_1\leq \epsilon
\end{equation}
with probability at least $1-\delta$.
If we additionally require $\hat\rho\in\mathcal S$, then we call the task \emph{proper} state
tomography for $\mathcal S$.
\end{definition}

We aim to minimise the number of copies of $\rho$ used by a tomography
procedure. We also aim for computational efficiency, although in several places we
separate the statistical and computational aspects of the analysis.

When $\mathcal S$ is the set of all states of rank at most $r$ on a Hilbert space of
dimension $D$, we call the task rank-$r$ state tomography. If $D=d^m$ is
the dimension of an $m$-qudit subsystem, then sample-optimal rank-$r$ tomography in
trace norm has copy complexity
\begin{equation}
    O\left(
        \frac{r d^m+\log(1/\delta)}{\epsilon^2}
    \right),
\end{equation}
and this scaling is optimal up to constants \cite{odonnell2015efficientquantumtomography, Haah_2017, scharnhorst2025optimallowerboundsquantum, pelecanos2025debiasedkeylsalgorithmnew, pelecanos2025mixedstatetomographyreduces}.

\begin{remark}[Choice of tomography routine]
\label{rem:tomography-modularity}
Sample-optimal tomography generally requires collective measurements across many
copies~\cite{chen2023whendoes,chen2024optimal-tradeoff-state-tomography}. When such
measurements are unavailable, one may instead use a single-copy tomography
procedure, at the cost of a larger copy complexity. Every subsequent call to
subnormalised tomography may be replaced by another routine with the same
input-output guarantee. The resulting copy bounds are obtained by substituting
the copy complexity of the chosen routine into
\Cref{lem:subnormalised-tomography-known-mu}.
\end{remark}

The algorithms below use iterative disentangling. At each stage, they learn a
small reduced state, apply a unitary that approximately disentangles part of the
system, project that subsystem onto $|0\rangle$, and continue with the successful
branch. Each subsequent tomography call uses fresh copies after applying the unitaries
and projections learned in the preceding steps.

The next subsection develops tomography for the subnormalised states produced by
known postselection maps. The map changes from one step to the next but is known
once the preceding tomography outcomes have been fixed.

The rest of this section is organised as follows. In \Cref{subsec:subnormalised-tomography},
we show how ordinary tomography implies tomography for subnormalised postselected
states. In \Cref{subsec:mps-tomography}, we recall the iterated-disentangling learner for
MPS. In \Cref{sec:ttn-tomography}, we extend this strategy to TTN states.
In \Cref{subsec:general-TNS-tomography-black-box}, we combine these learners
with the rerouting results of \Cref{section:rerouting} to obtain black-box
tomography for general-graph TNSs.  In \Cref{subsec:direct-graph-tomography}, we give a direct learner whose cost depends on the cuts of the original graph.  In \Cref{subsec:agnostic-tns-tomography}, we extend the
direct learner to the agnostic setting.

\subsection{Auxiliary lemmas on subnormalised tomography}
\label{subsec:subnormalised-tomography}

A subnormalised state is a positive semidefinite operator with trace at most one. The
subnormalised states that arise in our algorithms are produced by postselection. We use
the following notation to describe them.

A known postselection procedure is a two-outcome physical operation whose
successful branch is represented by a known linear map
$K:\mathcal H\to\mathcal H$ satisfying
\begin{equation}
    K^\dagger K\leq I.
\end{equation}
On success, the input state is mapped to
\begin{equation}
    \rho\mapsto K\rho K^\dagger.
\end{equation}
In our applications, $K$ is a composition of known unitaries and
computational-basis projections.

Let $\rho$ be a state on $\mathcal H=(\mathbb C^d)^{\otimes n}$. For a known
postselection operation with measured Kraus operator $K$ and for a subsystem
$L\subseteq[n]$, we define the subnormalised postselected reduced state
\begin{equation}
    \sigma_{K,L}(\rho)
    :=
    \operatorname{tr}_{[n]\setminus L}\left[K\rho K^\dagger\right].
\end{equation}
The success probability of the postselection operation is
\begin{equation}
    \mu_K(\rho)
    :=
    \operatorname{tr}\left[K\rho K^\dagger\right].
\end{equation}
When $\mu_K(\rho)>0$, the corresponding normalised postselected reduced state is
\begin{equation}
    \rho_{K,L}
    :=
    \frac{\sigma_{K,L}(\rho)}{\mu_K(\rho)}.
\end{equation}
With this notation established, we can now introduced the task of subnormalised tomography.

\begin{definition}[Subnormalised tomography]
\label{def:subnormalised-tomography}
Let $\mathcal S$ be a known set of $n$-qudit states. Let $K$ describe the success branch of a known postselection procedure, and let $L\subseteq[n]$ be a subsystem. The
subnormalised tomography problem for $\mathcal S$, $K$, and $L$, with accuracy
parameter $\epsilon$ and confidence parameter $\delta$, is the following task: given
i.i.d. copies of an unknown state $\rho\in\mathcal S$, output a classical description of
a subnormalised state $\hat\sigma_L$ on $L$ such that
\begin{equation}
    \left\|
        \sigma_{K,L}(\rho)
        -
        \hat\sigma_L
    \right\|_1
    \leq
    \epsilon
\end{equation}
with probability at least $1-\delta$.
\end{definition}

The simplest instance of a known postselection procedure is a computational-basis
projection, for example
\begin{equation}
    K=
    |0^i\rangle\langle 0^i|\otimes I_d^{\otimes(n-i)}.
\end{equation}
This describes the postselection in the original MPS learner of
\cite{Cramer_2010} after each disentangling unitary has been absorbed into a
change of basis.  In our formulation the unitaries are kept explicit: after $i$
steps of \Cref{alg:learn-mps}, the relevant operator is
$K_i=P_iU_i\cdots P_1U_1$, a composition of computational-basis projections with
previously learned disentangling unitaries.  This is the general form allowed by
the definition above.

We first consider the case in which the success probability $\mu_K(\rho)$ is known.

\begin{lemma}[Subnormalised tomography with known success probability]
\label{lem:subnormalised-tomography-known-mu}
Let $\mathcal S$ be a known set of $n$-qudit states, let $K$ be the measured Kraus
operator of a known postselection operation, and let $L\subseteq[n]$. Suppose that, for
every $\rho\in\mathcal S$ with $\mu_K(\rho)>0$, the normalised postselected reduced
state $\rho_{K,L}$ has rank at most $r$. Let
\begin{equation}
    D_L=d^{|L|},
\end{equation}
and let $A$ be a tomography algorithm for rank-$r$ states on $L$ which, when run
with accuracy parameter $\alpha\in(0,1]$ and confidence parameter
$\beta\in(0,1)$, uses
\begin{equation}
    m_A(r,D_L,\alpha,\beta)
\end{equation}
copies.
Assume that
\begin{equation}
    \mu=\mu_K(\rho)
\end{equation}
is known. Then there is a subnormalised tomography algorithm for
$\sigma_{K,L}(\rho)$ with the following sample complexity: If $\mu\leq\epsilon$, the
algorithm uses no samples and outputs the zero operator. If $\mu>\epsilon$, it suffices
to use
\begin{equation}
    m_B
    =
    \left\lceil
        \frac{2m_A(r,D_L,\epsilon/\mu,\delta/2)}{\mu}
        +
        \frac{8}{\mu}\log\left(\frac{2}{\delta}\right)
    \right\rceil
\end{equation}
copies of $\rho$.

In particular, if $A$ is a sample-optimal rank-$r$ tomography procedure, then for
$\mu>\epsilon$ the sample complexity is
\begin{equation}
    O\left(
        \frac{\mu(rD_L+\log(1/\delta))}{\epsilon^2}
        +
        \frac{\log(1/\delta)}{\mu}
    \right).
\end{equation}
\end{lemma}

\begin{proof}
If $\mu\leq\epsilon$, then
\begin{equation}
    \left\|\sigma_{K,L}(\rho)\right\|_1
    =
    \operatorname{tr}\left[\sigma_{K,L}(\rho)\right]
    =
    \mu
    \leq
    \epsilon,
\end{equation}
and we already have an $\epsilon$-accurate estimate.

It remains to consider $\mu>\epsilon$. Apply the postselection operation independently
to $m$ copies of $\rho$, and let $M$ be the number of successful outcomes. Then
\begin{equation}
    M\sim\operatorname{Binomial}(m,\mu).
\end{equation}
Let
\begin{equation}
    N
    =
    m_A(r,D_L,\epsilon/\mu,\delta/2).
\end{equation}
The choice of $m=m_B$ ensures that
\begin{equation}
    m\mu
    \geq
    2N+8\log\left(\frac{2}{\delta}\right).
\end{equation}
In particular, $N\leq m\mu/2$. By the multiplicative Chernoff bound,
\begin{equation}
    \Pr[M<N]
    \leq
    \Pr[M<m\mu/2]
    \leq
    \exp(-m\mu/8)
    \leq
    \delta/2.
\end{equation}
Thus, with probability at least $1-\delta/2$, we obtain at least $N$ successful
postselected copies.

Conditioned on success, the reduced state on $L$ is $\rho_{K,L}$. Therefore, on the
event $M\geq N$, we run $A$ on $N$ successful postselected copies and obtain an estimate
$\hat\rho_{K,L}$ satisfying
\begin{equation}
    \left\|\rho_{K,L}-\hat\rho_{K,L}\right\|_1
    \leq
    \epsilon/\mu
\end{equation}
with probability at least $1-\delta/2$. The algorithm outputs
\begin{equation}
    \hat\sigma_L
    =
    \mu\hat\rho_{K,L}.
\end{equation}
A union bound over the postselection and tomography events gives an overall success probability of
at least $1-\delta$, and on this event
\begin{align}
    \left\|\sigma_{K,L}(\rho)-\hat\sigma_L\right\|_1
    &=
    \left\|\mu\rho_{K,L}-\mu\hat\rho_{K,L}\right\|_1 \\
    &=
    \mu\left\|\rho_{K,L}-\hat\rho_{K,L}\right\|_1 \\
    &\leq
    \epsilon.
\end{align}
This proves the claim.
\end{proof}

In the algorithms below, the exact success probability is usually not known. Instead,
the analysis provides upper and lower bounds on it. We thus provide a suitable version of subnormalised tomography for when only bounds on the success probability of the postselection are known.

\begin{corollary}[Subnormalised tomography from bounds on success probability]
\label{cor:subnormalised-tomography-bounded-mu}
Let $\mathcal S$, $K$, $L$, $r$, $D_L$, and $A$ be as in
\Cref{lem:subnormalised-tomography-known-mu}. Suppose that the success probability
\begin{equation}
    \mu=\mu_K(\rho)
\end{equation}
is unknown, but that known bounds
\begin{equation}
    0<\mu_\ell\leq \mu\leq\mu_u\leq 1
\end{equation}
are available. If $\mu_u\leq\epsilon$, then the zero operator is an $\epsilon$-accurate
estimate. Otherwise, there is a subnormalised tomography algorithm with sample
complexity
\begin{equation}
    m_B
    =
    \left\lceil
        \frac{2m_A(r,D_L,\epsilon/(2\mu_u),\delta/3)}{\mu_\ell}
        +
        \frac{8}{\mu_\ell}\log\left(\frac{3}{\delta}\right)
        +
        \frac{2}{\epsilon^2}\log\left(\frac{6}{\delta}\right)
    \right\rceil.
\end{equation}
If $A$ is a sample-optimal rank-$r$ tomography procedure, this becomes
\begin{equation}
    O\left(
        \frac{\mu_u^2(rD_L+\log(1/\delta))}{\mu_\ell\epsilon^2}
        +
        \frac{\log(1/\delta)}{\mu_\ell}
        +
        \frac{\log(1/\delta)}{\epsilon^2}
    \right).
\end{equation}
\end{corollary}

\begin{proof}
If $\mu_u\leq\epsilon$, then
\begin{equation}
    \left\|\sigma_{K,L}(\rho)\right\|_1
    =
    \mu
    \leq
    \mu_u
    \leq
    \epsilon,
\end{equation}
so the zero operator suffices as an $\epsilon$-accurate estimate.

Assume now that $\mu_u>\epsilon$. Apply the postselection operation independently to
$m$ copies of $\rho$, let $M$ be the number of successful outcomes, and set
\begin{equation}
    \hat\mu=\frac{M}{m}.
\end{equation}
Let
\begin{equation}
    N
    =
    m_A(r,D_L,\epsilon/(2\mu_u),\delta/3).
\end{equation}
The first two terms in the definition of $m_B$ imply
\begin{equation}
    m\mu
    \geq
    m\mu_\ell
    \geq
    2N+8\log\left(\frac{3}{\delta}\right).
\end{equation}
As in the proof of \Cref{lem:subnormalised-tomography-known-mu}, a Chernoff bound gives
\begin{equation}
    \Pr[M<N]\leq \delta/3.
\end{equation}
The final term in the definition of $m_B$ gives, by Hoeffding's inequality,
\begin{equation}
    \Pr\left[
        |\hat\mu-\mu|>\epsilon/2
    \right]
    \leq
    2\exp(-m\epsilon^2/2)
    \leq
    \delta/3.
\end{equation}
On the event $M\geq N$, we run $A$ on $N$ successful postselected copies. With
probability at least $1-\delta/3$, it outputs $\hat\rho_{K,L}$ satisfying
\begin{equation}
    \left\|\rho_{K,L}-\hat\rho_{K,L}\right\|_1
    \leq
    \frac{\epsilon}{2\mu_u}.
\end{equation}
The algorithm outputs
\begin{equation}
    \hat\sigma_L
    =
    \hat\mu\hat\rho_{K,L}.
\end{equation}
By a union bound, all three good events occur with probability at least $1-\delta$. On
this event,
\begin{align}
    \left\|\sigma_{K,L}(\rho)-\hat\sigma_L\right\|_1
    &=
    \left\|\mu\rho_{K,L}-\hat\mu\hat\rho_{K,L}\right\|_1 \\
    &\leq
    \mu\left\|\rho_{K,L}-\hat\rho_{K,L}\right\|_1
    +
    |\mu-\hat\mu|\left\|\hat\rho_{K,L}\right\|_1 \\
    &\leq
    \mu_u\frac{\epsilon}{2\mu_u}
    +
    \epsilon/2 \\
    &=
    \epsilon.
\end{align}
This proves the claim.
\end{proof}

In the remainder of the paper, we write
$\operatorname{Sub\text{-}Tomography}(K,L,r,\epsilon,\delta,[\mu_\ell,\mu_u])$
for the procedure of \Cref{cor:subnormalised-tomography-bounded-mu} run with
postselection operator $K$, subsystem $L$, rank parameter $r$, accuracy
$\epsilon$, confidence $\delta$, and success probability bounds
$\mu_\ell\leq\mu_K(\rho)\leq\mu_u$.  This is the calling convention used by all
algorithms below.  The singleton interval $[\mu,\mu]$ recovers the setting of
\Cref{lem:subnormalised-tomography-known-mu}, and the rank parameter may be set
to the full dimension $D_L$ when no rank promise is available, as in
\Cref{alg:agnostic-learn-tns-from-sequence}.

\paragraph{Renormalisation.}
The procedure returns a subnormalised estimate. At several points we require the
corresponding normalised state. The following estimate bounds the error introduced
by normalisation.

Let $\sigma$ be a subnormalised state with
\begin{equation}
    \operatorname{tr}(\sigma)=\mu\geq\mu_\ell>0,
\end{equation}
and suppose that $\hat\sigma$ is a positive semidefinite estimate satisfying
\begin{equation}
    \|\sigma-\hat\sigma\|_1\leq\eta<\mu_\ell.
\end{equation}
Let
\begin{equation}
    \hat\mu=\operatorname{tr}(\hat\sigma).
\end{equation}
Then $|\hat\mu-\mu|\leq\eta$, $\hat\mu>0$, and
\begin{equation}
    \left\|
        \frac{\sigma}{\mu}
        -
        \frac{\hat\sigma}{\hat\mu}
    \right\|_1
    \leq
    \frac{2\eta}{\mu_\ell}.
\end{equation}
Indeed,
\begin{align}
    \left\|
        \frac{\sigma}{\mu}
        -
        \frac{\hat\sigma}{\hat\mu}
    \right\|_1
    &\leq
    \frac{\|\sigma-\hat\sigma\|_1}{\mu}
    +
    \left|\frac{1}{\mu}-\frac{1}{\hat\mu}\right|
    \operatorname{tr}(\hat\sigma) \\
    &=
    \frac{\|\sigma-\hat\sigma\|_1}{\mu}
    +
    \frac{|\mu-\hat\mu|}{\mu} \\
    &\leq
    \frac{2\eta}{\mu_\ell}.
\end{align}
Normalising a subnormalised estimate introduces an error factor proportional to
the inverse success probability. The algorithms below work with subnormalised
branches and normalise only in the final step.

\subsection{MPS tomography via iterated disentangling}
\label{subsec:mps-tomography}

We now recall the iterated-disentangling approach to MPS tomography, originally due
to \cite{Cramer_2010}. We present the original, sequential version because it is the conceptual
starting point for the TTN and general TNS algorithms we develop later. The more recent
parallel learner of \cite{lin2025efficientclosestmatrixproduct} achieves logarithmic
circuit depth and improves the system-size dependence of the original analysis of
\cite{Cramer_2010}. We compare the two guarantees at the end of this subsection.

Let $S\subseteq[n]$, and let $A\subseteq S$. For a positive semidefinite operator $\rho$
on $n$ qudits, we say that a unitary $U$ supported on $S$ disentangles the subsystem
$A$ for $\rho$ if
\begin{equation}
    U\rho U^\dagger
    =
    |0^{|A|}\rangle\langle 0^{|A|}|_A\otimes \sigma_{[n]\setminus A}
\end{equation}
for some positive semidefinite operator $\sigma_{[n]\setminus A}$. We allow $\rho$ and
$\sigma_{[n]\setminus A}$ to be subnormalised. When $|A|=1$, we also say that $U$
disentangles the corresponding qudit.

For an MPS on the path $1-2-\cdots-n$, the basic observation is that, after the first
$i-1$ qudits have been disentangled and projected onto $|0\rangle$, the next active block
\begin{equation}
    S_i=\{i,i+1,\ldots,i+\kappa\},
    \qquad
    \kappa=\max\{1,\lceil\log_d \chi\rceil\},
\end{equation}
has reduced-state rank at most $\chi$. Since $d^\kappa\geq \chi$, a unitary supported
on $S_i$ can map the relevant $\chi$-dimensional support into the subspace in which
qudit $i$ is fixed to $|0\rangle$. Repeating this from left to right reduces the state to a
state on the final $\kappa$ qudits, which can then be learned directly.

The following algorithm is written making use of the subnormalised tomography primitive from
\Cref{subsec:subnormalised-tomography}. At each stage, the cumulative postselection
operation is described by a known linear map $K_i$. Operationally, a call to
$\operatorname{Sub\text{-}Tomography}$ with input $K_i$ means that we apply the
known postselection procedure described by $K_i$ to fresh copies of the original state
and keep the successful branch (as in \Cref{lem:subnormalised-tomography-known-mu,cor:subnormalised-tomography-bounded-mu}).
\begin{algorithm}
\caption{LearnMPS}
\label{alg:learn-mps}
\textbf{Input:} Path $P_n$, bond dimension $\chi$, physical dimension $d$, copies of
$|\psi\rangle\in \mathcal S_d(P_n,\chi)$, accuracy parameter $\epsilon$, confidence
parameter $\delta$.

\textbf{Output:} A classical description of a pure state $|\hat\psi\rangle$.
\begin{algorithmic}[1]
\State $\kappa\leftarrow \max\{1,\lceil\log_d\chi\rceil\}$.
\State $m\leftarrow \max\{n-\kappa,0\}$.
\State $\eta\leftarrow \epsilon^2/(128n)$.
\State $K_0\leftarrow I$ and $\widetilde\mu_0\leftarrow 1$.
\For{$i=1$ to $m$}
    \State $S_i\leftarrow \{i,i+1,\ldots,i+\kappa\}$. \Comment{Active block.}
    \State
    $\hat\sigma_i
    \leftarrow
    \operatorname{Sub\text{-}Tomography}
    \left(
        K_{i-1},
        S_i,
        \chi,
        \eta,
        \delta/(2n),
        [\widetilde\mu_{i-1},1]
    \right).$
    \State Let $W_i$ be the span of the eigenvectors associated with the $\chi$
largest eigenvalues of $\hat\sigma_i$.
    \State Compute a unitary $\widetilde U_i$ supported on $S_i$ such that
    \begin{equation}
        \widetilde U_i W_i
        \subseteq
        |0\rangle_i\otimes
        \left(\mathbb C^d\right)^{\otimes \kappa}.
    \end{equation}
    \State Let $U_i = \widetilde U_i\otimes I_{[n]\setminus S_i}$.
    \State Let
    \begin{equation}
        P_i
        =
        |0\rangle\langle 0|_i\otimes I_{[n]\setminus\{i\}}.
    \end{equation}
    \State $K_i\leftarrow P_iU_iK_{i-1}$.
    \State $\widetilde\mu_i\leftarrow \widetilde\mu_{i-1}-2\eta$.
\EndFor
\State $R\leftarrow \{m+1,m+2,\ldots,n\}$.
\State $\hat\tau
\leftarrow
\operatorname{Sub\text{-}Tomography}
\left(K_m,R,1,\eta,\delta/(2n),[\widetilde\mu_m,1] \right).$
\State Let $|\hat\varphi\rangle$ be a top eigenvector of $\hat\tau$.
\State \Return
\begin{equation}
    |\hat\psi\rangle
    =
    U_1^\dagger U_2^\dagger\cdots U_m^\dagger
    \left(
        |0^m\rangle\otimes |\hat\varphi\rangle
    \right).
\end{equation}
\end{algorithmic}
\end{algorithm}

The required unitary exists because
$\dim(W_i)\leq\chi\leq d^\kappa$, while the target subspace has dimension
$d^\kappa$. Choose an orthonormal basis of $W_i$, map it isometrically into the
target subspace, and extend the isometry to a unitary on $S_i$.

We next establish the results used to analyse \Cref{alg:learn-mps}. The first
gives an exact disentangler from a rank bound.

\begin{claim}[Exact disentangling from a rank bound]
\label{claim:exact-rank-disentangling}
Let $\rho=|\psi\rangle\langle\psi|$ be a pure $n$-qudit state, let $S\subseteq[n]$,
and suppose that
\begin{equation}
    \operatorname{rank}(\rho_S)=R.
\end{equation}
Then, for every integer
\begin{equation}
    0\leq a\leq |S|-\lceil\log_d R\rceil,
\end{equation}
there exists a unitary supported on $S$ that disentangles some chosen $a$ qudits in
$S$ for $\rho$. In particular, one can disentangle
\begin{equation}
    \max\{|S|-\lceil\log_d R\rceil,0\}
\end{equation}
qudits in $S$.
\end{claim}

\begin{proof}
Let $a\leq |S|-\lceil\log_d R\rceil$. After relabelling the qudits in $S$, we may assume
that the $a$ qudits to be disentangled are the first $a$ qudits of $S$. Since
\begin{equation}
    d^{|S|-a}\geq R,
\end{equation}
there is an isometry from the support of $\rho_S$ into the subspace
\begin{equation}
    |0^a\rangle\otimes
    \left(\mathbb C^d\right)^{\otimes(|S|-a)}.
\end{equation}
Extend this isometry to a unitary $U$ on the Hilbert space of the qudits in $S$. Taking a Schmidt
decomposition of $|\psi\rangle$ across the bipartition $S\mid[n]\setminus S$, the
Schmidt vectors on $S$ span the support of $\rho_S$. Hence, applying $U$ maps
all Schmidt vectors on $S$ into the subspace in which the first $a$ qudits are fixed to
$|0^a\rangle$. Therefore,
\begin{equation}
    (U\otimes I_{[n]\setminus S})|\psi\rangle
    =
    |0^a\rangle\otimes|\nu\rangle
\end{equation}
for some state $|\nu\rangle$ on the remaining $n-a$ qudits. This proves the claim.
\end{proof}

Combining this claim with the rank bound from \Cref{claim:cut-rank-bound} gives the
following disentangling corollary for tensor network states.

\begin{corollary}[Disentangling from a TN cut]
\label{cor:tn-cut-disentangling}
Let $G=([n],E)$, let $w:E\to\mathbb N$, let $|\psi\rangle\in\mathcal S_d(G,w)$, and let $S\subseteq[n]$. Then one can construct a unitary supported on $S$ that disentangles at least
\begin{equation}
    \max
    \left\{
        |S|-
        \left\lceil
            \sum_{e\in\cut_G(S)}\log_d w(e)
        \right\rceil,
        0
    \right\}
\end{equation}
qudits in $S$ for $|\psi\rangle$. In particular, for every state in $\mathcal S_d(G,\chi)$, one can disentangle at least
\begin{equation}
    \max
    \left\{
        |S|-\left\lceil |\cut_G(S)|\log_d\chi\right\rceil,
        0
    \right\}
\end{equation}
qudits in $S$.
\end{corollary}

\begin{proof}
By \Cref{claim:cut-rank-bound},
\begin{equation}
    \operatorname{rank}(\rho_S)
    \leq
    \prod_{e\in\cut_G(S)}w(e).
\end{equation}
Therefore,
\begin{equation}
    \left\lceil\log_d\operatorname{rank}(\rho_S)\right\rceil
    \leq
    \left\lceil
        \sum_{e\in\operatorname{cut}(S)}\log_d w(e)
    \right\rceil.
\end{equation}
The result follows from \Cref{claim:exact-rank-disentangling}.
\end{proof}

The next lemma is the perturbation statement that we use for approximate
disentanglers. It is a subnormalised version of the top-eigenspace argument used in
\cite[Lemma B.1]{lin2025efficientclosestmatrixproduct}.

\begin{lemma}[Approximate top-eigenspace overlap]
\label{lem:top-eigenspace-stability}
Let $\tau$ and $\hat\tau$ be positive semidefinite operators on a finite-dimensional
Hilbert space. Suppose that
\begin{equation}
    \operatorname{rank}(\tau)\leq r
    \qquad\text{and}\qquad
    \|\tau-\hat\tau\|_1\leq \eta.
\end{equation}
Let $W$ be the span of the eigenvectors associated with the $r$ largest
eigenvalues of $\hat\tau$, and let
$\Pi_W$ be the orthogonal projector onto $W$. Then
\begin{equation}
    \operatorname{tr}\left[(I-\Pi_W)\tau\right]\leq 2\eta.
\end{equation}
\end{lemma}

\begin{proof}
If $r$ is at least the Hilbert-space dimension, the claim is trivial. Otherwise, let
$\hat\tau_{(r)}$ be the truncation of $\hat\tau$ to the eigenspaces associated
with its $r$ largest eigenvalues, i.e.,
\begin{equation}
    \hat\tau_{(r)}=\Pi_W\hat\tau\Pi_W.
\end{equation}
Since $\Pi_W$ is a spectral projector of $\hat\tau$, we have
$\hat\tau-\hat\tau_{(r)}=(I-\Pi_W)\hat\tau(I-\Pi_W)\succeq0$ and hence
$\operatorname{tr}[(I-\Pi_W)\hat\tau]=\|\hat\tau-\hat\tau_{(r)}\|_1$.
By the Eckart--Young--Mirsky theorem, $\hat\tau_{(r)}$ is a best rank-$r$ approximation
to $\hat\tau$ in trace norm. Since $\tau$ has rank at most $r$,
\begin{equation}
    \|\hat\tau-\hat\tau_{(r)}\|_1
    \leq
    \|\hat\tau-\tau\|_1
    \leq
    \eta.
\end{equation}
Now,
\begin{align}
    \operatorname{tr}\left[(I-\Pi_W)\tau\right]
    &=
    \operatorname{tr}\left[(I-\Pi_W)(\tau-\hat\tau)\right]
    +
    \operatorname{tr}\left[(I-\Pi_W)\hat\tau\right] \\
    &\leq
    \|\tau-\hat\tau\|_1
    +
    \|\hat\tau-\hat\tau_{(r)}\|_1 \\
    &\leq
    2\eta.
\end{align}
This proves the lemma.
\end{proof}

The next corollary gives the form of
\Cref{lem:top-eigenspace-stability} used in the iterative algorithms. It shows
that a unitary which maps the top $r$-dimensional eigenspace of the estimate
into the subspace where the state on subsystem $A$ is fixed to $|0^{|A|}\rangle$ can be followed by
postselection with a loss of at most $2\eta$ in squared branch norm.

\begin{corollary}[Approximate disentangling of a subnormalised branch]
\label{cor:approximate-disentangling-subnormalised}
Let
$|\Psi\rangle\in\mathcal H_S\otimes\mathcal H_B$
be a subnormalised pure state, and write
\begin{equation}
    \mu=\langle\Psi|\Psi\rangle,
    \qquad
    \tau_S
    =
    \operatorname{tr}_B
    \left[
        |\Psi\rangle\langle\Psi|
    \right].
\end{equation}
Assume that
\begin{equation}
    \mu>0,
    \qquad
    \operatorname{rank}(\tau_S)\leq r,
\end{equation}
and let $\hat\tau_S\geq0$ satisfy
\begin{equation}
    \|\tau_S-\hat\tau_S\|_1\leq\eta.
\end{equation}
Let $W$ be the span of the eigenvectors associated with the $r$ largest
eigenvalues of $\hat\tau_S$. Let $A\subseteq S$ satisfy
\begin{equation}
    d^{|S|-|A|}\geq r,
\end{equation}
and let $U$ be a unitary on $\mathcal H_S$ such that
\begin{equation}
    UW
    \subseteq
    |0^{|A|}\rangle_A
    \otimes
    \left(\mathbb C^d\right)^{\otimes(|S|-|A|)}.
\end{equation}
Define
\begin{equation}
    P_A
    =
    |0^{|A|}\rangle\langle0^{|A|}|_A
    \otimes
    I_{S\setminus A}.
\end{equation}
Then
\begin{equation}
\label{eq:disentangling-postselection-norm-bound}
    \left\|
        (P_AU\otimes I_B)|\Psi\rangle
    \right\|^2
    \geq
    \mu-2\eta.
\end{equation}

For the normalised state
\begin{equation}
    |\psi\rangle
    =
    \frac{|\Psi\rangle}{\sqrt\mu},
\end{equation}
let
\begin{equation}
    p
    =
    \left\|
        (P_AU\otimes I_B)|\psi\rangle
    \right\|^2.
\end{equation}
Then
\begin{equation}
    p
    \geq
    1-\frac{2\eta}{\mu}.
\end{equation}
If $\eta<\mu/2$, the normalised postselected state
\begin{equation}
    |\psi'\rangle
    =
    \frac{
        (P_AU\otimes I_B)|\psi\rangle
    }{
        \sqrt p
    }
\end{equation}
satisfies
\begin{equation}
    \left\|
        (U\otimes I_B)
        |\psi\rangle\langle\psi|
        (U^\dagger\otimes I_B)
        -
        |\psi'\rangle\langle\psi'|
    \right\|_1
    \leq
    2\sqrt{\frac{2\eta}{\mu}}.
\end{equation}
\end{corollary}

\begin{proof}
Let $\Pi_W$ be the orthogonal projector onto $W$, and define
\begin{equation}
    \Pi
    =
    U^\dagger P_AU.
\end{equation}
Since $UW$ lies in the range of $P_A$, the range of $\Pi$ contains $W$.
It follows that
\begin{equation}
    I-\Pi
    \leq
    I-\Pi_W.
\end{equation}
Using \Cref{lem:top-eigenspace-stability}, we obtain
\begin{align}
    \left\|
        (P_AU\otimes I_B)|\Psi\rangle
    \right\|^2
    &=
    \operatorname{tr}\left[\Pi\tau_S\right] \\
    &=
    \mu-
    \operatorname{tr}\left[(I-\Pi)\tau_S\right] \\
    &\geq
    \mu-
    \operatorname{tr}\left[(I-\Pi_W)\tau_S\right] \\
    &\geq
    \mu-2\eta.
\end{align}
Dividing by $\mu$ gives the lower bound on $p$.

If $\eta<\mu/2$, then $p>0$. Since $P_A$ is an orthogonal projector,
\begin{equation}
    \left|
        \langle\psi'|
        (U\otimes I_B)
        |\psi\rangle
    \right|^2
    =
    p.
\end{equation}
The trace norm identity for pure states now gives
\begin{align}
    \left\|
        (U\otimes I_B)
        |\psi\rangle\langle\psi|
        (U^\dagger\otimes I_B)
        -
        |\psi'\rangle\langle\psi'|
    \right\|_1
    &=
    2\sqrt{1-p} \\
    &\leq
    2\sqrt{\frac{2\eta}{\mu}}.
\end{align}
\end{proof}

Equation
\eqref{eq:disentangling-postselection-norm-bound}
is the estimate used in the iterative learners. Each disentangling step removes
at most $2\eta$ from the squared norm of the current branch.

We also need the following rank bound for the active blocks in the MPS learner.

\begin{lemma}[Rank bound for active MPS blocks]
\label{lem:mps-active-block-rank}
Let $|\psi\rangle\in\mathcal S_d(P_n,\chi)$, and let
$\kappa=\max\{1,\lceil\log_d\chi\rceil\}$. Consider the first $i-1$ iterations of
\Cref{alg:learn-mps}, and let
\begin{equation}
    K_{i-1}=P_{i-1}U_{i-1}\cdots P_1U_1
\end{equation}
be the resulting cumulative postselection operation. Let
\begin{equation}
    S_i=\{i,i+1,\ldots,i+\kappa\}.
\end{equation}
Then the subnormalised reduced state
\begin{equation}
    \sigma_i
    =
    \operatorname{tr}_{[n]\setminus S_i}
    \left[
        K_{i-1}|\psi\rangle\langle\psi|K_{i-1}^\dagger
    \right]
\end{equation}
has rank at most $\chi$.
\end{lemma}

\begin{proof}
All unitaries and projections appearing in $K_{i-1}$ are supported on qudits contained
in the prefix
\begin{equation}
    \{1,2,\ldots,i+\kappa-1\}.
\end{equation}
If $i+\kappa<n$, then $K_{i-1}$ acts only on the left side of the MPS cut between qudits $i+\kappa$ and $i+\kappa+1$. The Schmidt rank across this cut is at most $\chi$ for the original MPS, and local operations on one side of a bipartition cannot increase Schmidt rank. After the first $i-1$ successful projections, the qudits $1,\ldots,i-1$ are fixed to $|0^{i-1}\rangle$. Projecting onto a fixed product state and then discarding those projected qudits cannot increase the rank of the reduced state on the remaining active block. Hence the reduced state on $S_i$ has rank at most $\chi$.

If $i+\kappa=n$, then $S_i$ contains every qudit that has not already been projected. Conditioned on the previous successful projections, the global postselected state is pure and factors as $|0^{i-1}\rangle\otimes|\varphi\rangle_{S_i}$. Thus the reduced state on $S_i$ has rank $1\leq\chi$.
\end{proof}

We first prove correctness conditioned on every call to subnormalised tomography
succeeding.

\begin{proposition}[Correctness of the sequential MPS learner]
\label{prop:learn-mps-correctness}
Assume that every call to $\operatorname{Sub\text{-}Tomography}$ in
\Cref{alg:learn-mps} succeeds with trace-norm error at most $\eta$. If
\begin{equation}
    \eta\leq \frac{1}{8n},
\end{equation}
then the state $|\hat\psi\rangle$ output by \Cref{alg:learn-mps} satisfies
\begin{equation}
    \left\|
        |\psi\rangle\langle\psi|
        -
        |\hat\psi\rangle\langle\hat\psi|
    \right\|_1
    \leq
    2\sqrt{2n\eta}+4\sqrt{\eta}.
\end{equation}
In particular, the choice
\begin{equation}
    \eta=\frac{\epsilon^2}{128n}
\end{equation}
is sufficient to make the final trace norm error at most $\epsilon$ for every
$\epsilon\in(0,1]$.
\end{proposition}

\begin{proof}
Let
\begin{equation}
    \mu_i   =   \left\|  K_i|\psi\rangle   \right\|^2
\end{equation}
be the true cumulative success probability after the first $i$ disentangling steps.
We prove by induction that
\begin{equation}
    \mu_i\geq 1-2i\eta.
\end{equation}

The case $i=0$ is immediate because the initial state is normalised. Assume the
claim holds for $i-1$. By \Cref{lem:mps-active-block-rank}, the active reduced
state at step $i$ has rank at most $\chi$. Since the tomography call has trace
norm error at most $\eta$,
\Cref{eq:disentangling-postselection-norm-bound} shows that the next projection
decreases the squared branch norm by at most $2\eta$. Hence
\begin{equation}
    \mu_i
    \geq
    \mu_{i-1}-2\eta
    \geq
    1-2i\eta.
\end{equation}
This completes the induction.

Thus, for $m\leq n$,
\begin{equation}
    \mu_m\geq 1-2m\eta\geq \frac{1}{2},
\end{equation}
where we used $\eta\leq 1/8n$.

Let
\begin{equation}
    U_{\leq m}=U_mU_{m-1}\cdots U_1
    \qquad\text{and}\qquad
    P_{\leq m}=P_mP_{m-1}\cdots P_1.
\end{equation}
Since $U_j$ is supported on $\{j,j+1,\ldots,j+\kappa\}$, it does not act on
qudit $i$ when $j>i$. Hence $P_i$ commutes with every subsequent unitary $U_j$, and
\begin{equation}
    K_m=P_{\leq m}U_{\leq m}.
\end{equation}
Define the exact postselected reconstruction state as
\begin{equation}
    |\psi^\star\rangle
    =
    U_{\leq m}^\dagger
    \frac{P_{\leq m}U_{\leq m}|\psi\rangle}{\sqrt{\mu_m}}.
\end{equation}
Then
\begin{equation}
    |\langle\psi|\psi^\star\rangle|^2=\mu_m,
\end{equation}
and hence
\begin{equation}
    \left\|
        |\psi\rangle\langle\psi|
        -
        |\psi^\star\rangle\langle\psi^\star|
    \right\|_1
    =
    2\sqrt{1-\mu_m}
    \leq
    2\sqrt{2m\eta}
    \leq
    2\sqrt{2n\eta}.
\end{equation}

It remains to account for the final tomography step on the residual system. Since the
branch after $m$ projections lies in the subspace where the first $m$ qudits are equal
to $|0^m\rangle$, there exists a pure residual state $|\varphi\rangle$ on $R$ such that
\begin{equation}
    \frac{P_{\leq m}U_{\leq m}|\psi\rangle}{\sqrt{\mu_m}}
    =
    |0^m\rangle\otimes|\varphi\rangle.
\end{equation}
The ideal subnormalised residual state is
\begin{equation}
    \tau=\mu_m|\varphi\rangle\langle\varphi|.
\end{equation}
The final tomography call outputs $\hat\tau$ with
\begin{equation}
    \|\tau-\hat\tau\|_1\leq\eta .
\end{equation}
Let $|\hat\varphi\rangle$ be a top eigenvector of $\hat\tau$. By the variational
characterisation of the largest eigenvalue,
\begin{align}
    \langle\hat\varphi|\tau|\hat\varphi\rangle
    &\geq
    \langle\hat\varphi|\hat\tau|\hat\varphi\rangle-\eta \\
    &\geq
    \langle\varphi|\hat\tau|\varphi\rangle-\eta \\
    &\geq
    \langle\varphi|\tau|\varphi\rangle-2\eta \\
    &=
    \mu_m-2\eta.
\end{align}
Since $\tau=\mu_m|\varphi\rangle\langle\varphi|$, this gives
\begin{equation}
    |\langle\hat\varphi|\varphi\rangle|^2
    \geq
    1-\frac{2\eta}{\mu_m}
    \geq
    1-4\eta.
\end{equation}
Thus
\begin{equation}
    \left\|
        |\varphi\rangle\langle\varphi|
        -
        |\hat\varphi\rangle\langle\hat\varphi|
    \right\|_1
    \leq
    4\sqrt{\eta}.
\end{equation}
Tensoring with a fixed state and applying the unitary $U_{\leq m}^\dagger$ both preserve trace norm, so the same bound holds between
$|\psi^\star\rangle = U_{\leq m}^\dagger(\ket{0^m}\otimes\ket{\varphi})$ and the actual output $|\hat\psi\rangle = U_{\leq m}^\dagger(\ket{0^m}\otimes\ket{\hat\varphi})$. The triangle inequality gives
\begin{equation}
    \left\|
        |\psi\rangle\langle\psi|
        -
        |\hat\psi\rangle\langle\hat\psi|
    \right\|_1
    \leq
    2\sqrt{2n\eta}+4\sqrt{\eta}.
\end{equation}
The choice $\eta=\epsilon^2/(128n)$ makes the right-hand side at most
$\epsilon$.
\end{proof}

The sample complexity follows from this correctness guarantee.

\begin{theorem}[Sample complexity of the sequential MPS learner]
\label{thm:learn-mps-sample-complexity}
Suppose that \Cref{alg:learn-mps} uses the sample-optimal rank-constrained
subnormalised tomography primitive from \Cref{subsec:subnormalised-tomography}.
Then an unknown state $|\psi\rangle\in\mathcal S_d(P_n,\chi)$ can be learned to trace norm error at most $\epsilon$ with success probability at least $1-\delta$ using
\begin{equation}
    O\left(
        \frac{n^3}{\epsilon^4}
        \left(
            \chi d^{\kappa+1}
            +
            \log(n/\delta)
        \right)
    \right)
\end{equation}
copies, where
\begin{equation}
    \kappa=\max\{1,\lceil\log_d\chi\rceil\}.
\end{equation}
In particular, since $d^\kappa\leq d\chi$, this is at most
\begin{equation}
    O\left(
        \frac{n^3}{\epsilon^4}
        \left(
            d^2\chi^2+\log(n/\delta)
        \right)
    \right).
\end{equation}
The runtime is polynomial in $n$, $d^{\kappa+1}$, $\chi$, $1/\epsilon$, and
$\log(1/\delta)$, assuming the tomography primitive and the eigendecompositions on
the active blocks are implemented in polynomial time in the respective dimension.
\end{theorem}

\begin{proof}
There are at most $n$ calls to $\operatorname{Sub\text{-}Tomography}$: at most $n-1$ local calls and one final residual call. Allocate failure probability $\delta/(2n)$ to each call. By \Cref{prop:learn-mps-correctness}, it suffices to take
\begin{equation}
    \eta=\frac{\epsilon^2}{128n}.
\end{equation}
The lower bounds $\widetilde\mu_i$ used in the algorithm satisfy
\begin{equation}
    \widetilde\mu_i\geq \frac{1}{2}
\end{equation}
for this choice of $\eta$. Hence the dependence on success probability in
\Cref{cor:subnormalised-tomography-bounded-mu} contributes only a constant factor.

For each local disentangling step, the active subsystem has dimension $d^{\kappa+1}$
and rank at most $\chi$. One local call uses
\begin{equation}
O\left(
    \frac{
        \chi d^{\kappa+1}
        +
        \log(n/\delta)
    }{\eta^2}
\right)
\end{equation}
copies. Multiplying by at most $n$ such calls and substituting the value of $\eta$
gives
\begin{equation}
    O\left(
        \frac{n^3}{\epsilon^4}
        \left(
            \chi d^{\kappa+1}
            +
            \log(n/\delta)
        \right)
    \right).
\end{equation}
The final residual tomography call is lower order, since the residual system has dimension
$d^\kappa\leq d^{\kappa+1}$. The runtime statement follows from the fact that all
classical linear-algebra operations are performed on matrices of dimension at most
$d^{\kappa+1}$.
\end{proof}

\Cref{sec:ttn-tomography,subsec:direct-graph-tomography} generalise this
procedure to trees and to arbitrary graphs, and the path case is recovered in
\Cref{rem:relation-to-mps-learner,rem:relation-to-ttn-learner}.

\begin{remark}[Comparison with the logarithmic-depth learner of \cite{lin2025efficientclosestmatrixproduct}]
\label{rem:lch-mps-black-box}
The guarantee of \cite{lin2025efficientclosestmatrixproduct} is stated in
infidelity.  Writing $\alpha$ for their target infidelity, their algorithm
outputs $|\hat\psi\rangle$ with $|\langle\psi|\hat\psi\rangle|^2\geq1-\alpha$
using
\begin{equation}
    N_{\mathrm{LCH}}(\alpha,\delta)
    =
    O\left(
        \frac{\chi^6d^4n^3\log(n/\delta)}{\max\{1,\log_d\chi\}^3\alpha^4}
    \right)
\end{equation}
copies and a reconstruction circuit of depth $O(\log n)$. For pure states,
\begin{equation}
    \left\|
        |\psi\rangle\langle\psi|
        -
        |\hat\psi\rangle\langle\hat\psi|
    \right\|_1
    =
    2\sqrt{1-|\langle\psi|\hat\psi\rangle|^2}.
\end{equation}
Thus trace norm error $\epsilon$ corresponds to $\alpha=\epsilon^2/4$, under which their bound reads
$O(\chi^6d^4n^3\log(n/\delta)/(\max\{1,\log_d\chi\}^3\epsilon^8))$.
Conversely, \Cref{thm:learn-mps-sample-complexity} in the infidelity convention
reads $O\big(\tfrac{n^3}{\alpha^2}(\chi d^{\kappa+1}+\log(n/\delta))\big)$.

The two bounds give different tradeoffs, and neither uniformly dominates the
other. The sequential bound has better dependence on $\chi$ and on the target
accuracy for two reasons. First, every block learned by the sequential procedure
is separated from the unprocessed suffix by one virtual edge, so its rank is at
most $\chi$ and one call costs at most
$\chi d^{\kappa+1}\leq d^2\chi^2$. A middle block in the parallel schedule can
be separated by two virtual edges and can have rank as large as $\chi^2$.
Second, the sequential analysis tracks the subnormalised branch directly. Each
step removes at most $O(\eta)$ probability mass, the losses add, and the square
root relating this loss to trace norm error is taken only at the end.

The sequential reconstruction has depth linear in $n$, while the learner of
\cite{lin2025efficientclosestmatrixproduct} has depth $O(\log n)$. We use the
latter in \Cref{thm:blackbox-mps-learning} when logarithmic depth is required.
\end{remark}

\subsection{TTN state tomography via iterated disentangling}\label{sec:ttn-tomography}

We extend the sequential disentangling method from MPSs to TTNs of known
topology. \Cref{alg:LearnTTN} processes the rooted tree from leaves to root.
After processing a subtree, it retains only a small residual subsystem. The
parent step acts on the parent vertex and the residual subsystems of its
children.

For a rooted tree $T$, let $\operatorname{ch}(u)$ denote the set of children of $u$, and
let $T_u$ be the subtree rooted at $u$, with vertex set $V_u$.  For
$S\subseteq[n]$, write
\begin{equation}
    \mathcal H_S=(\mathbb C^d)^{\otimes |S|},
    \qquad
    |0\rangle_S=|0\rangle^{\otimes |S|}.
\end{equation}
Throughout this subsection we set
\begin{equation}
    \kappa=\max\{1,\lceil\log_d\chi\rceil\}.
\end{equation}
Thus $d^\kappa\geq \chi$.  The residual subsystem left after processing $T_u$ will be
denoted by $R_u\subseteq V_u$, and it will always satisfy $|R_u|\leq \kappa$.

\begin{algorithm}
\caption{LearnTTN}\label{alg:LearnTTN}
\textbf{Input:} Tree $T$ on vertex set $[n]$, bond dimension $\chi$, physical dimension
$d\geq2$, copies of $|\psi\rangle\in\mathcal S_d(T,\chi)$, accuracy parameter
$\epsilon\in(0,1]$, confidence parameter $\delta\in(0,1)$.

\textbf{Output:} A classical description of a pure state $|\hat\psi\rangle$.
\begin{algorithmic}[1]
\State If $n=1$, root $T$ at its only vertex.  If $n\geq2$, choose a vertex $r$ with
$\deg_T(r)=1$ and root $T$ at $r$.
\State Let $h$ be the height of the rooted tree.
\State $\kappa\leftarrow \max\{1,\lceil\log_d\chi\rceil\}$ and
$\eta\leftarrow \epsilon^2/(128n)$.
\State $K_0\leftarrow I$, $\widetilde\mu_0\leftarrow1$, $Q_0\leftarrow\emptyset$, and
$c\leftarrow0$.
\For{$i=h-1$ down to $0$}
    \ForAll{$u\in V(T)$ with $\operatorname{depth}(u)=i$}
        \State $S_u\leftarrow \{u\}\cup\bigcup_{v\in\operatorname{ch}(u)}R_v$.
        \If{$|S_u|\leq\kappa$}
            \State $R_u\leftarrow S_u$.
        \Else
            \State Choose $R_u\subseteq S_u$ such that $u\in R_u$ and $|R_u|=\kappa$.
            \State $Q_u\leftarrow S_u\setminus R_u$.
            \State $\hat\sigma_u
            \leftarrow
            \operatorname{Sub\text{-}Tomography}
            \left(
                K_c,
                S_u,
                \chi,
                \eta,
                \delta/(2n),
                [\widetilde\mu_c,1]
            \right).$
            \State Let $W_u$ be the span of the eigenvectors associated with the $\chi$
largest eigenvalues of $\hat\sigma_u$.
            \State Compute a unitary $\widetilde U_u$ supported on $S_u$ such that
            \begin{equation}
                \widetilde U_u W_u
                \subseteq
                |0^{|Q_u|}\rangle_{Q_u}\otimes\mathcal H_{R_u}.
            \end{equation}
            \State Let $U_{c+1} = \widetilde U_u\otimes I_{[n]\setminus S_u}$.
            \State Let
            \begin{equation}
                P_{c+1}
                =
                |0^{|Q_u|}\rangle\langle0^{|Q_u|}|_{Q_u}\otimes I_{[n]\setminus Q_u}.
            \end{equation}
            \State $K_{c+1}\leftarrow P_{c+1}U_{c+1}K_c$.
            \State $\widetilde\mu_{c+1}\leftarrow\widetilde\mu_c-2\eta$.
            \State $Q_{c+1}\leftarrow Q_c\cup Q_u$ and $c\leftarrow c+1$.
        \EndIf
    \EndFor
\EndFor
\State
$
\hat\tau
\leftarrow
\operatorname{Sub\text{-}Tomography}
\left(
    K_c,
    R_r,
    1,
    \eta,
    \delta/(2n),
    [\widetilde\mu_c,1]
\right).
$
\State Let $|\hat\varphi\rangle$ be a top eigenvector of $\hat\tau$.
\State \Return
\begin{equation}
    |\hat\psi\rangle
    =
    U_1^\dagger U_2^\dagger\cdots U_c^\dagger
    \left(
        |0^{|Q_c|}\rangle_{Q_c}\otimes |\hat\varphi\rangle_{R_r}
    \right),
\end{equation}
where the tensor product is interpreted according to the canonical ordering of the
qudits.
\end{algorithmic}
\end{algorithm}

The required unitary exists because
$\dim(W_u)\leq\chi\leq d^\kappa=\dim(\mathcal H_{R_u})$. Choose an
orthonormal basis of $W_u$, map it isometrically into
$|0^{|Q_u|}\rangle_{Q_u}\otimes\mathcal H_{R_u}$, and extend the isometry to a
unitary on $\mathcal H_{S_u}$.

We next establish the invariants used to analyse \Cref{alg:LearnTTN}.

\begin{lemma}[Residual-subsystem invariant]\label{lem:ttn-residual-subsystem-invariant}
During \Cref{alg:LearnTTN}, after a vertex $u$ has been processed, the following hold:
\begin{enumerate}
    \item $R_u\subseteq V_u$ and $|R_u|\leq\kappa$.
    \item Every qudit in $V_u\setminus R_u$ has been projected onto $|0\rangle$ in the
    current postselected branch.
    \item The sets $Q_u$ introduced at different nontrivial disentangling steps are
    pairwise disjoint.
\end{enumerate}
Moreover, there are at most $n-1$ nontrivial disentangling steps.  Consequently, if
\begin{equation}
    b=\max_{u\in V(T)}|\operatorname{ch}(u)|
\end{equation}
denotes the maximum number of children in the chosen rooting, then every subsystem
$S_u$ on which \Cref{alg:LearnTTN} performs tomography satisfies
\begin{equation}
    |S_u|\leq 1+b\kappa.
\end{equation}
Finally, with the rooting used in the algorithm, $b\leq1$ for $n\leq2$, and
$b\leq\Delta(T)-1$ for $n\geq3$.
\end{lemma}

\begin{proof}
We prove the first two claims by induction over the traversal from leaves to root.  If
$u$ has no children, then $S_u=\{u\}$.  Since $\kappa\geq1$, the algorithm sets
$R_u=\{u\}$ and no projection is applied, so the claims are immediate.

Now suppose that all children of $u$ have already been processed. By the induction
hypothesis, for each child $v$, every qudit in $V_v\setminus R_v$ has already
been projected onto $|0\rangle$. Hence, immediately before processing $u$, the
qudits in $V_u$ that have not yet been projected onto $|0\rangle$ are 
\begin{equation}\label{eq:ttn-unprojected-set}
    S_u=\{u\}\cup\bigcup_{v\in\operatorname{ch}(u)}R_v.
\end{equation}
If $|S_u|\leq\kappa$, the algorithm sets $R_u=S_u$, and the invariant follows.  If
$|S_u|>\kappa$, the algorithm chooses $R_u\subseteq S_u$ with $|R_u|=\kappa$ and
projects $Q_u=S_u\setminus R_u$ onto $|0\rangle$.  After this step, the only qudits that have not already been projected in $V_u$ are those in $R_u$.

At each nontrivial step, $Q_u$ is chosen from qudits that have not yet been
projected onto $|0\rangle$. Hence $Q_u$ is disjoint from all sets projected at
earlier steps, which proves pairwise disjointness.  Each nontrivial step has
$Q_u\neq\emptyset$, and after the root has been processed the residual set $R_r$ is
nonempty.  Since the projected sets are pairwise disjoint, at most $n-1$ nontrivial
disentangling steps can occur.

The size bound follows from $|R_v|\leq\kappa$ for every child $v$:
\begin{equation}
    |S_u|
    \leq
    1+|\operatorname{ch}(u)|\kappa
    \leq
    1+b\kappa
\end{equation}
by \eqref{eq:ttn-unprojected-set}.
If $n=1$, then $b=0$.  If $n=2$, the chosen root has one child and the other vertex has
none, so $b=1$.  If $n\geq3$, the root has one child because it was chosen to have
degree one, and every non-root vertex has at most $\Delta(T)-1$ children.  Since
$\Delta(T)\geq2$ in this case, $b\leq\Delta(T)-1$.
\end{proof}

\begin{lemma}[Rank bound for TTN residual subsystems]\label{lem:ttn-residual-rank-bound}
Let $|\psi\rangle\in\mathcal S_d(T,\chi)$, and consider \Cref{alg:LearnTTN}.  Immediately
before a vertex $u$ is processed, let $K_{<u}$ be the cumulative postselection map
constructed so far, and define
\begin{equation}
    \sigma_u
    =
    \operatorname{tr}_{[n]\setminus S_u}
    \left[
        K_{<u}|\psi\rangle\langle\psi|K_{<u}^\dagger
    \right].
\end{equation}
Then
\begin{equation}
    \operatorname{rank}(\sigma_u)\leq\chi.
\end{equation}
\end{lemma}

\begin{proof}
The cut between $V_u$ and $[n] \setminus V_u$ in the tree $T$ contains at most
one edge: it contains the edge from $u$ to its parent if $u$ is not the root,
and it is empty if $u$ is the root. By
\Cref{claim:cut-rank-bound}, the Schmidt rank of $\ket{\psi}$ across the
bipartition $V_u \mid [n] \setminus V_u$ is at most $\chi$.

All operations appearing in $K_{<u}$ are supported either entirely inside $V_u$
or entirely outside $V_u$. Indeed, operations associated with descendants of $u$
are supported in $V_u$, while any previously processed vertex that is not a
descendant of $u$ lies in a rooted subtree disjoint from $V_u$. No ancestor of
$u$ has been processed yet. Thus $K_{<u}$ is local with respect to the
bipartition $V_u \mid [n] \setminus V_u$. Local linear maps cannot increase
Schmidt rank, so the subnormalised reduced state on $V_u$ after applying
$K_{<u}$ has rank at most $\chi$.

By the second item of \Cref{lem:ttn-residual-subsystem-invariant}, all qudits in
$V_u \setminus S_u$ have already been projected onto the product state
$\ket{0}_{V_u \setminus S_u}$. Removing a fixed product factor cannot increase
rank. Hence the reduced state on $S_u$ has rank at most $\chi$.
\end{proof}

\begin{proposition}[Correctness of TTN learner]\label{prop:learn-ttn-correctness}
Assume that every $\operatorname{Sub\text{-}Tomography}$ call in
\Cref{alg:LearnTTN} succeeds with trace-norm error at most $\eta$.  If
\begin{equation}
    \eta\leq\frac{1}{8n},
\end{equation}
then the state $|\hat\psi\rangle$ output by \Cref{alg:LearnTTN} satisfies
\begin{equation}
    \left\|
        |\psi\rangle\langle\psi|
        -
        |\hat\psi\rangle\langle\hat\psi|
    \right\|_1
    \leq
    2\sqrt{2n\eta}+4\sqrt\eta.
\end{equation}
In particular, the choice $\eta=\epsilon^2/(128n)$ is sufficient to make the final
trace norm error at most $\epsilon$ for every $\epsilon\in(0,1]$.
\end{proposition}

\begin{proof}
Index the nontrivial disentangling steps by $t=1,\ldots,c$ in the order in which they are
performed, and let
\begin{equation}
    \mu_t=\|K_t|\psi\rangle\|^2
\end{equation}
be the true cumulative success probability after the first $t$ such steps.  We first show
that
\begin{equation}
    \mu_t\geq 1-2t\eta
\end{equation}
for all $t$.  The claim is trivial for $t=0$.  Suppose it holds before the next disentangling
step, which processes some vertex $u$.  By \Cref{lem:ttn-residual-rank-bound}, the
ideal subnormalised reduced state on $S_u$ has rank at most $\chi$.  Since the
corresponding tomography call succeeds to trace-norm error at most $\eta$,
\Cref{cor:approximate-disentangling-subnormalised} implies that the projection chosen
by the algorithm decreases the branch weight by at most $2\eta$.  Hence
\begin{equation}
    \mu_{t+1}\geq \mu_t-2\eta\geq 1-2(t+1)\eta.
\end{equation}
There are at most $n-1$ nontrivial disentangling steps, so
\begin{equation}
    \mu_c\geq1-2n\eta\geq\frac12.
\end{equation}

Let
\begin{equation}
    U_{\leq c}=U_cU_{c-1}\cdots U_1.
\end{equation}
Once a qudit has been projected, it is never included in any subsequent residual
subsystem. Hence a unitary applied at a subsequent step acts trivially on every
previously projected qudit, so each projection commutes with all unitaries that
follow it. Thus
\begin{equation}
    K_c=P_{Q_c}U_{\leq c},
    \qquad
    P_{Q_c}=|0^{|Q_c|}\rangle\langle0^{|Q_c|}|_{Q_c}\otimes I_{R_r},
\end{equation}
where $Q_c$ is the final value of the projected set maintained by \Cref{alg:LearnTTN}.
Define the exact postselected reconstruction
\begin{equation}
    |\psi^\star\rangle    =    U_{\leq c}^\dagger\frac{P_{Q_c}U_{\leq c}|\psi\rangle}{\sqrt{\mu_c}}.
\end{equation}
Then
\begin{equation}
    |\langle\psi|\psi^\star\rangle|^2=\mu_c,
\end{equation}
and
\begin{equation}
    \left\|
        |\psi\rangle\langle\psi|
        -
        |\psi^\star\rangle\langle\psi^\star|
    \right\|_1
    =
    2\sqrt{1-\mu_c}
    \leq
    2\sqrt{2n\eta}.
\end{equation}

By the second item of \Cref{lem:ttn-residual-subsystem-invariant}, applied at
the root, every
qudit in $V_r\setminus R_r=[n]\setminus R_r$ has been projected onto $|0\rangle$ and
hence lies in $Q_c$, so that $[n]\setminus R_r\subseteq Q_c$.  Conversely, once a qudit has been projected, it cannot belong to $R_r$. Hence $Q_c\cap R_r=\emptyset, Q_c\cup R_r=[n]$. Hence, the normalised postselected state has the form
\begin{equation}
    \frac{P_{Q_c}U_{\leq c}|\psi\rangle}{\sqrt{\mu_c}}
    =
    |0^{|Q_c|}\rangle_{Q_c}\otimes|\varphi\rangle_{R_r}
\end{equation}
for some pure state $|\varphi\rangle$ on $R_r$.  The ideal final subnormalised residual
state is
\begin{equation}
    \tau
    =
    \operatorname{tr}_{[n]\setminus R_r}
    \left[
        K_c|\psi\rangle\langle\psi|K_c^\dagger
    \right]
    =
    \mu_c|\varphi\rangle\langle\varphi|.
\end{equation}
The final tomography call outputs $\hat\tau$ with
\begin{equation}
    \|\tau-\hat\tau\|_1\leq\eta.
\end{equation}
Let $|\hat\varphi\rangle$ be a top eigenvector of $\hat\tau$.  By the variational
characterisation of the largest eigenvalue,
\begin{align}
    \langle\hat\varphi|\tau|\hat\varphi\rangle
    &\geq
    \langle\hat\varphi|\hat\tau|\hat\varphi\rangle-\eta \\
    &\geq
    \langle\varphi|\hat\tau|\varphi\rangle-\eta \\
    &\geq
    \langle\varphi|\tau|\varphi\rangle-2\eta \\
    &=
    \mu_c-2\eta.
\end{align}
Since $\tau=\mu_c|\varphi\rangle\langle\varphi|$, this implies
\begin{equation}
    |\langle\hat\varphi|\varphi\rangle|^2
    \geq
    1-\frac{2\eta}{\mu_c}
    \geq
    1-4\eta.
\end{equation}
Therefore,
\begin{equation}
    \left\|
        |\varphi\rangle\langle\varphi|
        -
        |\hat\varphi\rangle\langle\hat\varphi|
    \right\|_1
    \leq
    4\sqrt\eta.
\end{equation}
Tensoring with a fixed state and applying the unitary $U_{\leq c}^\dagger$ preserves trace norm, so the same bound holds between
$|\psi^\star\rangle$ and the output $|\hat\psi\rangle$.  The triangle inequality gives
\begin{equation}
    \left\|
        |\psi\rangle\langle\psi|
        -
        |\hat\psi\rangle\langle\hat\psi|
    \right\|_1
    \leq
    2\sqrt{2n\eta}+4\sqrt\eta.
\end{equation}
Finally, substituting $\eta=\epsilon^2/(128n)$ gives a bound at most $\epsilon$ for
$\epsilon\in(0,1]$.
\end{proof}

The same argument gives the sample complexity of TTN state tomography.

\begin{theorem}[TTN state tomography]\label{thm:ttn-tomography}
Let $T$ be a tree on $n\geq1$ vertices with maximum degree $\Delta=\Delta(T)$, and let
$|\psi\rangle\in\mathcal S_d(T,\chi)$ be unknown, where $d\geq2$.  Given $T$, $d$, and
$\chi$, and implementing $\operatorname{Sub\text{-}Tomography}$ with the sample-optimal
rank-constrained primitive from \Cref{subsec:subnormalised-tomography},
\Cref{alg:LearnTTN} learns $|\psi\rangle$ to trace norm error at most $\epsilon$ with
success probability at least $1-\delta$ using
\begin{equation}
    O\left(
        \frac{n^3}{\epsilon^4}
        \left(
            \chi d^{1+b\kappa}
            +
            \log(n/\delta)
        \right)
    \right)
\end{equation}
copies, where $\kappa=\max\{1,\lceil\log_d\chi\rceil\}$ and
$b=\max_{u\in V(T)}|\operatorname{ch}(u)|$ is computed after the rooting chosen in the
algorithm.  In particular,
\begin{equation}
    \chi d^{1+b\kappa}\leq (d\chi)^{\max\{2,\Delta\}},
\end{equation}
and, for $n\geq3$,
\begin{equation}
    \chi d^{1+b\kappa}\leq (d\chi)^\Delta.
\end{equation}
Thus, for $n\geq3$, the copy complexity is
\begin{equation}
    O\left(
        \frac{n^3}{\epsilon^4}
        \left(
            (d\chi)^\Delta
            +
            \log(n/\delta)
        \right)
    \right).
\end{equation}
The classical runtime is polynomial in $n$, $d^{1+b\kappa}$, $\chi$, $1/\epsilon$, and
$\log(1/\delta)$, assuming the rank constrained tomography procedure and the relevant
eigendecompositions are implemented in time polynomial in their input dimension.
\end{theorem}

\begin{proof}
By \Cref{prop:learn-ttn-correctness}, it suffices to take
\begin{equation}
    \eta=\frac{\epsilon^2}{128n}.
\end{equation}
There are at most $n-1$ nontrivial disentangling calls and one final residual tomography
call, hence at most $n$ calls to $\operatorname{Sub\text{-}Tomography}$ in total.  Since
each call is assigned failure probability $\delta/(2n)$, a union bound gives total
failure probability at most $\delta/2\leq\delta$.

For this choice of $\eta$, all lower estimates on success probability used by the algorithm satisfy
\begin{equation}
    \widetilde\mu_t=1-2t\eta\geq\frac12
\end{equation}
for every $t\leq n$. The dependence on the lower bound on success probability in \Cref{cor:subnormalised-tomography-bounded-mu} thus contributes only a constant factor.
Each tomography call used to construct a disentangling unitary acts on a subsystem
$S_u$ satisfying
\begin{equation}
    |S_u|\leq 1+b\kappa
\end{equation}
by \Cref{lem:ttn-residual-subsystem-invariant}. The corresponding reduced state has
rank at most $\chi$ by \Cref{lem:ttn-residual-rank-bound}. Hence one such call uses
\begin{equation}
    O\left(
        \frac{\chi d^{1+b\kappa}+\log(n/\delta)}{\eta^2}
    \right)
\end{equation}
copies.  Multiplying by at most $n$ calls and substituting the value of $\eta$ yields
\begin{equation}
    O\left(
        \frac{n^3}{\epsilon^4}
        \left(
            \chi d^{1+b\kappa}
            +
            \log(n/\delta)
        \right)
    \right).
\end{equation}
The final residual tomography call has rank one and acts on $|R_r|\leq\kappa$ qudits.
If $n=1$, then $R_r=\{r\}$, so this call has dimension $d$ and is bounded by the above expression because $\chi\geq1$.  If $n\geq2$, then $b\geq1$ and
$d^\kappa\leq\chi d^{1+b\kappa}$, so the final call is again no larger than the bound.

It remains to simplify the degree dependence.  Since $d^\kappa\leq d\chi$, we have
\begin{equation}
    \chi d^{1+b\kappa}
    \leq
    \chi d(d\chi)^b
    =
    (d\chi)^{b+1}.
\end{equation}
By \Cref{lem:ttn-residual-subsystem-invariant}, $b+1\leq\max\{2,\Delta\}$ for all
$n\geq1$, and $b+1\leq\Delta$ for $n\geq3$.  This proves the simplified bounds.  The
runtime statement follows because all linear-algebra operations are performed
on matrices of dimension at most $d^{1+b\kappa}$, up to polynomial overhead in the number of vertices and the accuracy parameters.
\end{proof}

\begin{remark}[Relation to the MPS learner]
\label{rem:relation-to-mps-learner}
If $T$ is a path and $n\geq3$, the rooting chosen by \Cref{alg:LearnTTN} is an endpoint
rooting, so $b=1$.  Every tomography call is on a subsystem of size at most
$\kappa+1$, and the copy bound becomes
\begin{equation}
    O\left(
        \frac{n^3}{\epsilon^4}
        \left(
            \chi d^{\kappa+1}
            +
            \log(n/\delta)
        \right)
    \right),
\end{equation}
matching the sequential MPS bound in \Cref{thm:learn-mps-sample-complexity}.  Thus
\Cref{alg:LearnTTN} is the tree analogue of \Cref{alg:learn-mps}, with the path order
replaced by the traversal from leaves to root.
\end{remark}

\begin{remark}
    Vertices at the same depth have disjoint rooted subtrees, and thus the
    operations on the corresponding disjoint sets of qudits commute. This means that parts of \Cref{alg:LearnTTN} can be parallellised.    \Cref{alg:LearnTTN} is nevertheless stated sequentially,
    because this makes the cumulative postselection map $K_c$ and the estimates of
    the success probability identical to the MPS analysis in
    \Cref{subsec:mps-tomography}.
\end{remark}

\subsection{TNS tomography of general graphs via entanglement rerouting}\label{subsec:general-TNS-tomography-black-box}

We now combine the representation results of \Cref{section:rerouting} with the MPS
and TTN learners. Rerouting shows that a state in $\mathcal S_d(G,\chi)$ also
belongs to a simpler TNS class, with larger bond or local physical dimension. We then apply the learner for that class. The input consists only of $G$,
$\chi$, and the promise that such a representation exists. An explicit TN
description of the unknown state is not required.

\paragraph{Tomography via reduction to an MPS.}

\begin{theorem}[Black-box tomography via an MPS representation]
\label{thm:blackbox-mps-learning}
Let $G=([n],E)$ be a graph, and let $\pi:[n]\to[n]$ be a linear ordering of the vertices.  Set
\begin{equation}
    c=\cw(\pi),
    \qquad
    \kappa_c=
    \max\left\{
        1,
        \left\lceil\log_d(\chi^c)\right\rceil
    \right\}.
\end{equation}
Then an unknown state $|\psi\rangle\in\mathcal S_d(G,\chi)$ can be learned to trace norm error at most $\epsilon$ with success probability at least $1-\delta$ using
\begin{equation}
    O\left(
        \frac{n^3}{\epsilon^4}
        \left(
            \chi^c d^{\kappa_c+1}
            +
            \log(n/\delta)
        \right)
    \right)
\end{equation}
copies.
In particular, if $\pi$ is chosen to have optimal cutwidth, then
\begin{equation}
    O\left(
        \frac{n^3}{\epsilon^4}
        \left(
            d^2\chi^{2\cw(G)}
            +
            \log(n/\delta)
        \right)
    \right)
\label{eq:blackbox-mps-sequential-simplified}
\end{equation}
copies suffice.

If the ordering $\pi$ is not supplied, an optimal cutwidth ordering can first be
computed in time $2^{O(\cw(G)^2)}n$.  Apart from this graph preprocessing step, the
classical postprocessing time of the sequential learner is polynomial in
$n$, $d^{\kappa_c+1}$, $\chi^c$, $1/\epsilon$, and $\log(1/\delta)$, assuming the local rank-constrained tomography primitive and eigendecompositions are implemented in time polynomial in their input dimension.
\end{theorem}

\begin{proof}
By \Cref{thm:reroute-to-mps}, the ordering $\pi$ gives a path $P_n^\pi$ such that
\begin{equation}
    |\psi\rangle\in\mathcal S_d(G,\chi)
    \quad\Longrightarrow\quad
    |\psi\rangle\in\mathcal S_d(P_n^\pi,\chi^c).
\end{equation}
After relabelling the qudits according to $\pi$, the conditions of the MPS learner
are satisfied with path $P_n^\pi$ and bond dimension $\chi^c$.  Applying
\Cref{thm:learn-mps-sample-complexity} with bond dimension $\chi^c$ gives a sufficient number of
\begin{equation}
    O\left(
        \frac{n^3}{\epsilon^4}
        \left(
            \chi^c d^{\kappa_c+1}
            +
            \log(n/\delta)
        \right)
    \right)
\end{equation}
copies.  Since $d^{\kappa_c}\leq d\chi^c$, we have
\begin{equation}
    \chi^c d^{\kappa_c+1}
    \leq
    d^2\chi^{2c}.
\end{equation}
Choosing $c=\cw(G)$ proves \eqref{eq:blackbox-mps-sequential-simplified}.
The runtime statement follows from the runtime
of the cutwidth ordering algorithm in \Cref{thm:reroute-to-mps} and the runtime
statement in \Cref{thm:learn-mps-sample-complexity}.  The output is finally relabelled
back to the original ordering of the qudits.
\end{proof}

\begin{remark}[Logarithmic-depth variant]
The reduction in \Cref{thm:blackbox-mps-learning} is modular in the choice of MPS
learner.  If circuit depth is a concern, one may instead invoke the
logarithmic-depth MPS learner of \cite{lin2025efficientclosestmatrixproduct}
(see \Cref{rem:lch-mps-black-box}) with bond dimension $\chi^c$.  This yields a
reconstruction circuit of depth $O(\log n)$, at the cost of a worse copy
complexity of
\begin{equation}
    O\left(
        \frac{
            d^4\chi^{6c}n^3\log(n/\delta)
        }{
            \max\{1,c\log_d\chi\}^3\epsilon^8
        }
    \right).
\end{equation}
\end{remark}

\paragraph{Tomography via reduction to a TTN.}

\begin{theorem}[Black-box tomography via a tree-cut decomposition]
\label{thm:blackbox-tcw-learning}
Let $G=([n],E)$ be a graph, and let $(\mathcal{T},\mathcal X)$ be a tree-cut decomposition of $G$ of width $k\geq1$.  Let $(\hat{\mathcal{T}},\hat{\mathcal X})$ be the tree-indexed
partition obtained from $(\mathcal{T},\mathcal X)$ after removing empty bags as in
\Cref{lem:remove-empty-bags}, and write
\begin{equation}
    m=|V(\hat{\mathcal{T}})|.
\end{equation}
Root $\hat{\mathcal{T}}$ as in \Cref{alg:LearnTTN}, and let
\begin{equation}
    b_{\hat{\mathcal{T}}}=\max_{u\in V(\hat{\mathcal{T}})} |\operatorname{ch}(u)|
\end{equation}
be the maximum number of children in this rooting.  Set
\begin{equation}
    \kappa=\max\{1,\lceil\log_d\chi\rceil\}.
\end{equation}
Then an unknown state $|\psi\rangle\in\mathcal S_d(G,\chi)$ can be learned to trace norm error at most $\epsilon$ with success probability at least $1-\delta$ using
\begin{equation}
    O\left(
        \frac{m^3}{\epsilon^4}
        \left(
            \chi^k d^{k(1+b_{\hat{\mathcal{T}}}\kappa)}
            +
            \log(m/\delta)
        \right)
    \right)
\end{equation}
copies.  Consequently,
\begin{equation}
    O\left(
        \frac{n^3}{\epsilon^4}
        \left(
            (d\chi)^{k\max\{2,\Delta(\hat{\mathcal{T}})\}}
            +
            \log(n/\delta)
        \right)
    \right)
\label{eq:blackbox-tcw-simplified}
\end{equation}
copies suffice.  If $m\geq3$, this simplifies further to
\begin{equation}
    O\left(
        \frac{n^3}{\epsilon^4}
        \left(
            (d\chi)^{k\Delta(\hat{\mathcal{T}})}
            +
            \log(n/\delta)
        \right)
    \right).
\label{eq:blackbox-tcw-simplified-mgeq-three}
\end{equation}

If no tree-cut decomposition is supplied, one may first compute a decomposition of width
$k\leq 2\tcw(G)$ in time $2^{O(\tcw(G)^2\log\tcw(G))}n^2$.  For the tree $\hat{\mathcal{T}}$
obtained from this computed decomposition, \eqref{eq:blackbox-tcw-simplified} becomes
\begin{equation}
    O\left(
        \frac{n^3}{\epsilon^4}
        \left(
            (d\chi)^{2\tcw(G)\max\{2,\Delta(\hat{\mathcal{T}})\}}
            +
            \log(n/\delta)
        \right)
    \right).
\label{eq:blackbox-tcw-approx}
\end{equation}
The classical postprocessing time, after the tree-cut decomposition has been chosen, is
polynomial in $m$, $d^{k(1+b_{\hat{\mathcal{T}}}\kappa)}$, $\chi^k$, $1/\epsilon$, and
$\log(1/\delta)$, assuming the local tomography primitive and eigendecompositions are
implemented in time polynomial in their input dimension.
\end{theorem}

\begin{proof}
By \Cref{thm:reroute-to-ttn}, grouping the qudits according to the nonempty bags of
$\hat{\mathcal X}$ gives a TTN representation on the tree $\hat{\mathcal{T}}$ with local
physical dimension at most
\begin{equation}
    d'=d^k
\end{equation}
and bond dimension at most
\begin{equation}
    \chi'=\chi^k.
\end{equation}
The number of grouped physical subsystems is $m\leq n$.  Two points ensure that
the TTN learner below applies to the grouped state.  First, grouping is merely a
reinterpretation of the tensor factors under the canonical identification
$\bigotimes_{v\in[n]}\mathbb C^d\cong\bigotimes_{t\in V(\hat{\mathcal{T}})}
(\mathbb C^d)^{\otimes|\hat X_t|}$, so the represented state is unchanged.
Second, bags of size smaller than $k$ give local dimensions smaller than $d^k$.
Padding each local Hilbert space to dimension exactly $d^k$ by tensoring with
unused ancilla directions does not affect the state and lets us apply
\Cref{thm:ttn-tomography} with a uniform local dimension $d'=d^k$.

Apply \Cref{thm:ttn-tomography} to this grouped TTN.  The residual-size parameter for
local dimension $d'$ and bond dimension $\chi'$ is
\begin{equation}
    \max\{1,\lceil\log_{d'}\chi'\rceil\}
    =
    \max\{1,\lceil\log_{d^k}\chi^k\rceil\}
    =
    \max\{1,\lceil\log_d\chi\rceil\}
    =
    \kappa.
\end{equation}
The TTN tomography bound from \Cref{thm:ttn-tomography} gives
\begin{equation}
    O\left(
        \frac{m^3}{\epsilon^4}
        \left(
            \chi'(d')^{1+b_{\hat{\mathcal{T}}}\kappa}
            +
            \log(m/\delta)
        \right)
    \right)
    =
    O\left(
        \frac{m^3}{\epsilon^4}
        \left(
            \chi^k d^{k(1+b_{\hat{\mathcal{T}}}\kappa)}
            +
            \log(m/\delta)
        \right)
    \right).
\end{equation}
This proves the first displayed bound.

For the simplified bound, use $d^\kappa\leq d\chi$ to obtain
\begin{equation}
    \chi^k d^{k(1+b_{\hat{\mathcal{T}}}\kappa)}
    \leq
    \chi^k d^k(d\chi)^{kb_{\hat{\mathcal{T}}}}
    =
    (d\chi)^{k(b_{\hat{\mathcal{T}}}+1)}.
\end{equation}
By the rooting convention in \Cref{alg:LearnTTN},
\begin{equation}
    b_{\hat{\mathcal{T}}}+1\leq \max\{2,\Delta(\hat{\mathcal{T}})\},
\label{eq:children-degree-simplification}
\end{equation}
and, if $m\geq3$, then $b_{\hat{\mathcal{T}}}+1\leq\Delta(\hat{\mathcal{T}})$.  Since $m\leq n$,
this proves \eqref{eq:blackbox-tcw-simplified} and
\eqref{eq:blackbox-tcw-simplified-mgeq-three}.  The statement for a computed
approximate tree-cut decomposition follows from the approximation guarantee quoted in
\Cref{thm:reroute-to-ttn}.  Finally, the output of the TTN learner on the grouped
systems is interpreted as an $n$-qudit state by undoing the canonical grouping of the
physical indices.
\end{proof}

\begin{remark}[Choosing the better black-box reduction]
If both an ordering $\pi$ and a tree-cut decomposition are available, one can simply run
the cheaper of the two black-box reductions.  Using the simplified sequential bounds,
this gives the combined guarantee
\begin{equation}
    O\left(
        \frac{n^3}{\epsilon^4}
        \left(
            \min\left\{
                d^2\chi^{2\cw(\pi)},
                (d\chi)^{k\max\{2,\Delta(\hat{\mathcal{T}})\}}
            \right\}
            +
            \log(n/\delta)
        \right)
    \right),
\end{equation}
with $\cw(\pi)$ replaced by $\cw(G)$ if an optimal ordering is used.  If the
tree $\hat{\mathcal{T}}$ obtained after removing empty bags has at least three vertices,
i.e., $m\geq3$, the second term in the minimum can be replaced by
$(d\chi)^{k\Delta(\hat{\mathcal{T}})}$.
\end{remark}

\begin{remark}[Degree of the tree-cut decomposition]
The dependence on $\Delta(\hat{\mathcal{T}})$ is intentionally left explicit.  Tree-cutwidth
controls the size of the bags and the adhesions of the decomposition, but it does
not by itself bound the maximum degree of the decomposition tree, so
$\Delta(\hat{\mathcal{T}})$ has to be treated as a property of the chosen decomposition
rather than of $G$.  Whenever a decomposition of width $k$ with
$\Delta(\hat{\mathcal{T}})\leq f(G)$ can be constructed for some graph parameter $f$,
\eqref{eq:blackbox-tcw-simplified} immediately improves by replacing
$\Delta(\hat{\mathcal{T}})$ with $f(G)$.
\end{remark}

We note that both black-box reductions above rely only on the realisable learners
from \Cref{subsec:mps-tomography,sec:ttn-tomography}. No agnostic learning machinery is
needed under the promise $|\psi\rangle\in\mathcal S_d(G,\chi)$.

\subsection{Direct TNS tomography for graphs}\label{subsec:direct-graph-tomography}

The black-box procedures first replace the tensor-network graph by a path or
tree and then apply the corresponding MPS or TTN learner. This can lose useful information because the learning cost is estimated using the chosen target representation. We instead choose the active subsystems directly from the original graph while retaining the same iterative-disentangling mechanism.

The learner uses a sequence of subsets describing how the full vertex set is
assembled. After processing a subset, it retains only a small residual subsystem. A larger subset is processed using its fresh vertices and the
residual subsystems retained by its children. The rank of the resulting active state is bounded by the cut of the larger subset in the original graph.

The input to the learner is a sequence of subsets that describes how the vertex set is
assembled.  After a subset has been processed, all but a small residual subsystem of that
subset have been disentangled and projected onto $|0\rangle$.  When a larger subset is
processed, the learner only sees the fresh vertices introduced at that step together with
the residual subsystems left by its children.  The rank of the corresponding reduced
state is determined by the cut of the larger subset in the original graph.

\subsubsection{Learning sequences}

We begin by introducing learning sequences, the combinatorial objects used by our direct learner. A learning sequence records an order in which subsets of the vertex set are assembled into the full vertex set. Beyond the subsets $S_i$ themselves, a learning sequence specifies which previously assembled subsets are combined at each step, through the child sets $I_i$, and which vertices are introduced for the first time, through the fresh sets $F_i$. This determines which residual subsystems are combined at each step of the tomography algorithm in \Cref{alg:learn-tns-from-sequence}.

\begin{definition}[Learning sequence]\label{def:learning-sequence}
Let $G=(V,E)$ be a graph.  A learning sequence for $G$ is a
finite family
\begin{equation}
    \mathcal L=(S_i,I_i,F_i)_{i=1}^L
\end{equation}
with the following properties:
\begin{enumerate}
    \item $S_i\subseteq V$ is nonempty, $I_i\subseteq [i-1]$, and $F_i\subseteq V$ for
    every $i\in[L]$.
    \item For every $i\in[L]$,
    \begin{equation}
        S_i
        =
        F_i\,\sqcup\,\bigsqcup_{j\in I_i} S_j,
        \label{eq:learning-sequence-assembly}
    \end{equation}
    where $\sqcup$ indicates that the union is disjoint.
    \item The sets $F_1,\ldots,F_L$ are pairwise disjoint.
    \item The directed graph on vertex set $[L]$ with an edge $i\to j$ whenever
    $j\in I_i$ is a rooted tree with root $L$.
    \item $S_L=V$.
\end{enumerate}
The elements of $I_i$ are indices of earlier steps, not vertices of $G$. They
are the children of $i$ in the dependency tree. The set $F_i$ contains the
vertices first introduced at step $i$. We call $L$ the length of the learning
sequence.
\end{definition}

\begin{remark}[Analogy with contraction sequences]
\label{rem:learning-vs-contraction-sequences}
Learning sequences are the tomography analogue of the contraction sequences of
Markov and Shi~\cite{Markov_2008}. A contraction sequence records how
tensor-network components are merged during simulation. A learning sequence
records how subsystems are assembled and disentangled during tomography.
\Cref{lem:contraction-to-learning-sequence} constructs a learning sequence from
any contraction sequence.
\end{remark}

The condition that the dependency graph in \Cref{def:learning-sequence} is a rooted tree implies that the sets
$S_1,\ldots,S_L$ form a \emph{laminar family}, a standard notion from
combinatorial optimisation \cite{schrijver2003combinatorial}: any two sets in
the family are either disjoint or nested.

\begin{lemma}[Laminarity of learning sequences]
\label{lem:learning-sequence-laminarity}
Let $\mathcal L=(S_i,I_i,F_i)_{i=1}^L$ be a learning sequence.  For any two indices
$i,j\in[L]$, the sets $S_i$ and $S_j$ are either disjoint, or one contains the other.
In particular, if $j<i$, then either $S_j\subseteq S_i$ or
$S_j\cap S_i=\emptyset$.
\end{lemma}

\begin{proof}
Let $D(i)$ be the set containing $i$ and all descendants of $i$ in the rooted tree on
$[L]$ from \Cref{def:learning-sequence}.  We claim that
\begin{equation}
    S_i=\bigsqcup_{\ell\in D(i)}F_\ell .
\end{equation}
This follows by induction from the defining identity
\(S_i=F_i\,\sqcup\,\bigsqcup_{j\in I_i}S_j\).  Since the dependency graph is a
rooted tree, for any two indices $i$ and $j$, the descendant sets $D(i)$ and $D(j)$ are
either disjoint, or one is contained in the other.  In the nested case, say
$D(j)\subseteq D(i)$, the union representation directly gives
\(S_j=\bigsqcup_{\ell\in D(j)}F_\ell\subseteq\bigsqcup_{\ell\in D(i)}F_\ell=S_i\).
In the disjoint case, $D(i)\cap D(j)=\emptyset$, the two unions run over disjoint
index sets, and since the sets $F_1,\ldots,F_L$ are pairwise disjoint, no vertex
can appear in both unions, so $S_i\cap S_j=\emptyset$.  If $j<i$, then
$j$ cannot be an ancestor of $i$, because edges always point from a larger index to a
smaller index.  Hence either $j$ is a descendant of $i$, giving $S_j\subseteq S_i$, or
the two sets are disjoint.
\end{proof}

Laminarity ensures that every earlier postselection map is local with respect
to the cut used at the current step.

For a learning sequence $\mathcal L$, define
\begin{equation}
    r_i
    :=
    \prod_{e\in \cut_G(S_i)} w(e),
    \qquad
    q_i
    :=
    \lceil \log_d r_i\rceil.
    \label{eq:ri-qi-learning-sequence}
\end{equation}
Here the empty product is equal to $1$. The quantity $r_i$ is the rank bound
from \Cref{claim:cut-rank-bound}. The quantity $q_i$ is the number of qudits
needed for a residual subsystem that can support a state of rank at most $r_i$,
as follows from \Cref{cor:tn-cut-disentangling}. In particular, $q_i=0$ when
$r_i=1$, in which case the residual register may be empty. We also define
\begin{equation}
    a_i
    :=
    |F_i|+
    \sum_{j\in I_i} q_j,
    \label{eq:ai-learning-sequence}
\end{equation}
and the weighted learning complexity of $\mathcal L$ by
\begin{equation}
    \lc_{d,w}(\mathcal L)
    :=
    \max_{i\in[L]} (a_i+q_i).
    \label{eq:weighted-learning-complexity}
\end{equation}
Thus $a_i$ bounds the size of the active subsystem at step $i$, while $q_i$
is the size of the residual subsystem retained after that step.

If one only wants to use the uniform bond dimension bound $w(e)\leq\chi$, then one may
replace $r_i$ and $q_i$ by
\begin{equation}
    \bar r_i=\chi^{|\cut_G(S_i)|},
    \qquad
    \bar q_i=
    \lceil |\cut_G(S_i)|\log_d\chi\rceil,
    \qquad
    \bar a_i
    :=
    |F_i|+
    \sum_{j\in I_i}\bar q_j.
    \label{eq:uniform-ri-qi-ai-learning-sequence}
\end{equation}
and define
\begin{equation}
    \lc_{d,\chi}(\mathcal L)
    :=
    \max_{i\in[L]}
    \left(
        \bar a_i+
        \bar q_i
    \right).
    \label{eq:uniform-learning-complexity}
\end{equation}
Then $\lc_{d,w}(\mathcal L)\leq \lc_{d,\chi}(\mathcal L)$.  Finally, define
\begin{equation}
    \lc_{d,\chi}(G)
    :=
    \min_{\mathcal L}\lc_{d,\chi}(\mathcal L),
\end{equation}
where the minimum is over all learning sequences for $G$.  Such a sequence always
exists: for instance, take $L=1$, $S_1=F_1=V$, and $I_1=\emptyset$.

\begin{remark}[Normalising learning sequences]
\label{rem:learning-sequence-normalisation}
The proof of \Cref{lem:learning-sequence-laminarity} gives
\begin{equation}
    S_i=\bigsqcup_{\ell\in D(i)}F_\ell.
\end{equation}
Applying this identity at $i=L$ shows that the fresh sets
partition the vertex set, so that $\sum_{i=1}^{L}|F_i|=n$.  Moreover, we may
assume without loss of generality that every step satisfies $F_i\neq\emptyset$
or $|I_i|\geq2$.  Indeed, a step with $F_i=\emptyset$ and $I_i=\{j\}$ satisfies
$S_i=S_j$.  Deleting this step, while replacing $i$ by $j$ in the child set of
the parent of $i$ (or making $j$ the root if $i=L$), yields a learning sequence
for $G$ with the same sets.  Since $\cut_G(S_i)=\cut_G(S_j)$, this deletion
leaves every remaining quantity $r_k$, $q_k$, and $a_k$ unchanged and removes
the term $a_i+q_i$ from the maximum in \eqref{eq:weighted-learning-complexity},
so it does not increase the learning complexity.  Under this normalisation,
every step either introduces a fresh vertex or has at least two children.
There are at most $n$ steps of the first kind, because the fresh sets are
nonempty and pairwise disjoint.  There are at most $n-1$ steps of the second
kind, because every leaf of the dependency tree has $I_i=\emptyset$ and hence
$F_i=S_i\neq\emptyset$, so the tree has at most $n$ leaves, and a rooted tree
with at most $n$ leaves has at most $n-1$ vertices with two or more children.
Hence $L\leq2n-1$.
\end{remark}

\subsubsection{Tomography using a learning sequence}
\Cref{alg:learn-tns-from-sequence} applies the iterative disentangling procedure
according to a supplied learning sequence. For each processed
set $S_i$, it stores a residual subsystem $R_i\subseteq S_i$ with $|R_i|\leq q_i$.  When
processing $S_i$, the subsystem to be learned is
\begin{equation}
    M_i
    :=
    F_i\cup\bigcup_{j\in I_i}R_j.
    \label{eq:Mi-learning-sequence}
\end{equation}
If $M_i$ already has size at most $q_i$, no tomography is needed at that step and we set
$R_i=M_i$.  Otherwise we learn the reduced state on $M_i$, map the learned
rank-$r_i$ support into a subspace supported on $q_i$ qudits, and project the remaining
qudits onto $|0\rangle$.

\begin{algorithm}
\caption{LearnTNSFromSequence}\label{alg:learn-tns-from-sequence}
\begin{algorithmic}[1]
\Require graph $G=([n],E)$, edge-dimension function $w:E\to\mathbb N$, physical dimension $d\geq2$, learning sequence $\mathcal L=(S_i,I_i,F_i)_{i=1}^L$, copies of $|\psi\rangle\in\mathcal S_d(G,w)$, accuracy parameter $\epsilon\in(0,1]$, confidence parameter $\delta\in(0,1)$
\Ensure Classical description of a pure state $|\hat\psi\rangle$.
\State For every $i\in[L]$, compute $r_i$ and $q_i$ as in
\Cref{eq:ri-qi-learning-sequence}.
\State $\eta\gets \epsilon^2/(128L)$.
\State $K_0\gets I$, $\widetilde\mu_0\gets1$, $Q\gets\emptyset$, and $c\gets0$.
\For{$i=1$ to $L$}
    \State $M_i\gets F_i\cup\bigcup_{j\in I_i}R_j$.
    \If{$|M_i|\leq q_i$}
        \State $R_i\gets M_i$.
    \Else
        \State Choose $R_i\subseteq M_i$ with $|R_i|=q_i$.
        \State $Q_i\gets M_i\setminus R_i$.
        \State
        $
        \hat\sigma_i
        \gets
        \operatorname{Sub\text{-}Tomography}
        \left(
            K_c,
            M_i,
            r_i,
            \eta,
            \delta/(2L),
            [\widetilde\mu_c,1]
        \right).
        $
        \State Let $W_i$ be the span of the eigenvectors associated with the $r_i$
largest eigenvalues of $\hat\sigma_i$.
        \State Compute a unitary $\widetilde U_i$ supported on $M_i$ such that
        \begin{equation*}
            \widetilde U_i W_i
            \subseteq
            |0^{|Q_i|}\rangle_{Q_i}\otimes\mathcal H_{R_i}.
        \end{equation*}
        \State Let $U_{c+1}$ be $\widetilde U_i$ tensored with the identity outside
        $M_i$, $U_{c+1}= \widetilde{U}_i\otimes I_{M_i^c}$.
        \State Let
        \begin{equation*}
            P_{c+1}
            =
            |0^{|Q_i|}\rangle\langle0^{|Q_i|}|_{Q_i}\otimes I_{[n]\setminus Q_i}.
        \end{equation*}
        \State $K_{c+1}\gets P_{c+1}U_{c+1}K_c$.
        \State $\widetilde\mu_{c+1}\gets \widetilde\mu_c-2\eta$.
        \State $Q\gets Q\cup Q_i$ and $c\gets c+1$.
    \EndIf
\EndFor
\State
$
\hat\tau
\gets
\operatorname{Sub\text{-}Tomography}
\left(
    K_c,
    R_L,
    1,
    \eta,
    \delta/(2L),
    [\widetilde\mu_c,1]
\right).
$
\State Let $|\hat\varphi\rangle$ be a top eigenvector of $\hat\tau$.
\State \Return
\begin{equation*}
    |\hat\psi\rangle
    =
    U_1^\dagger U_2^\dagger\cdots U_c^\dagger
    \left(
        |0^{|Q|}\rangle_Q\otimes |\hat\varphi\rangle_{R_L}
    \right),
\end{equation*}
where the tensor product is interpreted according to the canonical ordering of the
qudits.
\end{algorithmic}
\end{algorithm}

When only $G$ and an upper bound $\chi$ on the bond dimension are known, \Cref{alg:learn-tns-from-sequence} uses $\bar r_i$ and $\bar q_i$ in place of $r_i$ and $q_i$. The individual edge dimensions $w(e)$ need not be known.

The unitary in \Cref{alg:learn-tns-from-sequence} exists because $\dim(W_i)\leq r_i\leq d^{q_i}$ and
$\dim(\mathcal H_{R_i})=d^{q_i}$ whenever the nontrivial branch $|M_i|>q_i$ is entered.
As before, one constructs it by choosing an orthonormal basis of $W_i$, mapping this
basis isometrically into $|0^{|Q_i|}\rangle_{Q_i}\otimes\mathcal H_{R_i}$, and extending the
isometry to a unitary on $\mathcal H_{M_i}$. When only the bound $\chi$ is known, the same argument uses $\dim(W_i)\leq\bar r_i\leq d^{\bar q_i}$.

We next establish correctness and sample complexity. 
\begin{lemma}[Residual invariant]
\label{lem:learning-sequence-residual-invariant}
During \Cref{alg:learn-tns-from-sequence}, after step $i$ has been processed, the
following hold:
\begin{enumerate}
    \item $R_i\subseteq S_i$ and $|R_i|\leq q_i$.
    \item Every qudit in $S_i\setminus R_i$ has been projected onto $|0\rangle$ in the
    current postselected branch.
    \item The sets projected at different nontrivial steps are pairwise disjoint.
\end{enumerate}
Moreover, for every $i\in[L]$,
\begin{equation}
    |M_i|
    \leq
    a_i
    =
    |F_i|+
    \sum_{j\in I_i}q_j.
\end{equation}
\end{lemma}

\begin{proof}
We prove the first two claims by induction over $i$.  If $I_i=\emptyset$, then
$M_i=F_i=S_i$.  If $|M_i|\leq q_i$, the algorithm sets $R_i=M_i=S_i$, and there is
nothing to project.  If $|M_i|>q_i$, it chooses $R_i\subseteq M_i$ and projects
$Q_i=M_i\setminus R_i$ onto $|0\rangle$, so the invariant holds.

Now suppose $I_i\neq\emptyset$ and that the invariant has been proved for all children
$j\in I_i$.  By \Cref{eq:learning-sequence-assembly},
\begin{equation}
    S_i
    =
    F_i\,\sqcup\,\bigsqcup_{j\in I_i}S_j.
\end{equation}
For each child $j$, the only qudits in $S_j$ not already projected are those in $R_j$.
Hence, immediately before processing $i$, the qudits in $S_i$ that have not yet
been projected onto $|0\rangle$ are precisely
\begin{equation}
    M_i=F_i\cup\bigcup_{j\in I_i}R_j.
\end{equation}
If $|M_i|\leq q_i$, the algorithm sets $R_i=M_i$.  If $|M_i|>q_i$, it projects
$M_i\setminus R_i$ and leaves only $R_i$.  This proves the first two claims for $i$.

At a nontrivial step, the projected set is chosen from qudits that have not yet
been projected onto $|0\rangle$.  Hence it is disjoint from every set projected earlier.
This proves pairwise disjointness.  Finally, the size bound for $M_i$ follows from
$|R_j|\leq q_j$ for all $j\in I_i$.
\end{proof}

\begin{lemma}[Rank bound for subsystems]
\label{lem:learning-sequence-rank-bound}
Let $|\psi\rangle\in\mathcal S_d(G,w)$, and consider
\Cref{alg:learn-tns-from-sequence}.  Immediately before step $i$ is processed,
let $K_{<i}$ be the cumulative postselection map constructed so far, i.e., the
map $K_c$ for the value of the counter $c$ at that moment.  (Note that
$K_{<i}\neq K_{i-1}$ in general, since trivial steps with $|M_j|\leq q_j$ do not
increment $c$.)  Define
\begin{equation}
    \sigma_i
    =
    \operatorname{tr}_{[n]\setminus M_i}
    \left[
        K_{<i}|\psi\rangle\langle\psi|K_{<i}^\dagger
    \right].
\end{equation}
Then
\begin{equation}
    \operatorname{rank}(\sigma_i)\leq r_i.
\end{equation}
\end{lemma}

\begin{proof}
By \Cref{claim:cut-rank-bound}, the Schmidt rank of $|\psi\rangle$ across the
bipartition
\begin{equation}
    S_i\mid [n]\setminus S_i
\end{equation}
is at most
\begin{equation}
    r_i=\prod_{e\in\cut_G(S_i)}w(e).
\end{equation}
We claim that all operations appearing in $K_{<i}$ are supported either entirely inside
$S_i$ or entirely outside $S_i$.  Indeed, by \Cref{lem:learning-sequence-laminarity}, for every earlier index $j<i$ we have either
$S_j\subseteq S_i$ or $S_j\cap S_i=\emptyset$.  The operation associated with step $j$,
if nontrivial, is supported on $M_j\subseteq S_j$.  Hence it is local with respect to the
bipartition $S_i\mid[n]\setminus S_i$.  Local linear maps cannot increase Schmidt rank
across this bipartition.  The subnormalised reduced state on $S_i$ after
applying $K_{<i}$ has rank at most $r_i$.

By \Cref{lem:learning-sequence-residual-invariant}, all qudits in $S_i\setminus M_i$
have already been projected onto the product state $|0\rangle_{S_i\setminus M_i}$.
Removing a fixed product factor cannot increase rank.  Thus the reduced state on
$M_i$ has rank at most $r_i$.
\end{proof}

\begin{proposition}[Correctness]
\label{prop:learning-sequence-correctness}
Assume that every call to $\operatorname{Sub\text{-}Tomography}$ in
\Cref{alg:learn-tns-from-sequence} succeeds with trace-norm error at most $\eta$.  If
\begin{equation}
    \eta\leq \frac{1}{8L},
\end{equation}
then the output state $|\hat\psi\rangle$ satisfies
\begin{equation}
    \left\|
        |\psi\rangle\langle\psi|
        -
        |\hat\psi\rangle\langle\hat\psi|
    \right\|_1
    \leq
    2\sqrt{2L\eta}+4\sqrt\eta.
\end{equation}
In particular, the choice $\eta=\epsilon^2/(128L)$ is sufficient to make the final trace norm error at most $\epsilon$ for every $\epsilon\in(0,1]$.
\end{proposition}

\begin{proof}
Index the nontrivial disentangling steps by $t=1,\ldots,c$ in the order in which they are
performed, and let
\begin{equation}
    \mu_t=\|K_t|\psi\rangle\|^2
\end{equation}
be the true cumulative success probability after the first $t$ such steps.  We prove by
induction that
\begin{equation}
    \mu_t\geq 1-2t\eta.
\end{equation}
The claim is immediate for $t=0$.  Suppose it holds before the next nontrivial step,
which processes some index $i$.  By \Cref{lem:learning-sequence-rank-bound}, the ideal
subnormalised reduced state on $M_i$ has rank at most $r_i$.  Since the tomography call
succeeds to trace-norm error at most $\eta$,
\Cref{cor:approximate-disentangling-subnormalised} implies that the projection chosen
by the algorithm decreases the branch weight by at most $2\eta$.  Hence
\begin{equation}
    \mu_{t+1}\geq\mu_t-2\eta\geq1-2(t+1)\eta.
\end{equation}
Since there is at most one nontrivial disentangling step per index, $c\leq L$.  Therefore
\begin{equation}
    \mu_c\geq1-2L\eta\geq\frac12.
\end{equation}

Let
\begin{equation}
    U_{\leq c}=U_cU_{c-1}\cdots U_1.
\end{equation}
Once a qudit has been projected, it is excluded from every residual subsystem. Hence no subsequent unitary acts on that qudit, so each
projection commutes with all unitaries applied after it, and
\begin{equation}
    K_c=P_QU_{\leq c},
P_Q=|0^{|Q|}\rangle\langle0^{|Q|}|_Q\otimes I_{R_L},
\end{equation}
where $Q=\bigsqcup_{i:\,|M_i|>q_i}Q_i$ is the final value of the projected set
maintained by \Cref{alg:learn-tns-from-sequence}.
Define the exact postselected reconstruction
\begin{equation}
    |\psi^\star\rangle
    =
    U_{\leq c}^\dagger\frac{P_QU_{\leq c}|\psi\rangle}{\sqrt{\mu_c}}.
\end{equation}
Then $|\langle\psi|\psi^\star\rangle|^2=\mu_c$, and hence
\begin{equation}
    \left\|
        |\psi\rangle\langle\psi|
        -|\psi^\star\rangle\langle\psi^\star|
    \right\|_1
    =
    2\sqrt{1-\mu_c}
    \leq
    2\sqrt{2L\eta}.
\end{equation}

By \Cref{lem:learning-sequence-residual-invariant} applied to the final set $S_L=V$,
we have $Q\cup R_L=[n]$ and $Q\cap R_L=\emptyset$.  Hence the normalised postselected
state has the form
\begin{equation}
    \frac{P_QU_{\leq c}|\psi\rangle}{\sqrt{\mu_c}}
    =
    |0^{|Q|}\rangle_Q\otimes |\varphi\rangle_{R_L}
\end{equation}
for some pure state $|\varphi\rangle$ on $R_L$.  The ideal final subnormalised residual
state is
\begin{equation}
    \tau
    =
    \operatorname{tr}_{[n]\setminus R_L}
    \left[
        K_c|\psi\rangle\langle\psi|K_c^\dagger
    \right]
    =
    \mu_c|\varphi\rangle\langle\varphi|.
\end{equation}
The final tomography call outputs $\hat\tau$ with
$\|\tau-\hat\tau\|_1\leq\eta$.  If $|\hat\varphi\rangle$ is a top eigenvector of
$\hat\tau$, the same variational argument as in
\Cref{prop:learn-mps-correctness,prop:learn-ttn-correctness} gives
\begin{equation}
    |\langle\hat\varphi|\varphi\rangle|^2
    \geq
    1-\frac{2\eta}{\mu_c}
    \geq
    1-4\eta.
\end{equation}
Thus
\begin{equation}
    \left\|
        |\varphi\rangle\langle\varphi|
        -
        |\hat\varphi\rangle\langle\hat\varphi|
    \right\|_1
    \leq
    4\sqrt\eta.
\end{equation}
Applying $U_{\leq c}^\dagger$ preserves trace norm, and the triangle inequality proves the bound.  Substituting $\eta=\epsilon^2/(128L)$ gives a bound at most $\epsilon$
for $\epsilon\in(0,1]$.
\end{proof}

\begin{theorem}[Tomography using a learning sequence]
\label{thm:learning-sequence-tomography}
Let $G=([n],E)$ be a graph, let $w:E\to\mathbb N$, and let
$\mathcal L=(S_i,I_i,F_i)_{i=1}^L$ be a learning sequence for $G$. Given $G$, $w$,
$d$, and $\mathcal L$, \Cref{alg:learn-tns-from-sequence} learns any unknown
$|\psi\rangle\in\mathcal S_d(G,w)$ to trace norm error at most $\epsilon$ with
success probability at least $1-\delta$ using
\begin{equation}
    O\left(
        \frac{L^3}{\epsilon^4}
        \left(
            \max_{i\in[L]} r_i d^{a_i}
            +
            \log(L/\delta)
        \right)
    \right)
    \label{eq:learning-sequence-precise-bound}
\end{equation}
copies, where $r_i$ and $a_i$ are defined in
\Cref{eq:ri-qi-learning-sequence,eq:ai-learning-sequence}. In particular,
\begin{equation}
    O\left(
        \frac{L^3}{\epsilon^4}
        \left(
            d^{\lc_{d,w}(\mathcal L)}
            +
            \log(L/\delta)
        \right)
    \right)
    \label{eq:learning-sequence-weighted-bound}
\end{equation}
copies suffice.

If the input is only promised to lie in $\mathcal S_d(G,\chi)$, the algorithm instead uses $\bar r_i$, $\bar q_i$, and $\bar a_i$ from \Cref{eq:uniform-ri-qi-ai-learning-sequence}. It uses
\begin{equation}
    O\left(
        \frac{L^3}{\epsilon^4}
        \left(
            \max_{i\in[L]} \bar r_i d^{\bar a_i}
            +
            \log(L/\delta)
        \right)
    \right)
    \label{eq:learning-sequence-uniform-precise-bound}
\end{equation}
copies. In particular,
\begin{equation}
    O\left(
        \frac{L^3}{\epsilon^4}
        \left(
            d^{\lc_{d,\chi}(\mathcal L)}
            +
            \log(L/\delta)
        \right)
    \right)
    \label{eq:learning-sequence-uniform-bound}
\end{equation}
copies suffice.

The classical postprocessing time is polynomial in $L$, the largest local Hilbert-space
dimension, the largest rank parameter, $1/\epsilon$, and $\log(1/\delta)$, assuming
that the rank-constrained tomography primitive and the relevant eigendecompositions
are implemented in time polynomial in their input dimension.
\end{theorem}

\begin{proof}
By \Cref{prop:learning-sequence-correctness}, it suffices to set
\begin{equation}
    \eta=\frac{\epsilon^2}{128L}.
\end{equation}
There are at most $L$ nontrivial disentangling calls and one final residual tomography
call.  Since each call is assigned failure probability $\delta/(2L)$, a union bound gives
total failure probability at most $\delta$.

For this choice of $\eta$, all lower estimates for success probability used by the algorithm
satisfy
\begin{equation}
    \widetilde\mu_t=1-2t\eta\geq \frac12
\end{equation}
for all $t\leq L$.  Hence the dependence on the success probability in
\Cref{cor:subnormalised-tomography-bounded-mu} contributes only a constant factor.

For a nontrivial disentangling call at step $i$, \Cref{lem:learning-sequence-rank-bound}
gives rank at most $r_i$, and \Cref{lem:learning-sequence-residual-invariant} gives
$|M_i|\leq a_i$.  Thus this call uses
\begin{equation}
    O\left(
        \frac{r_i d^{a_i}+\log(L/\delta)}{\eta^2}
    \right)
\end{equation}
copies.  The final residual tomography call has rank one and acts on $|R_L|\leq q_L$
qudits. It is bounded by the same maximum because $R_L\subseteq M_L$ if no final
disentangling occurs, while $|R_L|=q_L\leq a_L$ whenever a final disentangling occurs.
Multiplying by at most $L+1$ calls and substituting the value of $\eta$ gives
\eqref{eq:learning-sequence-precise-bound}.

Since $r_i\leq d^{q_i}$, we have
\begin{equation}
    r_i d^{a_i}
    \leq
    d^{a_i+q_i}
    \leq
    d^{\lc_{d,w}(\mathcal L)}.
\end{equation}
This proves \eqref{eq:learning-sequence-weighted-bound}.

Now suppose that $|\psi\rangle\in\mathcal S_d(G,\chi)$. By definition, there exists
a weight function $w_\psi:E\to[\chi]$ such that
$|\psi\rangle\in\mathcal S_d(G,w_\psi)$. For every $i\in[L]$,
\begin{equation}
    \prod_{e\in\cut_G(S_i)} w_\psi(e)
    \leq
    \chi^{|\cut_G(S_i)|}
    =
    \bar r_i.
\end{equation}
Hence the rank bound used at step $i$ is at most $\bar r_i$, a residual register
of size $\bar q_i$ suffices, and the active subsystem has size at most $\bar a_i$.
Repeating the preceding argument with these parameters proves
\eqref{eq:learning-sequence-uniform-precise-bound}. Since
$\bar r_i\leq d^{\bar q_i}$, we also have

\begin{equation}
    \bar r_i d^{\bar a_i}
    \leq
    d^{\bar a_i+\bar q_i}
    \leq
    d^{\lc_{d,\chi}(\mathcal L)},
\end{equation}
which proves \eqref{eq:learning-sequence-uniform-bound}. The runtime statement follows
because all linear-algebra operations are performed on matrices whose dimension and
rank are bounded by the corresponding parameters defined from the individual edge dimensions or from the common bound $\chi$, up to
polynomial overhead in $L$ and the accuracy parameters.
\end{proof}

\subsubsection{Constructing learning sequences from contraction sequences}
\label{subsec:contraction-to-learning}

Learning sequences can be supplied directly, but they can also be obtained from graph
contraction sequences.  We use the following standard interpretation of a contraction
sequence: during the contraction process, every current vertex represents a subset of the
original vertex set, and the degree of a current vertex counts incident edges with
multiplicity.  With this convention, the degree of the current vertex representing a set
$S\subseteq V$ is $|\cut_G(S)|$.

\begin{lemma}[From contraction sequences to learning sequences]
\label{lem:contraction-to-learning-sequence}
Let $G=(V,E)$ be a connected graph with $n\geq2$, and let $\pi$ be a contraction sequence of
contraction complexity $C$.  Then $\pi$ induces a learning sequence
$\mathcal L_\pi=(S_i,I_i,F_i)_{i=1}^L$ for $G$ with
\begin{equation}
    L\leq n-1
\end{equation}
and
\begin{equation}
    |\cut_G(S_i)|\leq C
    \qquad
    \text{for every }i\in[L].
\end{equation}
Moreover, the sequence can be constructed from $\pi$ in polynomial time.
\end{lemma}

\begin{proof}
Run the contraction process described by $\pi$.  At every moment, each current vertex
corresponds to the subset of original vertices that have been merged into it.  We build
the learning sequence alongside this process.  When the next contraction merges two
distinct current vertices corresponding to disjoint sets $A$ and $B$, create a new index
$i$.  If $A$ was created by an earlier contraction, include its corresponding index in
$I_i$. Otherwise, $A$ is a singleton original vertex and we include that vertex in $F_i$.
Do the same for $B$.  Finally set
\begin{equation}
    S_i=A\cup B.
\end{equation}
Because each nontrivial contraction reduces the number of current vertices by one, a
connected $n$-vertex graph has at most $n-1$ such contractions, and the final set is
$V$.  Each earlier constructed set is used once, namely when its current vertex is
merged into a larger current vertex.  The fresh vertices are introduced when
their singleton current vertices are first merged.  Hence the data
$(S_i,I_i,F_i)_{i=1}^L$ satisfy the conditions of \Cref{def:learning-sequence}.

Immediately after the contraction creating $S_i$, the incident edges of the new current
vertex are precisely the original edges crossing the cut $\cut_G(S_i)$, counted with
multiplicity.  Since $\pi$ has contraction complexity $C$, this degree is at most $C$.
The construction only requires maintaining the current partition of $V$ and the index
associated with each current part, so it is polynomial-time in the length of $\pi$ and the
size of $G$.
\end{proof}

\begin{remark}[Disconnected graphs]
The tomography theorem for a fixed learning sequence (\Cref{thm:learning-sequence-tomography}) applies to arbitrary graphs. For a disconnected graph, one may construct a learning sequence for each connected component, introduce isolated vertices as fresh systems, and combine the component sequences in a final step. We state the contraction-complexity corollaries below for connected graphs to keep the bounds and notation simple.
\end{remark}

The conversion of \Cref{lem:contraction-to-learning-sequence} also yields an upper
bound on the learning complexity in terms of contraction complexity, and hence, via
\Cref{thm:cc-vs-tw-delta}, in terms of degree and treewidth.

\begin{corollary}[Learning complexity from contraction complexity]
\label{cor:lc-upper-bound}
Let $G$ be a connected graph with $n\geq2$ vertices, and let $C=\CC(G)$ be its
contraction complexity.  Then
\begin{equation}
    \lc_{d,\chi}(G)
    \leq
    \min\left\{
        n,
        3\max\{1,\lceil C\log_d\chi\rceil\}
    \right\}
    \leq
    3\Delta(G)\left(\tw(G)+1\right)\max\{1,\lceil\log_d\chi\rceil\}.
\end{equation}
\end{corollary}

\begin{proof}
Let $\pi$ be an optimal contraction sequence, and let $\mathcal L_\pi$ be the
learning sequence induced by \Cref{lem:contraction-to-learning-sequence}.  Every set
of $\mathcal L_\pi$ satisfies $|\cut_G(S_i)|\leq C$, so
\begin{equation}
    \bar q_i
    \leq
    Q
    :=
    \max\{1,\lceil C\log_d\chi\rceil\}
\end{equation}
for every $i$.  Each step of $\mathcal L_\pi$ combines the two endpoints of a
contraction, and each endpoint contributes either one fresh vertex or the residual
register of an earlier set, so
\begin{equation}
    |F_i|+\sum_{j\in I_i}\bar q_j
    \leq
    2Q,
\end{equation}
using $Q\geq1$.  Together with $\bar q_i\leq Q$, this gives
$\lc_{d,\chi}(\mathcal L_\pi)\leq3Q$.  The one-step learning sequence with
$S_1=F_1=V$ and $I_1=\emptyset$ has $\bar q_1=0$ and $\bar a_1=n$, and hence
$\lc_{d,\chi}(G)\leq n$. Combining these two bounds proves the first inequality.  For the second inequality, $\lceil ab\rceil\leq a\lceil b\rceil$
for $a\in\mathbb N$ gives
$Q\leq\max\{1,C\lceil\log_d\chi\rceil\}\leq\max\{1,C\}\max\{1,\lceil\log_d\chi\rceil\}$. Since $G$ is connected and has at least two vertices, $\Delta(G)\geq1$ and $\tw(G)\geq1$. Hence
\begin{equation}
    \max\{1,C\}
    \leq
    \Delta(G)(\tw(G)+1)-1
\end{equation}
by \Cref{thm:cc-vs-tw-delta}.
\end{proof}

Combining this construction with the sequence learner gives sample complexity bounds for TNS tomography in terms of contraction complexity.

\begin{corollary}[Direct tomography from contraction complexity]
\label{cor:lseq-cc-tomography}
Let $G=([n],E)$ be a connected graph with bond dimension bound $\chi$, and let
$C=\CC(G)$ be its contraction complexity. Then an unknown state
$|\psi\rangle\in\mathcal S_d(G,\chi)$ can be learned to trace norm error at most
$\epsilon$ with success probability at least $1-\delta$ using
\begin{equation}
    O\left(
        \frac{n^3}{\epsilon^4}
        \left(
            d^2\chi^{3C}
            +
            \log(n/\delta)
        \right)
    \right)
    \label{eq:lseq-cc-tomography-bound}
\end{equation}
copies, given a contraction sequence of complexity $C$.  If no such sequence is supplied, one may first compute an optimal contraction
sequence using the algorithm described after
\Cref{thm:cc-vs-tw-delta}.
\end{corollary}

\begin{proof}
If $n=1$, the claim is immediate from ordinary tomography on one qudit.  Assume
$n\geq2$.  Apply \Cref{lem:contraction-to-learning-sequence} to a contraction sequence
of complexity $C$, obtaining a learning sequence $\mathcal L$ of length $L=n-1$ with
$|\cut_G(S_i)|\leq C$ for every $i$.  Hence
\begin{equation}
    \bar r_i=\chi^{|\cut_G(S_i)|}\leq\chi^C.
\end{equation}
At a contraction step, the subsystem $M_i$ is built from the two current contraction
endpoints. Each endpoint is either a fresh original vertex, contributing dimension $d$,
or the residual subsystem of a previously created set $S_j$, contributing dimension at
most
\begin{equation}
    d^{\bar q_j}\leq d\bar r_j\leq d\chi^C.
\end{equation}
Since $\chi^C\geq1$, each endpoint contributes dimension at most $d\chi^C$.  Therefore
\begin{equation}
    d^{\bar a_i}\leq d^2\chi^{2C}
\end{equation}
for every $i$, and so
\begin{equation}
    \max_{i\in[L]} \bar r_i d^{\bar a_i}
    \leq
    d^2\chi^{3C}.
\end{equation}
The stated sample complexity follows from
\eqref{eq:learning-sequence-uniform-precise-bound}, using $L\leq n$.
\end{proof}

Since contraction complexity can be bounded in terms of treewidth and degree, these parameters also give sample complexity bounds for TNS tomography.

\begin{corollary}[A treewidth and degree bound]
\label{cor:lseq-tw-delta-tomography}
Let $G=([n],E)$ be a connected graph with bond dimension bound $\chi$, and let
$\Delta=\Delta(G)$ be its maximum degree. Then an unknown state
$|\psi\rangle\in\mathcal S_d(G,\chi)$ can be learned to trace norm error at most $\epsilon$ with success probability at least $1-\delta$ using
\begin{equation}
    O\left(
        \frac{n^3}{\epsilon^4}
        \left(
            d^2\chi^{3(\Delta(\tw(G)+1)-1)}
            +
            \log(n/\delta)
        \right)
    \right)
\end{equation}
copies. A contraction sequence with the stated width bound is computed first. Equivalently,
the exponential dependence is
\begin{equation}
    d^2\chi^{O(\Delta\tw(G))}.
\end{equation}
\end{corollary}

\begin{proof}
This follows from \Cref{cor:lseq-cc-tomography} and the bound
\begin{equation}
    \CC(G)\leq \Delta(G)(\tw(G)+1)-1
\end{equation}
from \Cref{thm:cc-vs-tw-delta}.
\end{proof}

\begin{remark}[Relation to the black-box reductions]
\label{rem:relation-to-black-box}
The learning sequence approach includes the schedules used by the two black-box reductions.
For a linear ordering $\pi$ with $c=\cw(\pi)$, take
    $S_i=\{v_1,\ldots,v_i\},
    F_i=\{v_i\}$,
with $I_1=\emptyset$ and $I_i=\{i-1\}$ for $i\geq2$. Then
\begin{equation}
    \bar r_i\leq\chi^c,
    \qquad
    \bar q_i\leq
    \kappa_c
    :=
    \max\left\{
        1,
        \left\lceil c\log_d\chi\right\rceil
    \right\},
\end{equation}
and hence
\begin{equation}
    \bar a_i\leq1+\kappa_c.
\end{equation}
The precise bound in \Cref{thm:learning-sequence-tomography} then gives the same
dependence
\begin{equation}
    \chi^c d^{\kappa_c+1}
\end{equation}
as \Cref{thm:blackbox-mps-learning}. Likewise, root the tree obtained from a tree-cut decomposition of width $k$, and
let $S_i$ be the union of the bags in a rooted subtree, with $F_i$ equal to the
bag at its root. If $b$ is the maximum number of children, then
\begin{equation}
    \bar r_i\leq\chi^k,
    \qquad
    \bar q_i\leq
    k\max\left\{
        1,
        \left\lceil\log_d\chi\right\rceil
    \right\},
\end{equation}
and
\begin{equation}
    \bar a_i
    \leq
    k\left(
        1+
        b\max\left\{
            1,
            \left\lceil\log_d\chi\right\rceil
        \right\}
    \right).
\end{equation}
This recovers the dependence in \Cref{thm:blackbox-tcw-learning}.

Thus the precise learning sequence bound can reproduce the two black-box estimates
and may improve them for other choices of learning sequence. The bounds obtained
from contraction complexity in
\Cref{cor:lseq-cc-tomography,cor:lseq-tw-delta-tomography}
use a particular learning sequence and replace its individual cuts by a common upper
bound. These estimates can be less sharp than the specialised path or tree bounds.
\end{remark}

\begin{remark}[Relation to the TTN learner]
\label{rem:relation-to-ttn-learner}
When $G$ is itself a tree $T$ and the learning sequence is chosen from a rooted
traversal of subtrees, the sets $S_i$ are the rooted subtrees of $T$.  For the step
associated with a vertex $u$, the fresh set $F_i$ consists of $u$ itself, and the
child set $I_i$ collects the indices of the steps associated with the children of
$u$ in $T$, so that children in the dependency tree correspond exactly to children
in $T$.  Then $M_i=F_i\cup\bigcup_{j\in I_i}R_j$ is the vertex $u$ together with
the residual subsystems left by the subtrees below it.  Thus the learner above
yields the residual-subsystem formulation of \Cref{alg:LearnTTN} as a special
case.  The separate TTN theorem in \Cref{sec:ttn-tomography} keeps the tree degree
explicit and gives the sharper bound for that special case.
\end{remark}

\subsection{Agnostic TNS tomography}\label{subsec:agnostic-tns-tomography}

The previous subsections considered the realisable setting: the input state was promised
to be a pure tensor network state in the relevant class.  We now consider the agnostic
setting. Here, the input is an arbitrary density operator $\rho$ on $(\mathbb C^d)^{\otimes n}$,
and the goal is to output a pure state whose overlap with $\rho$ is nearly as large as
the best overlap achieved by a tensor network state in the target class.  This is the
natural TN analogue of closest-product-state and closest-MPS learning
\cite{bakshi2025learning,lin2025efficientclosestmatrixproduct}.

For a graph $G=([n],E)$ and a bond dimension bound $\chi$, define
\begin{equation}
    \operatorname{OPT}_{G,\chi}(\rho)
    :=
    \sup_{|\phi\rangle\in\mathcal S_d(G,\chi)}
    \langle\phi|\rho|\phi\rangle .
    \label{eq:agnostic-optimum}
\end{equation}
An agnostic learner with accuracy $\epsilon$ must, with high success probability, output a pure state $|\hat\psi\rangle$
such that
\begin{equation}
    \langle\hat\psi|\rho|\hat\psi\rangle
    \geq
    \operatorname{OPT}_{G,\chi}(\rho)-\epsilon .
    \label{eq:agnostic-goal}
\end{equation}
The output need not belong to $\mathcal S_d(G,\chi)$, so we allow the learner to be improper: it may return a state from a larger efficiently described family.
For pure states, trace norm and overlap are related by
\begin{equation}
    \left\|
        |\psi\rangle\langle\psi|
        -
        |\phi\rangle\langle\phi|
    \right\|_1
    =
    2\sqrt{1-|\langle\psi|\phi\rangle|^2}.
    \label{eq:pure-trace-overlap}
\end{equation}
In the agnostic setting, error is measured by the additive loss in overlap in
\eqref{eq:agnostic-goal}, following
\cite{bakshi2025learning,lin2025efficientclosestmatrixproduct}. When the input
is pure, \eqref{eq:pure-trace-overlap} relates this overlap to trace norm error.

The main difference from the realisable setting is that the reduced states of $\rho$ may have full rank. We retain a larger residual subsystem and show that each learned projection approximately preserves overlap with every reference TNS satisfying the relevant Schmidt-rank bound. The rank restriction is imposed
on the reference state (and not on the input).

\subsubsection{Projection estimates for agnostic learning}

We first prove the projection estimate used throughout the agnostic analysis.
Compared with \cite[Lemma B.12]{bakshi2025learning}, it improves the rank
dependence from $r$ to $\sqrt r$ by applying Cauchy--Schwarz to the positive
semidefinite form induced by $\rho$.

\begin{lemma}[Agnostic projection bound]\label{lem:agnostic-projection-bound}
Let $\mathcal H_A$ and $\mathcal H_B$ be finite-dimensional Hilbert spaces, and let
$\rho\geq0$ be a subnormalised state on $\mathcal H_A\otimes\mathcal H_B$, with
$\tr(\rho)\leq1$.  Let $|\phi\rangle\in\mathcal H_A\otimes\mathcal H_B$ be a vector
with $\|\phi\|\leq1$ and Schmidt rank at most $r$ across the bipartition
$A\mid B$.  Let $W\subseteq\mathcal H_A$, let $\Pi_W$ be the orthogonal projector
onto $W$, and set $Q=I_A-\Pi_W$.  If
\begin{equation}
    \left\|
        Q\,\tr_B(\rho)\,Q
    \right\|_\infty
    \leq
    \eta,
    \label{eq:agnostic-projection-discarded-opnorm}
\end{equation}
then
\begin{equation}
    \left|
        \langle\phi|(\Pi_W\otimes I_B)\rho(\Pi_W\otimes I_B)|\phi\rangle
        -
        \langle\phi|\rho|\phi\rangle
    \right|
    \leq
    2\sqrt{r\eta}.
    \label{eq:agnostic-projection-bound}
\end{equation}
More generally, the same conclusion holds with $\Pi_W$ replaced by any
orthogonal projector $\Pi$ on $\mathcal H_A$ whose range contains $W$, i.e., with
$\Pi\otimes I_B$ in place of $\Pi_W\otimes I_B$.
\end{lemma}

\begin{proof}
We first prove the claim for $P=\Pi_W\otimes I_B$.  Put
\begin{equation}
    R=Q\otimes I_B,
    \qquad
    |\zeta\rangle=R|\phi\rangle,
    \qquad
    a=\langle\zeta|\rho|\zeta\rangle .
\end{equation}
We claim that $a\leq r\eta$.  If $|\zeta\rangle=0$, this is immediate.  Otherwise,
let $|\widetilde\zeta\rangle=|\zeta\rangle/\|\zeta\|$, and write a Schmidt decomposition
\begin{equation}
    |\widetilde\zeta\rangle
    =
    \sum_{j=1}^s \sqrt{\lambda_j}\,|a_j\rangle|b_j\rangle,
    \qquad
    s\leq r,
\end{equation}
where the vectors $|a_j\rangle$ lie in the range of $Q$.  Define
\begin{equation}
    M_{ij}=\langle a_i,b_i|\rho|a_j,b_j\rangle .
\end{equation}
The matrix $M$ is positive semidefinite, being the Gram matrix of the vectors
$|a_j,b_j\rangle$ with respect to the positive semidefinite sesquilinear form
$(x,y)\mapsto\langle x|\rho|y\rangle$: for every $v\in\mathbb C^s$, we have
$v^\dagger Mv=\langle w|\rho|w\rangle\geq0$ with
$|w\rangle=\sum_{j=1}^s v_j|a_j,b_j\rangle$.  Hence
$|M_{ij}|\leq\sqrt{M_{ii}M_{jj}}$.  Therefore
\begin{align}
    \langle\widetilde\zeta|\rho|\widetilde\zeta\rangle\leq \left(\sum_{j=1}^s \sqrt{\lambda_j M_{jj}}\right)^2 \leq \sum_{j=1}^s M_{jj} \leq
    \sum_{j=1}^s
    \langle a_j|Q\,\tr_B(\rho)\,Q|a_j\rangle\leq
    s\eta \leq r\eta .
\end{align}
Since $\|\zeta\|\leq\|\phi\|\leq1$, this gives $a\leq r\eta$.

Now use Cauchy--Schwarz for the positive semidefinite sesquilinear form
$(x,y)\mapsto\langle x|\rho|y\rangle$.  Since $P+R=I$,
\begin{align}
    \langle\phi|\rho|\phi\rangle
    -
    \langle\phi|P\rho P|\phi\rangle
    &=
    \langle\phi|\rho R|\phi\rangle
    +
    \langle R\phi|\rho P|\phi\rangle .
\end{align}
Thus
\begin{align}
    \left|
        \langle\phi|\rho|\phi\rangle
        -
        \langle\phi|P\rho P|\phi\rangle
    \right|
    &\leq
    \sqrt{\langle\phi|\rho|\phi\rangle\,a}
    +
    \sqrt{a\,\langle P\phi|\rho|P\phi\rangle}  \\
    &\leq
    2\sqrt a
    \leq
    2\sqrt{r\eta},
\end{align}
because $\tr(\rho)\leq1$ and $\|\phi\|\leq1$.

If $\Pi$ is any projector on $\mathcal H_A$ whose range contains $W$, then
$I_A-\Pi\leq Q$ in PSD order.  Hence
\begin{equation}
    \left\|(I_A-\Pi)\,\tr_B(\rho)\,(I_A-\Pi)\right\|_\infty
    \leq \eta .
\end{equation}
Repeating the same argument as above with $I_A-\Pi$ in place of $Q$ gives the
claimed generalisation.
\end{proof}

The next lemma derives the operator-norm condition in
\Cref{lem:agnostic-projection-bound} from a tomography estimate. It is the
agnostic counterpart of the top-eigenspace result used in the realisable
analysis.

\begin{lemma}[Top subspace from an approximate state]
\label{lem:top-subspace-opnorm}
Let $\sigma\geq0$ be a subnormalised state on a $D$-dimensional Hilbert space, and
let $\hat\sigma\geq0$ satisfy
\begin{equation}
    \|\sigma-\hat\sigma\|_\infty\leq\gamma .
\end{equation}
Let $1\leq \Lambda < D$, let $W$ be the span of the eigenvectors associated with the $\Lambda$ largest eigenvalues of $\hat\sigma$, and let $Q=I-\Pi_W$.  Then
\begin{equation}
    \|Q\sigma Q\|_\infty
    \leq
    \frac{\tr(\sigma)}{\Lambda+1}+2\gamma
    \leq
    \frac{1}{\Lambda+1}+2\gamma .
    \label{eq:top-subspace-opnorm-bound}
\end{equation}
If $\Lambda\geq D$, then $Q=0$ and the left-hand side is zero.
\end{lemma}

\begin{proof}
Assume $\Lambda<D$.  By Weyl's inequality,
\begin{equation}
    \lambda_{\Lambda+1}(\hat\sigma)
    \leq
    \lambda_{\Lambda+1}(\sigma)+\gamma .
\end{equation}
Since $\sigma\geq0$,
\begin{equation}
    \lambda_{\Lambda+1}(\sigma)
    \leq
    \frac{\tr(\sigma)}{\Lambda+1}.
\end{equation}
For every unit vector $|x\rangle\in W^\perp$, the variational characterisation of
$\lambda_{\Lambda+1}(\hat\sigma)$ gives
\begin{equation}
    \langle x|\hat\sigma|x\rangle
    \leq
    \lambda_{\Lambda+1}(\hat\sigma).
\end{equation}
Therefore
\begin{equation}
    \langle x|\sigma|x\rangle
    \leq
    \langle x|\hat\sigma|x\rangle+\|\sigma-\hat\sigma\|_\infty
    \leq
    \frac{\tr(\sigma)}{\Lambda+1}+2\gamma .
\end{equation}
Taking the supremum over unit vectors in $W^\perp$ proves the claim.
\end{proof}

Combining the two lemmas gives the one-step agnostic disentangling estimate used by
the algorithm.

\begin{corollary}[One-step agnostic disentangling]
\label{cor:one-step-agnostic-disentangling}
Let $\rho\geq0$ be a subnormalised state on
$\mathcal H_A\otimes\mathcal H_B$, with $\tr(\rho)\leq1$, and let
$|\phi\rangle$ be a vector of norm at most one and Schmidt rank at most $r$ across
$A\mid B$.  Let
\begin{equation}
    \sigma_A=\tr_B(\rho),
\end{equation}
and suppose $\hat\sigma_A\geq0$ satisfies
\begin{equation}
    \|\sigma_A-\hat\sigma_A\|_1\leq\gamma .
\end{equation}
Let $W$ be the span of the eigenvectors associated with the $\Lambda$ largest eigenvalues of $\hat\sigma_A$.  Let $\Pi$ be any orthogonal projector on $\mathcal H_A$
whose range contains $W$.  Then
\begin{equation}
    \left|
        \langle\phi|(\Pi\otimes I_B)\rho(\Pi\otimes I_B)|\phi\rangle
        -
        \langle\phi|\rho|\phi\rangle
    \right|
    \leq
    2\sqrt{r\left(\frac{1}{\Lambda+1}+2\gamma\right)} .
    \label{eq:one-step-agnostic-loss}
\end{equation}
\end{corollary}

\begin{proof}
The trace-norm estimate implies $\|\sigma_A-\hat\sigma_A\|_\infty\leq\gamma$.  Apply
\Cref{lem:top-subspace-opnorm} to bound the discarded operator norm, and then apply
\Cref{lem:agnostic-projection-bound}.
\end{proof}

\subsubsection{The agnostic sequence learner}

We now adapt the sequence algorithm from \Cref{subsec:direct-graph-tomography}
to the agnostic setting. Let $\mathcal L=(S_i,I_i,F_i)_{i=1}^L$ be a learning sequence
for $G$. Set
\begin{equation}
    r_i=
    \chi^{|\cut_G(S_i)|}.
    \label{eq:agnostic-ri}
\end{equation}
By \Cref{claim:cut-rank-bound}, every state in
$\mathcal S_d(G,\chi)$ has Schmidt rank at most $r_i$ across
$S_i\mid[n]\setminus S_i$. When a weight function $w$ is known and the
reference class is restricted to $\mathcal S_d(G,w)$, one may instead use
\begin{equation}
    r_i=
    \prod_{e\in\cut_G(S_i)}w(e).
\end{equation}
This may give a smaller value of $r_i$.

The agnostic learner retains a larger residual subsystem than the realisable
learner. Fix an overlap lower bound $\vartheta\in(0,1]$, and assume that
\begin{equation}
    \operatorname{OPT}_{G,\chi}(\rho)
    \geq
    \vartheta .
    \label{eq:agnostic-overlap-promise}
\end{equation}
For $\epsilon\in(0,1]$, define
\begin{equation}
    \alpha
    :=
    \frac{\min\{\epsilon,\vartheta\}}{4L},
    \qquad
    \Lambda_i
    :=
    \left\lceil \frac{64 r_i}{\alpha^2}\right\rceil,
    \qquad
    p_i
    :=
    \max\{1,\lceil\log_d\Lambda_i\rceil\},
    \qquad
    \gamma_i
    :=
    \frac{\alpha^2}{128r_i}.
    \label{eq:agnostic-parameters}
\end{equation}
The number $\Lambda_i$ is the dimension of the learned top subspace, $p_i$ is the
number of qudits retained after processing $S_i$, and $\gamma_i$ is the trace-norm
accuracy of the corresponding subnormalised tomography call.  Define also
\begin{equation}
    a_i^{\mathrm{agn}}
    :=
    |F_i|+
    \sum_{j\in I_i}p_j,
    \qquad
    a_{\max}^{\mathrm{agn}}
    :=
    \max_{i\in[L]}a_i^{\mathrm{agn}},
    \qquad
    r_{\max}:=
    \max_{i\in[L]}r_i .
    \label{eq:agnostic-ai}
\end{equation}

\begin{algorithm}
\caption{AgnosticLearnTNSFromSequence}\label{alg:agnostic-learn-tns-from-sequence}
\begin{algorithmic}[1]
\Require graph $G=([n],E)$, bond dimension bound $\chi$, physical dimension
 $d\geq2$, learning sequence $\mathcal L=(S_i,I_i,F_i)_{i=1}^L$, copies of
 an arbitrary state $\rho$, accuracy parameter $\epsilon\in(0,1]$, confidence
 parameter $\delta\in(0,1)$, and lower-overlap parameter $\vartheta\in(0,1]$
\Ensure Classical description of a pure state $|\hat\psi\rangle$.
\State Compute $r_i,\Lambda_i,p_i,\gamma_i$ as in
\Cref{eq:agnostic-ri,eq:agnostic-parameters}.
\State $K_0\gets I$, $Q\gets\emptyset$, and $c\gets0$.
\For{$i=1$ to $L$}
    \State $M_i\gets F_i\cup\bigcup_{j\in I_i}R_j$.
    \If{$|M_i|\leq p_i$}
        \State $R_i\gets M_i$.
    \Else
        \State Choose $R_i\subseteq M_i$ with $|R_i|=p_i$.
        \State $Q_i\gets M_i\setminus R_i$.
        \State
        $
        \hat\sigma_i
        \gets
        \operatorname{Sub\text{-}Tomography}
        \left(
            K_c,
            M_i,
            d^{|M_i|},
            \gamma_i,
            \delta/(2L+2),
            [\vartheta/2,1]
        \right).
        $
        \State Let $W_i$ be the span of the eigenvectors associated with the $\Lambda_i$ largest eigenvalues of $\hat\sigma_i$.
        \State Compute a unitary $\widetilde U_i$ supported on $M_i$ such that
        \begin{equation*}
            \widetilde U_i W_i
            \subseteq
            |0^{|Q_i|}\rangle_{Q_i}\otimes\mathcal H_{R_i}.
        \end{equation*}
        \State Let $U_{c+1}$ be $\widetilde U_i$ tensored with the identity outside
        $M_i$, $U_{c+1}=\widetilde U_i\otimes I_{M_i^c}$.
        \State Let
        \begin{equation*}
            P_{c+1}=|0^{|Q_i|}\rangle\langle0^{|Q_i|}|_{Q_i}\otimes I_{[n]\setminus Q_i}.
        \end{equation*}
        \State $K_{c+1}\gets P_{c+1}U_{c+1}K_c$.
        \State $Q\gets Q\cup Q_i$ and $c\gets c+1$.
    \EndIf
\EndFor
\State $\gamma_{\mathrm{fin}}\gets\epsilon/8$.
\State
$
\hat\tau
\gets
\operatorname{Sub\text{-}Tomography}
\left(
    K_c,
    R_L,
    d^{|R_L|},
    \gamma_{\mathrm{fin}},
    \delta/(2L+2),
    [\vartheta/2,1]
\right).
$
\State Let $|\hat\varphi\rangle$ be a top eigenvector of $\hat\tau$.
\State \Return
\begin{equation*}
    |\hat\psi\rangle
    =
    U_1^\dagger U_2^\dagger\cdots U_c^\dagger
    \left(
        |0^{|Q|}\rangle_Q\otimes |\hat\varphi\rangle_{R_L}
    \right),
\end{equation*}
with the tensor product interpreted according to the canonical ordering of the qudits.
\end{algorithmic}
\end{algorithm}

The unitary in \Cref{alg:agnostic-learn-tns-from-sequence} exists because
$\Lambda_i\leq d^{p_i}=\dim(\mathcal H_{R_i})$.  It is constructed by mapping an
orthonormal basis of $W_i$ isometrically into
$|0^{|Q_i|}\rangle_{Q_i}\otimes\mathcal H_{R_i}$ and extending the isometry to a unitary on
$\mathcal H_{M_i}$.

The residual-subsystem invariant from \Cref{lem:learning-sequence-residual-invariant}
continues to hold with $q_i$ replaced by $p_i$: after step $i$, the only qudits in
$S_i$ that have not been projected onto $|0\rangle$ are those in $R_i$, and
$|R_i|\leq p_i$.  Consequently
\begin{equation}
    |M_i|
    \leq
    a_i^{\mathrm{agn}}
    =
    |F_i|+
    \sum_{j\in I_i}p_j .
    \label{eq:agnostic-Mi-size}
\end{equation}
The proof is identical to the proof of \Cref{lem:learning-sequence-residual-invariant}.

We also require a rank invariant for each reference TNS.

\begin{lemma}[Rank invariant for reference TNSs]
\label{lem:agnostic-comparator-rank}
Let $|\phi\rangle\in\mathcal S_d(G,\chi)$.  Consider
\Cref{alg:agnostic-learn-tns-from-sequence}, and suppose step $i$ is about to be
processed.  Let $K_{<i}$ be the cumulative postselection map constructed so far,
i.e., the map $K_c$ for the value of the counter $c$ at that moment.  (As in
\Cref{lem:learning-sequence-rank-bound}, $K_{<i}\neq K_{i-1}$ in general, since
trivial steps with $|M_j|\leq p_j$ do not increment $c$.)
Then the vector $K_{<i}|\phi\rangle$, viewed across the bipartition
\begin{equation}
    M_i\mid [n]\setminus M_i,
\end{equation}
has Schmidt rank at most $r_i$.
\end{lemma}

\begin{proof}
Since $|\phi\rangle\in\mathcal S_d(G,\chi)$, there exists a weight function
$w_\phi:E\to[\chi]$ such that
$|\phi\rangle\in\mathcal S_d(G,w_\phi)$. By
\Cref{claim:cut-rank-bound}, the Schmidt rank of $|\phi\rangle$ across
$S_i\mid[n]\setminus S_i$ is at most
\begin{equation}
    \prod_{e\in\cut_G(S_i)} w_\phi(e)
    \leq
    \chi^{|\cut_G(S_i)|}   =   r_i.
\end{equation}
By \Cref{lem:learning-sequence-laminarity}, every earlier set $S_j$, with
$j<i$, is either contained in $S_i$ or disjoint from $S_i$. The operation
associated with step $j$
is supported on a subset of $S_j$.  Hence every operation in $K_{<i}$ is local with
respect to the bipartition $S_i\mid[n]\setminus S_i$, and local linear maps cannot
increase Schmidt rank across this bipartition.

Finally, by the residual-subsystem invariant, the qudits in $S_i\setminus M_i$ have
already been projected onto the product state $|0\rangle_{S_i\setminus M_i}$.  Removing
this fixed product factor cannot increase the Schmidt rank.  Therefore
$K_{<i}|\phi\rangle$ has Schmidt rank at most $r_i$ across
$M_i\mid[n]\setminus M_i$.
\end{proof}

\subsubsection{Correctness}

We first prove a deterministic statement conditioned on all tomography calls succeeding.

\begin{proposition}[Correctness of the agnostic sequence learner]
\label{prop:agnostic-learning-sequence-correctness}
Let $\rho$ be an arbitrary $n$-qudit state, and suppose that
\begin{equation}
    \operatorname{OPT}_{G,\chi}(\rho)\geq\vartheta .
\end{equation}
Assume that every call to $\operatorname{Sub\text{-}Tomography}$ in
\Cref{alg:agnostic-learn-tns-from-sequence} succeeds with its stated trace-norm
accuracy.  Then the output state $|\hat\psi\rangle$ satisfies
\begin{equation}
    \langle\hat\psi|\rho|\hat\psi\rangle
    \geq
    \operatorname{OPT}_{G,\chi}(\rho)-\epsilon .
    \label{eq:agnostic-correctness}
\end{equation}
Moreover, each postselection map used by the algorithm has success probability at
least $\vartheta/2$.
\end{proposition}

\begin{proof}
Set
\begin{equation}
    \beta:=L\alpha=\frac{\min\{\epsilon,\vartheta\}}{4}.
\end{equation}
Fix a reference state $|\phi\rangle\in\mathcal S_d(G,\chi)$ satisfying
\begin{equation}
    \langle\phi|\rho|\phi\rangle
    \geq
    \operatorname{OPT}_{G,\chi}(\rho)-\beta .
    \label{eq:near-optimal-comparator}
\end{equation}
Such a reference state exists by the definition of the supremum in
\eqref{eq:agnostic-optimum}.

Index the nontrivial disentangling steps by $t=1,\ldots,c$.  Let $K_t$ be the
cumulative postselection map after the first $t$ such steps, and define
\begin{equation}
    B_t
    :=
    \langle\phi|K_t^\dagger K_t\,\rho\,K_t^\dagger K_t|\phi\rangle,
    \qquad
    B_0=\langle\phi|\rho|\phi\rangle .
    \label{eq:agnostic-Bt}
\end{equation}
We now prove that
\begin{equation}
    B_t\geq B_0-t\alpha
    \label{eq:agnostic-Bt-induction}
\end{equation}
for all $t$.

Suppose the next nontrivial disentangling step processes the index $i$.  Put
\begin{equation}
    \rho_t=K_t\rho K_t^\dagger,
    \qquad
    |\phi_t\rangle=K_t|\phi\rangle .
\end{equation}
The reduced state learned by the algorithm is
\begin{equation}
    \sigma_i=\tr_{[n]\setminus M_i}(\rho_t).
\end{equation}
By assumption, the tomography estimate satisfies
\begin{equation}
    \|\sigma_i-\hat\sigma_i\|_1\leq\gamma_i .
\end{equation}
By \Cref{lem:agnostic-comparator-rank}, the vector $|\phi_t\rangle$ has Schmidt rank
at most $r_i$ across the bipartition $M_i\mid[n]\setminus M_i$.  In the coordinates before applying $\widetilde U_i$, the projection used at the current step is
\begin{equation}
    \Pi_i
    :=
    \widetilde U_i^\dagger
    \left(|0^{|Q_i|}\rangle\langle0^{|Q_i|}|_{Q_i}\otimes I_{R_i}\right)
    \widetilde U_i
\end{equation}
on $M_i$, tensored with the identity outside $M_i$.  Its range on $M_i$
contains $W_i$ by definition of $\widetilde{U}_i$.  \Cref{cor:one-step-agnostic-disentangling} now gives
\begin{equation}
    |B_{t+1}-B_t|
    \leq
    2\sqrt{r_i\left(\frac{1}{\Lambda_i+1}+2\gamma_i\right)} .
\end{equation}
Using the parameter choices in \Cref{eq:agnostic-parameters},
\begin{equation}
    \frac{1}{\Lambda_i+1}\leq \frac{\alpha^2}{64r_i},
    \qquad
    2\gamma_i=\frac{\alpha^2}{64r_i},
\end{equation}
and hence
\begin{equation}
    |B_{t+1}-B_t|
    \leq
    2\sqrt{\frac{\alpha^2}{32}}
    <
    \alpha .
\end{equation}
This proves \eqref{eq:agnostic-Bt-induction} by induction.  Since $c\leq L$,
\begin{equation}
    B_c
    \geq
    B_0-L\alpha
    =
    B_0-\beta .
    \label{eq:Bc-lower-bound}
\end{equation}
The same bound also implies the claimed lower bound on success probability.  Indeed,
\begin{equation}
    \tr(K_t\rho K_t^\dagger)
    \geq
    B_t
    \geq
    B_0-L\alpha
    \geq
    \operatorname{OPT}_{G,\chi}(\rho)-2\beta
    \geq
    \vartheta/2,
\end{equation}
for every postselection map encountered by the algorithm.  This justifies the success probability interval $[\vartheta/2,1]$ used in the tomography calls.  We note that this is the only place in the correctness argument where $\vartheta$ enters: the last inequality requires $2\beta\leq\vartheta/2$, and this is precisely why the parameter $\alpha=\min\{\epsilon,\vartheta\}/(4L)$ carries the minimum with $\vartheta$, whereas the accuracy guarantee below uses only $\beta\leq\epsilon/4$.

It remains to analyse the final residual step.  As in the realisable learners, the
residual-subsystem invariant implies that
\begin{equation}
    K_c=P_QU_{\leq c},
    \qquad
    U_{\leq c}=U_cU_{c-1}\cdots U_1,
    \qquad
    P_Q=|0^{|Q|}\rangle\langle0^{|Q|}|_Q\otimes I_{R_L},
\end{equation}
where $Q=\bigsqcup_{i:\,|M_i|>p_i}Q_i$ is the final value of the projected set
maintained by \Cref{alg:agnostic-learn-tns-from-sequence}.
Let
\begin{equation}
    \tau
    =
    \tr_{[n]\setminus R_L}
    \left[
        K_c\rho K_c^\dagger
    \right]
\end{equation}
be the ideal final residual state. Since $P_Q$ is an orthogonal projector and $U_{\leq c}$ is unitary,
\begin{equation}
    K_c^\dagger K_c
    =
    U_{\leq c}^\dagger P_QU_{\leq c}
\end{equation}
is an orthogonal projector. As $P_Q$ projects the qubits in $[n]\setminus R_L$ into the all-zero state, $K_c\rho K_c^\dagger=\ket{0^{|Q|}}\bra{0^{|Q|}}_{Q}\otimes \tau_{R_L}$, so the largest eigenvalue of $\tau$ equals the largest eigenvalue of $K_c\rho K_c^\dagger$, which is the maximum of
$\langle\xi|\rho|\xi\rangle$ over normalised states $|\xi\rangle$ in the range of the projection
$K_c^\dagger K_c$. Moreover,
\begin{equation}
    B_c
    =
    \langle\phi|K_c^\dagger K_c\rho K_c^\dagger K_c|\phi\rangle.
\end{equation}
Normalising the subnormalised vector $K_c^\dagger K_c|\phi\rangle$, if it is nonzero, shows that
\begin{equation}
    \lambda_{\max}(\tau)
    \geq B_c / \|K_c^\dagger K_c|\phi\rangle\|^2
    \geq B_c .
    \label{eq:lambda-max-final-lower-bound}
\end{equation}
The final tomography call gives $\|\tau-\hat\tau\|_1\leq\gamma_{\mathrm{fin}}$, where
$\gamma_{\mathrm{fin}}=\epsilon/8$.  If $|\hat\varphi\rangle$ is a top eigenvector of
$\hat\tau$, then the variational principle gives
\begin{equation}
    \langle\hat\varphi|\tau|\hat\varphi\rangle
    \geq
    \lambda_{\max}(\tau)-2\gamma_{\mathrm{fin}} .
    \label{eq:final-top-eigenvector-loss}
\end{equation}
The state returned by the algorithm is
\begin{equation}
    |\hat\psi\rangle
    =
    U_{\leq c}^\dagger\left(|0^{|Q|}\rangle_Q\otimes|\hat\varphi\rangle_{R_L}\right),
\end{equation}
and 
\begin{equation}\label{eq:final-top-eigenvector-loss-rewritten}
    \langle\hat\psi|\rho|\hat\psi\rangle=\langle\hat\varphi|\tau|\hat\varphi\rangle .
\end{equation}
Combining Equations \eqref{eq:near-optimal-comparator}, \eqref{eq:Bc-lower-bound},
\eqref{eq:lambda-max-final-lower-bound}, \eqref{eq:final-top-eigenvector-loss}, and \eqref{eq:final-top-eigenvector-loss-rewritten}, we
obtain
\begin{equation}
    \langle\hat\psi|\rho|\hat\psi\rangle
    \geq
    \operatorname{OPT}_{G,\chi}(\rho)
    -
    2\beta
    -
    2\gamma_{\mathrm{fin}}
    \geq
    \operatorname{OPT}_{G,\chi}(\rho)-\epsilon,
\end{equation}
because $2\beta\leq\epsilon/2$ and $2\gamma_{\mathrm{fin}}=\epsilon/4$.
This proves the proposition.
\end{proof}

\begin{theorem}[Agnostic tomography from a learning sequence]
\label{thm:agnostic-learning-sequence-tomography}
Let $G=([n],E)$ be a graph with bond dimension bound $\chi$, and let
$\mathcal L=(S_i,I_i,F_i)_{i=1}^L$ be a learning sequence for $G$.  Let $\rho$ be an arbitrary $n$-qudit state satisfying
\begin{equation}
    \operatorname{OPT}_{G,\chi}(\rho)\geq\vartheta
\end{equation}
for a known $\vartheta\in(0,1]$.  Then
\Cref{alg:agnostic-learn-tns-from-sequence} outputs a pure state $|\hat\psi\rangle$
satisfying
\begin{equation}
    \langle\hat\psi|\rho|\hat\psi\rangle
    \geq
    \operatorname{OPT}_{G,\chi}(\rho)-\epsilon
\end{equation}
with probability at least $1-\delta$, using
\begin{equation}
    O\left(
        \frac{L^5 r_{\max}^2}{\vartheta\,\min\{\epsilon,\vartheta\}^4}
        \left(
            d^{2a_{\max}^{\mathrm{agn}}}
            +
            \log(L/\delta)
        \right)
    \right)
    \label{eq:agnostic-learning-sequence-sample-bound}
\end{equation}
copies, where $r_{\max}$ and $a_{\max}^{\mathrm{agn}}$ are defined in
\Cref{eq:agnostic-ai}.  The classical postprocessing time is polynomial in
$L$, $d^{a_{\max}^{\mathrm{agn}}}$, $\max_i\Lambda_i$, $1/\epsilon$,
$1/\vartheta$, and $\log(1/\delta)$, assuming the full-rank tomography primitive and
the relevant eigendecompositions are implemented in time polynomial in their input
dimension.
\end{theorem}

\begin{proof}
The lower bound on success probability used in each call is valid inductively: conditioned on all previous tomography calls succeeding, the proof of
\Cref{prop:agnostic-learning-sequence-correctness} up to the current step gives
$\tr(K_c\rho K_c^\dagger)\geq\vartheta/2$.  Thus each subnormalised tomography
call has the advertised guarantee.  There are at most $L$ nontrivial disentangling
calls and one final residual call.  Each call is assigned failure probability
$\delta/(2L+2)$, so a union bound gives total failure probability at most $\delta$.
Conditioned on this event, \Cref{prop:agnostic-learning-sequence-correctness} gives the
claimed overlap guarantee.

It remains to count copies.  At a disentangling step $i$, the state being learned is an
arbitrary subnormalised state on $|M_i|$ qudits.  This allows us to invoke the full-rank case of the subnormalised tomography primitive from
\Cref{subsec:subnormalised-tomography}, with rank parameter $d^{|M_i|}$ and Hilbert
space dimension $d^{|M_i|}$.  Since the success probability is lower bounded by
$\vartheta/2$, \Cref{cor:subnormalised-tomography-bounded-mu} gives copy complexity
\begin{equation}
    O\left(
        \frac{d^{2|M_i|}+\log(L/\delta)}{\vartheta\gamma_i^2}
    \right)
\end{equation}
for this call.  By \eqref{eq:agnostic-Mi-size}, $|M_i|\leq a_{\max}^{\mathrm{agn}}$, and by
\eqref{eq:agnostic-parameters},
\begin{equation}
    \frac{1}{\gamma_i^2}
    =
    \frac{(128r_i)^2}{\alpha^4}
    \leq
    O\left(
        \frac{r_{\max}^2L^4}{\min\{\epsilon,\vartheta\}^4}
    \right).
\end{equation}
Multiplying by at most $L$ disentangling calls gives the stated bound.  The final
residual tomography call acts on at most $a_{\max}^{\mathrm{agn}}$ qudits and has accuracy
$\epsilon/8$.  Its cost is absorbed into \eqref{eq:agnostic-learning-sequence-sample-bound}
for $L\geq1$, $r_{\max}\geq1$, and $\epsilon,\vartheta\leq1$.  The runtime statement
follows from the dimensions of the matrices on which tomography and eigendecomposition
are performed.
\end{proof}

\begin{remark}
\label{rem:unknown-overlap-parameter}
The theorem assumes a known lower bound $\vartheta\leq\operatorname{OPT}_{G,\chi}(\rho)$.  This
parameter is used only to lower-bound the success probabilities of the postselected
branches.  If no such bound is known, one can run the algorithm over a geometric grid
of candidate values for $\vartheta$ and validate the final candidates by estimating
$\langle\hat\psi|\rho|\hat\psi\rangle$.  This adds only logarithmic overhead in the grid
size, but we keep $\vartheta$ explicit in the theorem statements.
\end{remark}

\subsubsection{Applications}

We first specialise the agnostic theorem to learning sequences obtained from contraction
sequences.

\begin{corollary}[Agnostic tomography from contraction complexity]
\label{cor:agnostic-contraction-complexity}
Let $G=([n],E)$ be a connected graph with bond dimension bound $\chi$, and let
$C=\CC(G)$ be its contraction complexity.  Suppose
$\rho$ is an arbitrary $n$-qudit state satisfying
$\operatorname{OPT}_{G,\chi}(\rho)\geq\vartheta$.  Given a contraction sequence of complexity $C$,
one can output $|\hat\psi\rangle$ such that
\begin{equation}
    \langle\hat\psi|\rho|\hat\psi\rangle
    \geq
    \operatorname{OPT}_{G,\chi}(\rho)-\epsilon
\end{equation}
with probability at least $1-\delta$, using
\begin{equation}
    O\left(
        \frac{n^5\chi^{2C}}{\vartheta\,\min\{\epsilon,\vartheta\}^4}
        \left(
            d^{4p_C}
            +
            \log(n/\delta)
        \right)
    \right)
    \label{eq:agnostic-cc-bound}
\end{equation}
copies, where
\begin{equation}
    p_C
    :=
    \max\left\{
        1,
        \left\lceil
            \log_d\left(
                \frac{2048 n^2\chi^C}{\min\{\epsilon,\vartheta\}^2}
            \right)
        \right\rceil
    \right\}.
    \label{eq:agnostic-pC}
\end{equation}
Using \Cref{thm:cc-vs-tw-delta}, one may take
\begin{equation}
    C\leq \Delta(G)(\tw(G)+1)-1,
\end{equation}
after computing a contraction sequence with this width bound.
\end{corollary}

\begin{proof}
By \Cref{lem:contraction-to-learning-sequence}, a contraction sequence of complexity
$C$ induces a learning sequence of length $L\leq n-1$ satisfying
$|\cut_G(S_i)|\leq C$ for every $i$.  Hence $r_i\leq\chi^C$ and
$r_{\max}\leq\chi^C$.  Since
\begin{equation}
    \alpha=\frac{\min\{\epsilon,\vartheta\}}{4L}
    \geq
    \frac{\min\{\epsilon,\vartheta\}}{4n},
\end{equation}
we have
\begin{equation}
    \Lambda_i
    =
    \left\lceil\frac{64r_i}{\alpha^2}\right\rceil
    \leq
    \frac{1024\,r_i n^2}{\min\{\epsilon,\vartheta\}^2}+1
    \leq
    \frac{2048\,r_i n^2}{\min\{\epsilon,\vartheta\}^2},
\end{equation}
and hence, using $r_i\leq\chi^C$, we get $p_i\leq p_C$ for all $i$.  Each contraction step combines two current
vertices.  Each current vertex contributes either one fresh physical qudit or one
residual subsystem of size at most $p_C$.  Since $p_C\geq1$, every step has
\begin{equation}
    a_i^{\mathrm{agn}}\leq 2p_C .
\end{equation}
Substituting $L\leq n$, $r_{\max}\leq\chi^C$, and
$a_{\max}^{\mathrm{agn}}\leq2p_C$ into
\Cref{thm:agnostic-learning-sequence-tomography} proves
\eqref{eq:agnostic-cc-bound}.  The statement in terms of treewidth and degree follows from
\Cref{thm:cc-vs-tw-delta}.
\end{proof}

Agnostic TTN state tomography is obtained as a special case by using the learning sequence built from rooted subtrees that is implicit in
\Cref{sec:ttn-tomography}.

\begin{corollary}[Agnostic TTN tomography]
\label{cor:agnostic-ttn-tomography}
Let $T$ be a tree on $n$ vertices with maximum degree $\Delta=\Delta(T)$, and let
$\rho$ be an arbitrary $n$-qudit state satisfying
\begin{equation}
    \operatorname{OPT}_{T,\chi}(\rho)
    :=
    \sup_{|\phi\rangle\in\mathcal S_d(T,\chi)}
    \langle\phi|\rho|\phi\rangle
    \geq
    \vartheta .
\end{equation}
Then one can output $|\hat\psi\rangle$ satisfying
\begin{equation}
    \langle\hat\psi|\rho|\hat\psi\rangle
    \geq
    \operatorname{OPT}_{T,\chi}(\rho)-\epsilon
\end{equation}
with probability at least $1-\delta$, using
\begin{equation}
    O\left(
        \frac{n^5\chi^2}{\vartheta\,\min\{\epsilon,\vartheta\}^4}
        \left(
            d^{2(1+b p_T)}
            +
            \log(n/\delta)
        \right)
    \right)
    \label{eq:agnostic-ttn-bound}
\end{equation}
copies, where $b$ is the maximum number of children after rooting $T$ at a leaf and
\begin{equation}
    p_T
    :=
    \max\left\{
        1,
        \left\lceil
            \log_d\left(
                \frac{2048 n^2\chi}{\min\{\epsilon,\vartheta\}^2}
            \right)
        \right\rceil
    \right\}.
\end{equation}
For $n\geq3$, one may take $b\leq\Delta-1$.
\end{corollary}

\begin{proof}
Use the learning sequence whose sets are the rooted subtrees of $T$, ordered from
leaves to root.  Every set in this sequence other than the final one is a rooted
subtree $T_u$ with $u\neq r$ and is separated from its complement by
the single edge joining $u$ to its parent, so $r_i\leq\chi$ for these sets.  The
final set is all of $V$, whose cut is empty, so $r_L=1$.  The length is at most $n$.
Since $\alpha\geq\min\{\epsilon,\vartheta\}/(4n)$ and $r_i\leq\chi$, the estimate
from the proof of \Cref{cor:agnostic-contraction-complexity} gives
$\Lambda_i\leq2048\,\chi n^2/\min\{\epsilon,\vartheta\}^2$ and hence
$p_i\leq p_T$ for all $i$.  At any
step, the learner acts on the parent vertex together with the residual subsystems of its
children, so
\begin{equation}
    a_i^{\mathrm{agn}}
    \leq
    1+b p_T .
\end{equation}
Substituting $L\leq n$, $r_{\max}\leq\chi$, and
$a_{\max}^{\mathrm{agn}}\leq1+bp_T$ into
\Cref{thm:agnostic-learning-sequence-tomography} proves the stated bound.  The degree
statement follows from the rooting convention used in \Cref{alg:LearnTTN}.
\end{proof}

\begin{remark}[Comparison with closest product state and closest MPS learning]
\label{rem:agnostic-comparison}
Bakshi et al.~\cite{bakshi2025learning} give a proper
closest-product-state learner and, in Appendix~B, an improper closest-MPS
learner. Their MPS guarantee uses MPSs of a given bond dimension as the
reference class, while the output may have larger bond dimension. Lin, Chia,
and Hung~\cite{lin2025efficientclosestmatrixproduct} also give an improper
closest-MPS learner, with better dependence on the system size and logarithmic
reconstruction depth.

\Cref{thm:agnostic-learning-sequence-tomography} allows an arbitrary known graph,
while \Cref{cor:agnostic-ttn-tomography} gives the tree case. Our result extends
the available graph families from edgeless graphs and paths to trees and arbitrary
known graphs.

The closest-product-state learner of \cite{bakshi2025learning} is proper. Their MPS
extension, the MPS learner of \cite{lin2025efficientclosestmatrixproduct}, and our
general-graph learner may return states outside the corresponding reference class
and are improper. All these results use an overlap guarantee of the form
\eqref{eq:agnostic-goal}.

The path case gives a different tradeoff from
\cite{lin2025efficientclosestmatrixproduct}. Specialising
\Cref{cor:agnostic-ttn-tomography} to a path gives at most
\begin{equation}
    \widetilde O\left(
        \frac{d^4\chi^4 n^9}{\vartheta\,\min\{\epsilon,\vartheta\}^8}
    \right)
\end{equation}
copies and a reconstruction circuit of depth linear in $n$. The closest-MPS learner of
\cite{lin2025efficientclosestmatrixproduct} uses
$\widetilde O(n^7D^{12}/\epsilon^{12})$ copies and has logarithmic depth, improving the dependence on $n$ at the cost of a worse dependence on the bond dimension and accuracy. Thus our path specialisation does not improve the system-size dependence or circuit depth, but neither bound uniformly dominates the other. The difference has three sources: the retained dimension
$\Lambda_i$ is of order $r_i/\alpha^2$ rather than $r_i$ because no rank is
assumed on the input, the sequential schedule forces
$\alpha=\min\{\epsilon,\vartheta\}/(4L)$ with $L\leq n$, and every tomography
call runs at the full-rank rate $d^{2|M_i|}$.

Compared with \cite{bakshi2025learning}, our result covers a broader family of
reference TNSs and uses a sharper projection estimate.
\Cref{lem:agnostic-projection-bound} improves
\cite[Lemma B.12]{bakshi2025learning} from $2r\sqrt\eta$ to $2\sqrt{r\eta}$. The two estimates agree for product states, where $r=1$. For general graphs, the sharper estimate saves a factor of $r_i$ in the retained dimension. The weaker estimate would require $\Lambda_i$ of order $r_i^2/\alpha^2$, doubling the exponent of $\chi$ arising from $|\cut_G(S_i)|$ in every residual register.
\end{remark}

\newpage
\section*{Acknowledgements}
We thank V{\'a}clav Bla{\v z}ej for collating information on graph parameters through the HOPS project~\cite{hops}, and V{\'a}clav Bla{\v z}ej and Ramanujan Sridharan for helpful discussions of parameterised complexity.
We also thank Nikhil Bansil, Neha Rino, Jedrzej Olkowski, Peter Strulo, Rafid Ameer Mahmud and Hirak Ghosh for helpful discussions.
NM acknowledges support from the EPSRC DTP 2224 University of Warwick (Grant number EP/W524645/1).
SS acknowledges support from the Wellcome Leap as part of the Q4Bio Program and the Royal Society University Research Fellowship. 
No AI tools were used in writing this paper. 

\newpage
\setcounter{secnumdepth}{0}
\defbibheading{head}{\section{References}}
\sloppy
\printbibliography[heading=head]

\end{document}